\documentclass{article}
\usepackage{graphicx} 
\usepackage{makecell}
\usepackage{fullpage}
\usepackage{amssymb}
\usepackage{amsmath}
\usepackage{amsthm}
\usepackage{multirow}
\usepackage{enumerate}
\usepackage{graphicx}
\usepackage{mathrsfs}
\usepackage[utf8]{inputenc}

\usepackage{float}
\usepackage{pdfpages}
\usepackage{listings}

\usepackage[T1]{fontenc}  
\usepackage[inline]{enumitem}  
\usepackage{lmodern}  
\usepackage{braket}

\makeatletter%
\newsavebox{\Exbox}\newlength{\len}%
{\end{itemize}}  
\makeatother%

\setlist[description]{font=\normalfont\itshape\space}
\newlength{\jeroenlen}
\newenvironment{example}
{\settowidth{\jeroenlen}{\textit{Remarks.}}%
	\begin{description}[leftmargin=\jeroenlen,labelwidth=0pt,labelsep=0pt]
		\item[\textit{Remarks.\qquad}]%
		\begin{enumerate}[leftmargin=0em,labelsep=0.5em]}
		{\end{enumerate}\end{description}}

\usepackage{color}
\usepackage[pdfstartview=FitH,pdfpagemode=UseNone,colorlinks=true,citecolor=blue,linkcolor=blue,backref=page]{hyperref}

\usepackage{cleveref}
\usepackage{mathtools}
\usepackage{relsize}
\usepackage{amsfonts}
\usepackage[T1]{fontenc}
\usepackage{mathpazo}
\usepackage{bm}
\usepackage{changepage}
\usepackage{tcolorbox}
\usepackage{tikz}
\usepackage{thmtools}
\usepackage{thm-restate}
\usepackage{subcaption}
\usepackage{enumitem}
\usepackage{algorithm}
\usepackage{algcompatible}

\usepackage{dirtytalk}
\usepackage[thinc]{esdiff}
\usepackage[numbers]{natbib}
\usepackage{bbm}

\usepackage{tikz}
\usetikzlibrary{positioning}
\usetikzlibrary{calc}
\usetikzlibrary{decorations.shapes}
\usetikzlibrary{decorations.pathreplacing,angles,quotes}
\usepackage{todonotes}
\usepackage{BOONDOX-cal} 
\usepackage{textcomp}
\usepackage{blkarray}

\allowdisplaybreaks

\definecolor{blueviolet}{RGB}{60,50,200}
\definecolor{oliveg}{RGB}{40,200,30}
\hypersetup{colorlinks=true,       
    linkcolor=blueviolet,
    citecolor=oliveg,
}

\newtheorem{theorem}{Theorem}
\theoremstyle{definition}
\newtheorem{definition}{Definition}[section]
\newtheorem{problem}{Problem}[section]
\newtheorem{proposition}{Proposition}[section] 
\newtheorem{lemma}{Lemma}[section]

\newtheorem{corollary}{Corollary}[section]
\newtheorem{claim}{Claim}[section]

\newtheorem{observation}{Observation}[section]



\algrenewcommand\alglinenumber[1]{\footnotesize #1.}

\newcommand{\veca}{{\mathbf{a}}}
\newcommand{\vecb}{{\mathbf{b}}}
\newcommand{\vecc}{{\mathbf{c}}}

\newcommand{\vece}{{\mathbf{e}}}

\newcommand{\vecg}{{\mathbf{g}}}
\newcommand{\vech}{{\mathbf{h}}}

\newcommand{\vecs}{{\mathbf{s}}}

\newcommand{\vecu}{{\mathbf{u}}}
\newcommand{\vecv}{{\mathbf{v}}}
\newcommand{\vecw}{{\mathbf{w}}}
\newcommand{\vecx}{{\mathbf{x}}}
\newcommand{\vecy}{{\mathbf{y}}}
\newcommand{\vecz}{{\mathbf{z}}}
\newcommand{\veczero}{{\mathbf{0}}}
\newcommand{\vecone}{{\mathbf{1}}}

\newcommand{\B}{\mathbb{B}}
\newcommand{\C}{\mathbb{C}}
\newcommand{\F}{\mathbb{F}}

\newcommand{\N}{\mathbb{N}}
\newcommand{\Q}{\mathbb{Q}}
\newcommand{\R}{\mathbb{R}}

\newcommand{\Z}{\mathbb{Z}}

\newcommand{\cF}{\mathcal{F}}

\newcommand{\NP}{\mathsf{NP}}

\newcommand{\GL}{\mathrm{GL}}

\makeatletter
\newcommand\footnoteref[1]{\protected@xdef\@thefnmark{\ref{#1}}\@footnotemark}
\makeatother

\definecolor{DSgray}{cmyk}{0,0,0,0.7}

\definecolor{anti-flashwhite}{rgb}{0.95, 0.95, 0.96}
\definecolor{britishracinggreen}{rgb}{0.0, 0.26, 0.15}
\definecolor{cornellred}{rgb}{0.7, 0.11, 0.11}
\definecolor{cream}{rgb}{1.0, 0.99, 0.82}
\definecolor{lavender}{rgb}{0.9, 0.9, 0.98}
\definecolor{midnightblue}{rgb}{0.1, 0.1, 0.44}	
\definecolor{deepsaffron}{rgb}{1.0, 0.6, 0.2}
\definecolor{piggypink}{rgb}{0.99, 0.87, 0.9}
\definecolor{dollarbill}{rgb}{0.52, 0.73, 0.4}
\definecolor{crimsonred}{rgb}{0.6, 0.0, 0.0}
\definecolor{persiangreen}{rgb}{0, 0.5, 0.0}
\definecolor{celadon}{rgb}{0.87, 0.99, 0.8}
\definecolor{dukeblue}{rgb}{0.53, 0.0, 0.69}
\definecolor{sealbrown}{rgb}{0.44, 0.28, 0.1}
\definecolor{wheat}{rgb}{0.94, 0.92, 0.84}	
\definecolor{patriarch}{rgb}{0.03, 0.57, 0.82}
\definecolor{cerise}{rgb}{0.96, 0.0, 0.63}
\definecolor{scarlet}{rgb}{0.96, 0.0, 0.63}
\definecolor{almond}{rgb}{0.94, 0.87, 0.8}
\definecolor{aliceblue}{rgb}{0.94, 0.97, 1.0}

\title{Testing Equivalence to the Hamiltonian Cycle Polynomial}
\author{{Agrim Dewan}\\
		\normalsize{Indian Institute of Science}\\
	    \normalsize{\tt{agrimdewan@iisc.ac.in}}}
\date{}
\begin{document}
\maketitle
\begin{abstract}
    The Hamiltonian Cycle polynomial, denoted as $HC_n$, is defined to be the sum of the weighted Hamiltonian Cycles in an $n$-vertex complete digraph, with vertices labeled $1$ to $n$ and edges weighted by formal variables $x_{i,j}$. The Permanent and $HC$, defined as the family $\{HC_n | \ n \geq 1\}$, were studied by Valiant (STOC 1979), with the former shown to be VNP-complete over all fields of characteristic other than $2$, and the latter to be VNP-complete over \emph{every} field. Since its introduction, $HC$ has been studied from the perspective of circuit lower bounds by Jerrum-Snir (JACM 1982), determinantal complexity by Huttenhain-Ikenmeyer (LAA 2016), and its connection with the Permanent and the Determinant polynomials by Goulden-Jackson (EJC 1981) and Grochow (ToC 2017). It has been the most prominent choice for generalising results to all fields, such as in Malod (CCC 2007) and Grochow-Mulmuley-Qiao (ICALP 2016), owing to its VNP-completeness over every field. Hrubes (ToCT, 2016) showed the VNP-completeness of many graph-based polynomial families over every field by using $HC$.
    
    In Kayal (STOC 2012), a randomised polynomial time algorithm was given for the following problem: Given an $n^2$-variate degree-$n$ polynomial $f(\vecx) \in \F[\vecx]$ as a black box, decide if there exists $A \in \GL_{n^2}(\F)$ such that $f(\vecx) = Perm_n(A\vecx)$. Here, the Permanent polynomial $Perm_n$  computes the permanent of the $n \times n$ symbolic matrix $(x_{i,j})$. This problem is known as testing equivalence to the Permanent, or alternatively, ET for Permanent. 
    
    In this work, we study ET for $HC$. While both families are VNP-complete, the efficient ET algorithm for Permanent does not imply the same for $HC$.  Besides, there are crucial differences between the two polynomials that make studying the complexity of ET for $HC$ interesting: The underlying decision problem corresponding to the Permanent is in P (detecting perfect matchings in a bipartite graph), but that for $HC$ (detecting  Hamiltonian cycles in a digraph) is NP-complete. The Permanent polynomial is known to be characterised by its symmetries as shown by Mulmuley-Sohoni (SIAM J. Computing, 2001). This property yields an \emph{efficient} algorithm for the circuit-testing problem for the Permanent, a special case of ET for the Permanent, in which we check whether a given circuit computes the Permanent. In contrast, we show $HC_n$ is \emph{not} characterised by its symmetries. 
    
    In this work, we give a randomised polynomial time ET algorithm for $HC$ with mild constraints on the underlying field. The algorithm is obtained by studying and completely characterising the Lie algebra and the symmetries of $HC_n$. We show that, like the Permanent polynomial, the symmetries of $HC_n$ are generated by permutation and scaling matrices over large enough fields. However, we also show that, unlike the Permanent polynomial, $HC_n$ is \emph{not} characterised by its symmetries. Nevertheless, like the Permanent polynomial, $HC_n$ is downward self-reducible, as shown in Zhang-Bai (TCS 2011), which implies $HC_n$ is characterised by circuit identities and that we can efficiently test whether a given circuit $\mathrm{C}$ computes $HC_n$. We also get a Flip theorem for $HC_n$ as a result of its circuit identities.
\end{abstract}
\clearpage
\newpage
\setcounter{tocdepth}{2}
\tableofcontents
\thispagestyle{empty}
\newpage
\setcounter{page}{1}

\section{Introduction} \label{Section: Introduction}
Two $n$-variate polynomials $f,g \in \F[\vecy]$ are \emph{equivalent}, denoted by $f \sim g$, if $f(\vecy) = g(A\vecy + \vecb)$ for some invertible $A \in \GL_{n}(\F)$ and $\vecb \in \F^{n}$. Clearly, the relation $f \sim g$ is an equivalence relation. The Polynomial Equivalence problem (PE) involves deciding for two given polynomials $f$ and $g$, if $f \sim g$ holds or not. Here, the inputs are assumed to be given as lists of coefficients. On one hand, PE is known to be at least as hard as Graph Isomorphism \cite{AgrawalS05,Kayal11}. On the other hand, the known upper bound on its complexity depends on the underlying field. Over finite fields, PE is in NP $\cap$ co-AM and hence unlikely to be NP-complete, while over $\Q$ we do not even know if it is decidable \cite{Saxenaphd,Thierauf98}. Even when both inputs are restricted, limited results are known, see Appendix \ref{sec-PEres} for a detailed discussion. Thus, understanding the complexity of PE remains a challenge. This has led to a natural variant of PE, called \emph{equivalence testing}, being studied in the literature to make progress towards addressing this challenge.

\paragraph{Equivalence testing.} PE asks to check the equivalence of two arbitrary polynomials. In contrast, for Equivalence Testing (ET), we first fix a polynomial family $\mathcal{F}$ or a circuit\footnote{By a circuit we mean an \emph{arithmetic circuit}, unless stated otherwise. An arithmetic circuit is like a Boolean circuit except it uses $+$ or $\times$ gates in place of AND, OR, and NOT gates, and the edges are labeled by $\F$-elements. The circuit computes a polynomial over $\F$.} class $\mathcal{C}$, the goal then is to decide for a given \emph{single} polynomial $f$, whether $f \sim g$, for some $g$ in $\mathcal{F}$ or in $ \mathcal{C}$, respectively. The two versions are called ET for polynomial families and ET for circuit classes, respectively. The study of ET for polynomial families was initiated in \cite{Kayal11,Kayal12}, where efficient ET algorithms were given for the Power Symmetric and Elementary Symmetric polynomials, the Permanent and the Determinant. Since then, efficient ET algorithms have been given for various important polynomial families. The ET algorithms are efficient even if $f$ is given as a black-box\footnote{Black-box access to a polynomial $f$ means we are given oracle access to $f$. Thus, we can query $f$ at any point $\veca$ of our choosing and obtain $f(\veca)$ in unit time.} or as a circuit. The study of ET for circuit classes is more recent, with the authors of \cite{GuptaST23} showing an efficient ET for read-once formulas (ROFs). The authors of \cite{BDSS24} showed that ET for sparse polynomials (depth-2 circuits) is NP-hard. Later, ET for read-once oblivious algebraic branching programs (ROABPs) was also shown to be NP-hard \cite{RS25,BDGT24}. See Appendix \ref{sec-ETres} for more results on ET for polynomial families and circuit classes.

Among all the polynomial families for which ET has been studied so far, the Permanent (see Definition \ref{def-HCn-Perm}) is the most prominent VNP-complete family.\footnote{The classes VP and VNP are the algebraic analogues of P and NP, respectively, see \cite{Val79a,Burgisser2000} for a rigorous definition. The authors of \cite{GuptaS19} studied ET for the Nisan-Wigderson polynomial (NW), which is in VNP but not known to be VNP-complete. The Pascal determinant, see \cite{Gurvits04} and Chapter 8 in \cite{landsberg2015geometry}, is a generalisation of the Determinant and is VNP-complete. The Pascal Determinant is characterised by its continuous symmetries \cite{landsberg2015geometry}, due to which an ET algorithm for the Pascal Determinant follows almost similarly to that for the Determinant given in Corollary 4.6.4 of \cite{grochowPhD}.} An efficient ET algorithm for a VNP-complete family does not imply the same for other VNP-complete families. Since ET has been shown to be NP-hard for some natural circuit classes, a natural question arises: 
\begin{center}
 \emph{Is ET ``hard'' for some ``natural'' VNP-complete polynomial family?}    
\end{center}
The Hamiltonian Cycle polynomial family, defined below, is another natural VNP-complete family.
\begin{definition}[Permanent and Hamiltonian Cycle polynomial families] \label{def-HCn-Perm}
    Let $n \geq 1$. The Permanent polynomial $Perm_n(X_n)$ and the Hamiltonian Cycle polynomial $HC_n(X_n)$ are homogeneous degree-$n$ polynomials defined on the $n \times n$ matrix of variables $X_n = (x_{i,j})_{i,j \in [n]}$  as:
    \begin{equation}
        Perm_n(X_n) := \sum_{\sigma \in S_n} \prod_{i \in [n]} x_{i, \sigma(i)}
        \quad \text{ and }\quad
        HC_n(X_n) := \sum_{\sigma \in C_n} \prod_{i \in [n]} x_{i, \sigma(i)},
    \end{equation}
respectively, with $HC_1$ defined as the zero polynomial. Here, $S_n$ and $C_n \subsetneq S_n$ are the set of all permutations on $\{1,2\dots n\}$, and the set of all \emph{cyclic} permutations on $\{1,2\dots n\}$, respectively. We denote by $Perm$ the family $\{Perm_n | \ n \geq 1\}$ and by $HC$ the family $\{HC_n | \ n \geq 1\}$. 
\end{definition}
Treating $X_n$ as the adjacency matrix of a complete digraph $G$ with edge weights as formal variables, $HC_n(X_n)$ is the sum of the weights of all the Hamiltonian cycles of $G$. The Permanent polynomial is the sum of all the weights of perfect matchings in a bipartite graph, whose biadjacency matrix is $X_n$. Observe that though $HC_n$ is defined on $n^2$ many variables, it depends only on the off-diagonal $n^2-n$ entries, which we denote as $\vecx$. We treat $HC_n$ as an $(n^2-n)$-variate polynomial in $\vecx$ variables, and the evaluation of $HC_n$ at an $n \times n$ matrix $A$ is denoted by $HC_n(A)$, where the diagonal entries of $A$ will be ignored.

After the Permanent, $HC$ seems to be the most prominent example of a VNP-complete family. In \cite{Val79a}, $HC$ was shown to be VNP-complete over \emph{every} field, while the VNP-completeness of Permanent holds only over fields of characteristic other than $2$. The proof of VNP-completeness of $HC$ by \cite{Val79a} does \emph{not} rely on the VNP-completeness of the Permanent. In fact, $HC$ is VNP-complete even over rings \cite{Malod03}. There have been various works where $HC$ has been studied in multiple contexts, such as lower bounds \cite{JS82,Jukna15,AB87,HI16}, and its connection to the Permanent and the Determinant \cite{GJ81,Burgisser2000,Gro17}. It has been the main choice as a VNP-complete family over all fields whenever such a family was sought to generalize results to all fields \cite{Malod07,KoiranPTT15,GMQ16,IM18,IS22,DRS22}. The author of \cite{Hrubes16} used the completeness of $HC$ over every field to show the existence of more VNP-complete families over every field. See Appendix \ref{sec-HCnmisc} for a detailed discussion on works involving $HC$. Thus, $HC$ is a well-studied, natural, and fundamental VNP-complete family.

In this work, we study the Equivalence Testing problem for $HC$, defined as follows:
\begin{problem}[ET for $HC$]
    Given black box access to an $n^2-n$-variate degree-$n$ polynomial $f(\vecx)$, decide if there exists $A \in \GL_{n^2-n}(\F)$ and $\vecb \in \F^{n^2-n}$ such that $f(\vecx) = HC_n(A\vecx + \vecb)$ and return $A$ and $\vecb$ if they exist.
\end{problem}

\paragraph{ET for $HC_n$.}While both the Permanent and $HC$ are VNP-complete, the efficient ET algorithm for Permanent does not imply the same for $HC$.\footnote{The completeness of the Permanent and $HC$ is with respect to $p$-projections \cite{Val79a}, which means any polynomial $f(\vecy)$ can be obtained from $Perm_m(X)$ (or $HC_n$) by replacing each $x$ variable in $Perm_m$ by a $\vecy$ variable or a field constant. In contrast, $f \sim Perm_m$ happens via \emph{invertible} linear transforms on the variables.  Thus, $g \sim HC_n$ does \emph{not} necessarily mean that $g \sim Perm_m$ for some $m$.} It is a priori unclear why ET for $HC$ could be ``easy'' because there are crucial differences between the Permanent and $HC$. First, the complexity of the \emph{zero-testing problem} for the Permanent and $HC$ is vastly different. The zero-testing problem for a polynomial family $\cF = \{f_n\}$, where $f_n$ is $n$-variate and has integer coefficients, is the task of checking for a \emph{given} point $\veca \in \{0,1\}^{n}$, whether $f(\veca)$ is $0$ or not over $\Z$. Note that $f$ is \emph{not} given to us in any form, that is, we don't even have black-box access to $f$. Zero-testing for the Permanent lies in P because it is the task of detecting a perfect matching in a bipartite graph.\footnote{In fact, it is in quasi-NC$^2$ \cite{FennerGT21}.} Zero-testing for $HC$ is NP-complete because it amounts to detecting a Hamiltonian Cycle in a digraph.\footnote{Similarly, the Permanent has an FPRAS \cite{JerrumSV04} while no FPRAS exists for $HC$ unless NP=RP. An FPRAS for the Permanent is a randomised algorithm which takes as input a non-negative $n \times n$ matrix $A$ and a parameter $ 0 < \epsilon < 1$, and with high probability outputs a value in the range $[(1-\epsilon)Perm_n(A), (1 + \epsilon)Perm_n(A)]$. The running time is $(n\frac{1}{\epsilon})^{O(1)}$.} In both cases, the graph is given as a 0/1 matrix. Second, the Permanent polynomial is known to be characterised by its symmetries \cite{MS01}, see Definition \ref{def-char-sym}, and \cite{grochowPhD} for a proof, but we show $HC_n$ is \emph{not} characterised by its symmetries in Theorem \ref{thm-nonsymchar}. This becomes important in the context of the \emph{circuit testing problem} for a polynomial family $\cF$ which asks to check if a given circuit $\mathrm{C}$ computes $f_n$ for some $n$, i.e., is $\mathrm{C}(\vecy) = f(\vecy)$? The circuit testing problem for a polynomial family $\cF$ is a special case of ET for $\cF$, with $A$ as the identity matrix and $\vecb = \veczero$. Efficient circuit testing algorithms for the Permanent follow by the downward self-reducibility\footnote{For the Permanent, downward self-reducibility means that computing the Permanent of $n \times n$ matrix polytime Turing reduces to computing the Permanent of $n$ many $(n-1) \times (n-1)$ matrices, which follows by the cofactor expansion of the Permanent. In general, a problem is said to be downward self-reducible if solving an instance $I$ of size $n$ Turing reduces to solving instances of smaller size.} of $Perm_n$ \cite{KI04} and also due to characterisation by symmetries of $Perm_n$ \cite{MS01}, see Page 81 in \cite{grochowPhD} for another proof. Since $HC_n$ is not characterised by its symmetries, one cannot hope to perform circuit testing efficiently for it via symmetries.  Given these crucial differences between the two polynomials, it is natural to ask:
\begin{center}
    \emph{Is ET ``easy'' for the Hamiltonian Cycle polynomial?}
\end{center}
In this work, we answer this question positively by showing a randomised polynomial time ET algorithm for $HC$ (see Theorem \ref{thm-ETalgo}). To our knowledge, an efficient ET for $HC$ was \emph{not} known before this work.

The author of \cite{Kayal12} developed the ET algorithm for the Permanent by using the knowledge of the symmetries and the Lie algebra of the Permanent, which have been well studied \cite{MM62,Botta67}. We follow the approach of \cite{Kayal12} and study the Lie algebra and symmetries of $HC_n$ to design an ET algorithm for it, though there are challenges to overcome in analysing the Lie algebra and designing the ET algorithm. The Permanent contains some ``nice'' monomials which are leveraged by \cite{Kayal12} in his algorithm, but these monomials are absent in $HC_n$. Thus, we deviate significantly from the implementation of the scaling matrix recovery step in \cite{Kayal12} and instead use ideas from \cite{GuptaS19}. The last step of the ET algorithm for $Perm_n$ uses circuit testing for a soundness check. Since $HC_n$ is not characterised by its symmetries, we cannot hope to perform circuit testing for it via this route. To perform circuit testing for $HC$, we show it is characterised by circuit identities (Theorems \ref{thm-ckt-id} and \ref{thm-ckt-test}) via its downward self-reducibility (Lemma 4 in \cite{BZ08}). See Section \ref{sec-proof-ideas} for a discussion on the proof ideas and Appendix \ref{sec-prevwork-comp} for a detailed comparison with \cite{Kayal12,GuptaS19}. Table \ref{tab:summary} summarises the results and compares the properties of $HC_n$ with the Permanent polynomial.

\subsection{Our Results} \label{Sec-results}
We first state our assumptions on the computational model. Over $\Q$ and finite fields $\F_q$, we assume the Turing machine model. Over fields like $\R$ and $\C$, we assume a model similar to the BSS model \cite{BlumSS89}, where one can perform an arithmetic operation in unit time. We also assume access to an efficient black-box univariate polynomial factorisation algorithm, an assumption which is justified over fields $\F_q$ \cite{Berlekamp70} and $\Q$ \cite{LLL82}. We consider $HC_n$ for $n \geq 3$. For $n = 2$, $HC_n$ is a monomial, and an efficient ET algorithm follows for it by the factorisation algorithm of \cite{KaltofenT90}.

\begin{theorem}[ET for $HC$] \label{thm-ETalgo}
   Let $\F$ be a field where $|\F| > 3n^5$ and $\text{char}(\F) = 0$ or $\text{char}(\F) > n$. There is a randomised $(n\beta)^{O(1)}$ time algorithm which, when given black-box access to an $(n^2-n)$-variate, homogeneous degree-$n$ polynomial $f(\vecx)$ with $\beta$ as the bit complexity of the coefficients, decides correctly, with high probability, if there exists $A \in \GL_{n^2-n}(\F)$ such that $f(\vecx) = HC_n(A\vecx)$ and outputs $A$ if it exists. Over finite fields, the running time is randomised $(n\log(|\F|))^{O(1)}$.
\end{theorem}
\begin{example}
\item The algorithm also works for the more general case where $f = HC_n(A\vecy + \vecb)$, $A \in \F^{(n^2-n) \times m}$ is a full row rank matrix (hence $n^2-n \leq m$) and $\vecb \in \F^{n^2-n}$ is a translation vector. See Theorem 28 in \cite{Kayal12} and Appendix B in \cite{KayalNST17}, which show how to appropriately modify the ET algorithm to handle the general case.

    \item The constraint on the characteristic arises from computing derivatives of $f$ and is imposed to ensure that these derivatives don't vanish. The constraint on the field size is due to the use of the Polynomial Identity Lemma \cite{DemilloL78,Lipton89,Schwartz80,Zippel79}.
\end{example}
The ET algorithm uses the knowledge of the group of symmetries $\mathcal{G}_{HC_n}$ and the Lie Algebra $\mathfrak{g}_{HC_n}$ of $HC_n$, see Definitions \ref{def-groupsym} and \ref{def-lie}. In Propositions \ref{prop-lie-alg-struct} and \ref{prop-lie-basis}, we completely characterise and give a basis for $\mathfrak{g}_{HC_n}$, while Theorem \ref{thm-PS-sym} gives the general structure of $\mathcal{G}_{HC_n}$. 
\begin{theorem}[Structure of $\mathcal{G}_{HC_n}$] \label{thm-PS-sym}
Let $|\F| > \binom{n^2-n}{2}$. If $A$ is a symmetry of $HC_n$, then $A = PS$ for some permutation matrix $P$ and scaling matrix $S$. 
\end{theorem}
\begin{example}
    \item Propositions \ref{prop-scalesymcont} and \ref{prop-Psym-gens} describe the permutation and scaling symmetries of $HC_n$ respectively. If we treat $HC_n$ as a polynomial in all the variables of $X_n$, then any $A \in \mathcal{G}_{HC_n}$ acts on $X_n$ as $A(X_n) = PLX_nRP^T$ or $A(X_n) = PLX_n^TRP^T$, where $P$ is a permutation matrix, and $L$ and $R$ are diagonal matrices such that $Perm_n(RL) = 1$. In contrast, any $A \in \mathcal{G}_{Perm_n}$ acts as $A(X_n) = PLX_nRQ$ or $A(X_n) = PLX_n^TRQ$, where $P$ and $Q$ are permutation matrices and $L$, $R$ are as before. 
\end{example}
Unlike $Perm_n$, $HC_n$ is \emph{not} characterised by its symmetries, see Definition \ref{def-char-sym}. 
\begin{theorem}[Non-characterisation by symmetries] \label{thm-nonsymchar}
   Let $n \geq 5$ and $\F$ be such that $|\F| > \binom{n^2-n}{2}$. Then, $HC_n$ is \emph{not} characterised by its symmetries over $\F$.
\end{theorem}
\begin{example}
    \item In Appendix \ref{subsec-charsym-HC4}, we show that $HC_n$, for $n = 3$ and $4$, is characterised by its symmetries over appropriate fields $\F$.
\end{example}
However, it turns out that $HC_n$ is characterised by circuit identities (see Definition \ref{def-char-cktid}) over \emph{any} field $\F$. This follows from the downward self-reducibility of $HC_n$ as shown in Lemma \ref{lemma-HCn-monotone}, and Lemma 4 of \cite{BZ08}.

\begin{theorem}[Circuit Identities] \label{thm-ckt-id}
Over any field $\F$, $HC_n$ is characterised by circuit identities. 
\end{theorem}
As a consequence of Theorem \ref{thm-ckt-id}, we get Theorems \ref{thm-ckt-test} and \ref{thm-flip-theorem}. We use Theorem \ref{thm-ckt-test} to prove the soundness of the ET algorithm given by Theorem \ref{thm-ETalgo}.
\begin{theorem}[Circuit Testability] \label{thm-ckt-test}
    Let $\mathrm{C}$ be a circuit in $n^2-n$ variables, of size $s$, and degree of output polynomial $d = s^{O(1)}$ over a field $\F$, where $|\F| > 2d$. There is a randomised $(ns)^{O(1)}$ time algorithm over $\F$ which, when given black box access to $\mathrm{C}$, decides correctly with high probability whether $\mathrm{C}(\vecx) = HC_n(\vecx)$ or not. 
\end{theorem}
A flip theorem essentially states that if some function is not computable by small-size circuits, then we can efficiently construct a short list of counterexamples such that any small-size circuit fails on some counterexample in the list. Flip theorem is known for the Permanent polynomial \cite{Mulmuley10,Mulmuley11} and the Nisan-Wigderson polynomial \cite{GuptaS19}. 
\begin{theorem}[Flip Theorem for $HC$] \label{thm-flip-theorem}
    Let $|\F| > n^{O(1)}$. Suppose $HC_n$ is not computable by $n^{O(1)}$ size circuits over $\F$. There is a randomised $n^{O(1)}$ time algorithm which on input $1^{n^2-n}$ outputs $n-1$ matrices $A_1, \dots A_{n-1}$, where $A_i \in \F^{n \times n}$, such that, with high probability, for any polynomial size circuit $\mathrm{C}$ there is an $i \in [n-1]$, $\mathrm{C}(A_i) \neq HC_n(A_i)$. Further, derandomisation of black-box polynomial identity testing implies that the $A_i$'s can be computed in deterministic $n^{O(1)}$ time.
\end{theorem}

\begin{example}
   \item Note that Theorem \ref{thm-ckt-test} is \emph{not} implied by Theorem \ref{thm-flip-theorem} because $HC_n$ is not known to be computable by $n^{O(1)}$ size circuit. Verifying $\mathrm{C}(A_i) \neq HC_n(A_i)$ for some $i$ needs the evaluations $HC_n(A_i)$, which can not be done efficiently. 
\end{example}
\begin{table} 
    \centering
    \begin{tabular}{|c|c|c|}
        \hline
         \textbf{\emph{Property}} & $HC_n$ \textbf{\emph{(This work)}} & \textbf{\emph{Permanent}} \\
         \hline
         \hline
         Efficient ET algorithm? & Yes & Yes \\
         \hline
         Explicit basis of Lie algebra known? & Yes & Yes \\ 
         \hline
         Characterisation by symmetries? & \textbf{No}, for $n \geq 5$ & \textbf{Almost all fields}  \\ 
         \hline
         Symmetries generated by PS matrices? & Yes  &  Yes  \\
         \hline
         Characterisation by circuit identities? & Yes  &  Yes  \\
         \hline
        Scaling symmetries continuous? 
        & \makecell{Almost all fields, \\ except for $n = 4$ over $\Q,\R,\C$ } & Almost all fields  \\
         \hline
        Flip Theorem known? & Yes & Yes \\
         \hline
        Complexity of Zero testing? & \textbf{NP-complete} & \textbf{P} \\
        \hline
    \end{tabular}
    \caption{A comparison between Permanent and $HC_n$}
    \label{tab:summary}
\end{table}

\subsection{Proof Techniques} \label{sec-proof-ideas}
We follow the Lie algebraic approach of \cite{Kayal12}, which was used to obtain an ET algorithm for the Permanent. We describe the ET algorithm for the $HC$ at a high level. Suppose we are given black-box access to an $n^2-n$-variate degree-$n$ polynomial $f(\vecx)$. The algorithm has four steps: 

\begin{enumerate}
\item \label{item-PSequiv} \emph{Reduction to $PS$-equivalence testing.} Using the structure of the Lie algebra of $HC_n$, compute an invertible $D$ such that $f(D\vecx) = HC_n(PS\vecx)$ for some permutation $P$ and scaling $S$, assuming $f \sim HC_n$. See Algorithm \ref{alg1}. 

\item \label{item-Sequivred} \emph{Reduction to $S$-equivalence testing.} Let $f_1(\vecx) = f(D\vecx)$. Using the knowledge of the permutation symmetries of $HC_n$, and the set of vanishing second-order partial derivatives of $HC_n$, recover $P$ so that $f_1(P\vecx) = HC_n(S\vecx)$, assuming $f \sim HC_n$. Thus, $f_1(P\vecx)$ is $S$-equivalent to $HC_n$. See Algorithm \ref{alg2}.

\item \label{item-Srecov} \emph{Recovering $S$.} Let $f_2(\vecx) = f_1(P\vecx)$, assuming $f \sim HC_n$. Query $f_2$ at $k = O(n^2)$ many points to get constants $c_1, c_2, \dots, c_k$, and solve a system of linear equations where the variables are the entries of $S$ which we wish to recover. See Algorithm \ref{alg3}.

\item \label{item-verify} \emph{Verification $\backslash$ Soundness Test.} Let $B = DPS$. Use the downward self-reducibility, Lemma \ref{lemma-HCn-monotone} or Lemma 4 of \cite{BZ08}, of $HC_n$ to verify if $f(B\vecx) = HC_n(\vecx)$. If so, output $B^{-1}$. See Algorithm \ref{alg4}.
\end{enumerate}
Though the ET algorithm for $HC$ is obtained via the Lie algebraic approach, like it was done for the Permanent in \cite{Kayal12}, there are two major differences between the algorithms. First, the Lie algebra and the symmetries of the Permanent have been well-studied \cite{Botta67,MM62}, and are leveraged by \cite{Kayal12}. To our knowledge, the Lie algebra and symmetries of $HC_n$ have not been studied before. Therefore, we first study and give a complete characterisation of the Lie algebra $\mathfrak{g}_{HC_n}$ of $HC_n$ over all fields (see Propositions \ref{prop-lie-alg-struct} and \ref{prop-lie-basis}) by using ideas from \cite{GuptaS19} to construct a basis for $\mathfrak{g}_{HC_n}$; we also characterise the group of symmetries $\mathcal{G}_{HC_n}$ over large enough fields (Proposition \ref{prop: symmetries-PS}). Second, we deviate from \cite{Kayal12} in the execution of Step \ref{item-Srecov}. The deviation occurs because of an important difference in the monomials of $HC_n$ and those of $Perm_n$. In $Perm_n$, there are enough pairs of monomials which differ in exactly $2$ variables and are leveraged by \cite{Kayal12} to recover $S$ by querying the scaled Permanent polynomial at these monomials. In contrast, \emph{any} two monomials of $HC_n$ differ in at least $3$ variables (see Claim \ref{claim: cycle-disjointness}), which makes it unclear if they can be leveraged in the same way as for $Perm_n$ to recover $S$. Instead, we follow the approach in \cite{GuptaS19} to recover $S$. There are other key technical differences between our analysis and that in \cite{Kayal12,GuptaS19}, see Appendix \ref{sec-prevwork-comp} for a detailed comparison. Refer to Table \ref{tab:summary} for a comparison of results between $HC_n$ and the Permanent.

To prove Theorem \ref{thm-PS-sym}, we leverage the structure of $\mathfrak{g}_{HC_n}$ and conjugacy of Lie algebras of equivalent polynomials (Lemma \ref{lemma-lieconjug}). We prove Theorem \ref{thm-nonsymchar} by using the knowledge of the permutation (Proposition \ref{prop-Psym-gens}) and scaling symmetries (Proposition \ref{prop-scalesymcont}) of $HC_n$. Theorem \ref{thm-ckt-id} is proved by using the downward self-reducibility of $HC_n$ and an adaptation of Lemma 7.13 from \cite{Burgisser2000}. Theorems \ref{thm-ckt-test} and \ref{thm-flip-theorem} then follow by using the identities established in Theorem \ref{thm-ckt-id}.
\section{Preliminaries} \label{Scn: Prelim}
\subsection{Notations and Definitions}
The set of natural numbers is $\N = \{1,2 \dots\}$, while $\Z$ denotes the integers. For $n \in \N$, We use $[n]$ to denote the set $\{1,\dots,n\}$ and $[a,b]$ to denote $\{a, \dots ,b\}$. The set of $n \times n$ invertible matrices over $\F$ is denoted by $GL_n(\F)$. The set of permutations on $[n]$ is denoted by $S_n$ and the set of cyclic permutations on $[n]$ by $C_n$. A field is denoted by $\F$, while $\F^{\times}$ denotes $ \F \backslash \{0\}$. We denote the integers mod $m$ by $\Z_m$. The dimension of a vector space $V$ over a field $\F$ is denoted by $\dim_{\F}(V)$. By $\in_r$ we denote uniform random selection from a set, and by ``b.b.a to $f$'' we mean black-box access to $f$.

We denote the set of variables on which $HC_n$ depends as $\vecx := \{x_{i,j} | \ i,j \in [n], i \neq j\}$, with $|\vecx| = n^2 - n$ and $X_n = \{x_{i,j} | \ i,j \in [n]\}$ to denote the entire set of $n^2$ variables on which $HC_n$ is defined. By scaling matrices, we mean diagonal matrices. We treat a scaling matrix $S \in \F^{m \times m}$ as a vector in $\F^{m}$ by identifying the diagonal as a vector in $\F^m$. Thus, the entries of $S \in \F^{n^2 \times n^2}$, will be referred to as $S_{i,j}$, where $i,j \in [n]$. For a matrix $M \in \F^{m \times m}$ and $S,T \subseteq [m]$, we denote by $M_{\bullet \times T}$ the matrix $M$ restricted to columns in $T$, by $M_{S \times \bullet}$ the matrix $M$ restricted to rows in $S$, and by $M_{S \times T}$ the matrix $M$ restricted to the rows and columns in $S$ and $T$ respectively. By $n^{O(1)}$, we denote a polynomially bounded function of $n$.

\subsection{Algebraic and algorithmic preliminaries}
\begin{lemma}[Computing Derivatives \cite{KayalNST17}] \label{lemma-pd-compute}
Let $f \in \F[\vecy]$ be an $n$-variate degree-$d$ polynomial with bit complexity $\beta$. Given black-box access to $f$, we can obtain in $(nd\beta)^{O(1)}$ time black-box access to a derivative $\frac{\partial f}{ \partial y}$ for any $y \in \vecy$.    
\end{lemma}

\begin{definition}[Group of Symmetries \cite{GuptaS19}] \label{def-groupsym}
Let $f \in \F[\vecy]$ be an $n$-variate polynomial. The group of symmetries of $f$ over $\F$, denoted by $\mathcal{G}_f$, is the set of matrices $\{A \in \GL_{n}(\F) : f(A\vecy) = f(\vecy)\}$ which also form a group under matrix multiplication.
\end{definition}
\begin{definition}[Lie Algebra of a polynomial ] \label{def-lie}
Let $f \in \F[\vecy]$, where $\vecy := \{y_1,y_2, \dots, y_n\}$. The Lie Algebra associated with $f$, denoted by $\mathfrak{g}_f$, is the set $\{A \in \F^{n \times n} \ | \ \sum_{i,j \in [n]} A_{i,j}y_j \frac{\partial f}{\partial y_i} = 0 \}$, which also forms a vector space over $\F$.
\end{definition}

\begin{lemma}[Conjugacy of Lie Algebras \cite{Kayal12}] \label{lemma-lieconjug}
Let $f \in \F[\vecy]$ be an $n$-variate polynomial. If $g(\vecy) = f(A\vecy)$, where $A \in \GL_{n}(\F)$, then $\mathfrak{g}_g$ = $A^{-1}\cdot \mathfrak{g}_f \cdot A$.
\end{lemma}

\begin{lemma}[Computing basis of Lie Algebra \cite{Kayal12}] \label{lemma-liebasis-comp}
Given black-box access to an $n$-variate degree $d$ polynomial $f \in \F[\vecy]$, a basis of $\mathfrak{g}_f$ can be computed in randomised $(nd\beta)^{O(1)}$ time, where $\beta$ is the bit complexity of the coefficients.
\end{lemma}
For $A \in \C^{n \times n}$, define $e^{A} := \sum_{i \in \N} \frac{A^i}{i!}$, the sum always converges. Over $\C$, $\mathfrak{g}_f$ is related to $\mathcal{G}_f$ as in Definition \ref{def-contsym}. 
\begin{definition}[Continuous and Discrete Symmetries ] \label{def-contsym}
Let $f \in \C[\vecy]$. If $A \in$ $\mathfrak{g}_f$, then $e^{tA} \in \mathcal{G}_f$ for all $t \in \R$; see \cite{Hall15} for a proof. The continuous symmetries of $f$ are the elements of the set $\{e^{tA}: A \in \mathcal{G}_f \text{ and } t \in \R\}$. All other symmetries in $\mathcal{G}_f$ are the discrete symmetries of $f$. 
\end{definition}
\begin{definition}[Characterisation by circuit identities \cite{grochowPhD}] \label{def-char-cktid}
    An $n$-variate polynomial $f(\vecy)$ is said to be characterised by circuit identities if there exist $m =n^{O(1)}$ many polynomials $g_1(\vecz), \dots, g_m(\vecz)$, with each $g_i$ computable by a $n^{O(1)}$ size circuit over $\Z$ and $|\vecz| \leq k$, $f$ is the only non-zero polynomial upto scaling that satisfies $g_i(f(\vecy_1),f(\vecy_2) \dots f(\vecy_k)) = 0$, where $\vecy_i$ can be computed from $\vecy$ by a $n^{O(1)}$ size circuit.
\end{definition}
\begin{definition}[Characterisation by symmetries] \label{def-char-sym}
    A homogeneous degree-$d$ polynomial $f \in \F[\vecy]$ is characterised by its symmetries if for every homogeneous degree-$d$ polynomial  $g \in \F[\vecy]$,  $\mathcal{G}_f \subseteq \mathcal{G}_g$ implies $g = c \cdot f$, for somes $c \in \F^{\times}$.
\end{definition}
\section{Lie Algebra and the Symmetries of $HC_n$} \label{Scn: Lie_alg}
In Section \ref{subsec-lie-alg}, we analyse $\mathfrak{g}_{HC_n}$, the Lie Algebra of $HC_n$, and show that it comprises diagonal matrices and that it is a $2n-2$ dimensional vector space over any field $\F$, for $n=3$ and $n \geq 5$. We analyse $\mathfrak{g}_{HC_4}$ and the symmetries of $HC_4$ in Appendix \ref{sec-HC4-results}. In Section \ref{subsec-sym}, we use $\mathfrak{g}_{HC_n}$ to show that over large enough fields the symmetries of $HC_n$ are generated by permutation and scaling matrices (which proves Theorem \ref{thm-PS-sym}) and analyse the permutation and scaling symmetries of $HC_n$. In Section \ref{subsec-non-char}, we show that $HC_n$ is \emph{not} characterised by its symmetries, proving Theorem \ref{thm-nonsymchar}. In Section \ref{subsec-Ckt-id}, we show that $HC_n$ is characterised by circuit identities over any field by using the downward self-reducibility of $HC_n$ and prove Theorems \ref{thm-ckt-id}, \ref{thm-ckt-test} and \ref{thm-flip-theorem}. All missing proofs can be found in Appendix \ref{proofs-lie}.

\subsection{Lie Algebra} \label{subsec-lie-alg}
Proposition \ref{prop-lie-alg-struct}, proved in Appendix \ref{proof-lie-alg-struct}, shows that $\mathfrak{g}_{HC_n}$ comprises diagonal matrices $A \in \F^{(n^2-n) \times (n^2-n)}$, where for all $\sigma \in C_n$ the sum of the diagonal entries in $A$ corresponding to $\sigma$ is $0$. To prove Proposition \ref{prop-lie-alg-struct}, we use Observation \ref{obs: disj-deriv}, which follows from Claim \ref{claim: cycle-disjointness}. Claim \ref{claim: cycle-disjointness} states that any two cyclic permutations, when interpreted as Hamiltonian cycles on the complete digraph, must have at least $3$ different edges.

\begin{proposition} \label{prop-lie-alg-struct}
Let $n \geq 3$. Then for $A \in \F^{(n^2-n) \times (n^2-n)}$, 
\[A \in  \mathfrak{g}_{HC_n} \iff A\text{ is diagonal and }\sum_{i=1}^{n} A_{(i,\sigma(i)),(i,\sigma(i))} = 0 \text{  for all $\sigma \in C_n$.}\] 
\end{proposition}

\begin{observation}  \label{obs: disj-deriv}
Let $\sigma_1, \sigma_2 \in C_n$ with $\sigma_1 \neq \sigma_2$ and  $m_\sigma $ denote $\prod_{i \in [n]}x_{i,\sigma(i)}$ for $\sigma \in C_n$. Then,
    \begin{equation*}
        x_{i_1,j_1} \frac{\partial m_{\sigma_1}}{\partial x_{i_2,j_2}} \neq x_{i_3,j_3} \frac{\partial m_{\sigma_2}}{\partial x_{i_4,j_4}},
    \end{equation*}
    where $i_k$'s, $j_k$'s $\in [n]$, $j_2 = \sigma_1(i_2)$, and $j_4 = \sigma_2(i_4)$ (to ensure non-zero derivatives).
\end{observation}

\begin{claim} \label{claim: cycle-disjointness}
Let $\sigma_1, \sigma_2 \in C_n$ with $\sigma_1 \neq \sigma_2$ . There exists $S \subseteq [n]$, such that $|S| = 3$ and $\sigma_1(i) \neq \sigma_2(i)$ for all $i \in S$.
\end{claim}
By Proposition \ref{prop-lie-alg-struct}, any $A \in \mathfrak{g}_{HC_n}$ can be identified with a vector in $\F^{n^2-n}$. Proposition \ref{prop-lie-basis} shows that for $n \neq 4$, $\mathfrak{g}_{HC_n}$ is a $(2n-2)$-dimensional over any $\F$. In Appendix \ref{subsec-lie-HC4}, we analyse $\mathfrak{g}_{HC_4}$ over all fields. 
\begin{proposition} \label{prop-lie-basis}
    Let $\F$ be any field and $n = 3$ or $n \geq 5$. Consider $A^{(k)}$, $B^{(\ell)}$ and $C \in \F^{n^2-n}$, where $k \in [2,n], \ell \in [2,n-1]$, defined as:
    \[A^{(k)}_{(i,j)} = \begin{cases} 
          1 & i = 1 , \\
          -1 & i = k , \\
          0 & \text{ otherwise} 
       \end{cases}, \ \quad B^{(\ell)}_{(i,j)} = \begin{cases} 
          1 &  j = 1, \\
          -1 & j = l, \\
          0 & \text{ otherwise} 
       \end{cases}, \ \quad C_{(i,j)} = \begin{cases} 
          -1 & i = 2, \\
          1 & j = 2, \\
          0 & \text{ otherwise} 
       \end{cases}\]
    Then, $\text{dim}_{\F}(\mathfrak{g}_{HC_n}) = 2n-2$ with $\mathcal{B}_n = \{A^{(2)}, \dots, A^{(n)}, B^{(2)}, \dots, B^{(n-1)}, C\}$ as a basis. 
\end{proposition}
Proposition \ref{prop-lie-basis} is proved in Appendix \ref{proof-lie-basis}. In the proof, we define a matrix $M^{HC_n} \in \F^{(n-1)! \times (n^2-n)}$ to represent the linear equations given by Proposition \ref{prop-lie-alg-struct}. Therefore, the nullspace of $M^{HC_n}$ is $\mathfrak{g}_{HC_n}$. The rows of $M^{HC_n}$ are indexed by $\sigma \in C_n$ and columns by variables $x_{i,j}$ in lex order. We then show that the set $\mathcal{B}_n$ in the proposition is linearly independent over \emph{any} $\F$ and that $\mathcal{B}_n \subset \mathfrak{g}_{HC_n}$. Lastly, we construct in $n^{O(1)}$ time a $(n-1)(n-2) \times (n^2-n)$ submatrix $M^{(n)}$ of $M^{HC_n}$ such that $M^{(n)}$ has full row rank over \emph{any} $\F$ (see Proposition \ref{prop-lie-dim-lb}). Thus, $M^{HC_n}$ has rank at least $(n-1)(n-2)$, which along with the fact that $\mathcal{B}_n \subset \mathfrak{g}_{HC_n}$ proves Proposition \ref{prop-lie-basis}. 

Corollary \ref{corollary-lie-soln-struct} describes the structure of any $\vecz \in \mathfrak{g}_{HC_n}$ and follows by considering the matrix formed by the equations in \eqref{eqn-lin-form-liebasis}, noting that the matrix is full row rank, and that the entries of the elements of $\mathcal{B}_{n}$ in Proposition \ref{prop-lie-basis} satisfy \eqref{eqn-lin-form-liebasis}. Corollary \ref{corollary-continuous-scaling-sym} describes the entries of any continuous symmetry of $HC_n$, and follows from Corollary \ref{corollary-lie-soln-struct} and Definition \ref{def-contsym}. In Appendix \ref{proof-det1-abelgroup}, we further analyse $M^{(n)}$, which proves useful in analysing the scaling symmetries of $HC_n$ over all fields and to prove the correctness of Algorithm \ref{alg3}.

\begin{corollary} \label{corollary-lie-soln-struct}
    Let $n = 3$ or $n \geq 5$ and $\vecz \in \F^{n^2-n}$. Then, $\vecz \in \mathfrak{g}_{HC_n}$ if and only if the entries $z_{i,j}$ satisfy \eqref{eqn-lin-form-liebasis}.
    \begin{equation} \label{eqn-lin-form-liebasis}
\begin{split}
    z_{i,j} &= z_{i,1} + z_{1,j} + z_{2,3} -z_{1,3} -z_{2,1}  \ \ \  i \in [2,n-1], \ j \in [2,n] \ , i \neq j, (i,j) \neq (2,3),\\
    z_{n,1} &= (n-2)(z_{1,3} + z_{2,1} - z_{2,3}) - \sum_{\ell=2}^{n} z_{1,\ell}   - \sum_{\ell=2}^{n-1} z_{\ell,1}, \\
    z_{n,j} &=  z_{1,j} + (n-3)(z_{1,3} + z_{2,1} - z_{2,3})- \sum_{\ell = 2}^{n} z_{1,\ell} - \sum_{\ell=2}^{n-1} z_{\ell,1} \ \  \\
            &= z_{n,1} + z_{1,j} + z_{2,3} -z_{1,3} -z_{2,1}, \ \ j \in [2,n-1].
\end{split}
\end{equation}

\end{corollary} 
\begin{corollary} \label{corollary-continuous-scaling-sym}
Over $\C$, any continuous symmetry $S \in \mathcal{G}_{HC_n}$ is a diagonal matrix, with $S_{i,j}$'s satisfying \eqref{eqn-scaling-sym-lie}. 
    \begin{equation} \label{eqn-scaling-sym-lie}
\begin{split}
    S_{i,j} &=  S_{i,1}S_{1,j}\frac{S_{2,3}}{S_{1,3}S_{2,1}}  \quad i \in [2,n-1], \ j \in [2,n] \ , i \neq j, (i,j) \neq (2,3),\\
    S_{n,1} &= \left(\frac{S_{1,3}S_{2,1}}{S_{2,3}}\right)^{n-2} \left(\prod_{\ell=2}^{n} S_{1,\ell} \prod_{\ell=2}^{n-1} S_{\ell,1}\right)^{-1}, \\
    S_{n,j} &= S_{n,1}S_{1,j} \frac{S_{2,3}}{S_{1,3}S_{2,1}} \quad j \in [2,n-1].
\end{split}
\end{equation}
\end{corollary}
\subsection{The Symmetries of $HC_n$} \label{subsec-sym}
Lemma \ref{prop-dist-eigen} shows that over large enough fields, $\mathfrak{g}_{HC_n}$ contains an element with distinct eigenvalues. Proposition \ref{prop: symmetries-PS} then shows that over such fields, every symmetry of $HC_n$ is generated by permutation and scaling symmetries, proving Theorem \ref{thm-PS-sym}. The proof of Proposition \ref{prop: symmetries-PS} follows from Lemmas \ref{lemma-lieconjug} and  \ref{prop-dist-eigen} (see Appendix \ref{proof-symmetries-PS}). We then analyse and characterise the permutation and scaling symmetries of $HC_n$.
\begin{lemma} \label{prop-dist-eigen}
    Suppose $|\F| > \binom{n^2-n}{2}$. There exists $A \in \mathfrak{g}_{HC_n}$ such that $A$ has distinct eigenvalues.
\end{lemma}

\begin{proposition}[Theorem \ref{thm-PS-sym} restated] \label{prop: symmetries-PS}
    Suppose $|\F| > \binom{n^2-n}{2}$. If $A \in \mathcal{G}_{HC_n}$, then $A = PS$, where $P$ is a permutation matrix and $S$ is a scaling matrix.
\end{proposition}

\noindent{\textbf{Scaling symmetries.}} Proposition \ref{prop-scalesymcont}, proved in Appendix \ref{proof-scalesymcont}, shows that any scaling symmetry $S$ of $HC_n$ satisfies \eqref{eqn-scaling-sym-lie} for $n \neq 4$ and follows from Lemma \ref{lemma-abeliangroup-soln}. In Appendix \ref{subsec-Sym-HC4}, we show that $HC_4$ has discrete scaling symmetries over $\Q,$ $\R,$ $\C$ and certain finite fields.

\begin{proposition}[Scaling symmetries are continuous] \label{prop-scalesymcont}
    Let $\F$ be any field and $n = 3$ or $n \geq 5$. If $S \in \mathcal{G}_{HC_n}$ is a scaling symmetry, then the entries $S_{i,j}$ satisfy \eqref{eqn-scaling-sym-lie}.  

\end{proposition}
\noindent{\textbf{Permutation Symmetries.}}  Proposition \ref{prop-psym} shows that $HC_n$ has non-trivial permutation symmetries. In terms of the matrix $X_n$, $P^{(\sigma)}$ acts as $X_n \mapsto QX_nQ^T$, where $Q$ is the $n \times n$ matrix corresponding to $\sigma \in S_n$, and $P^{(T)}$ acts as $X_n \mapsto X_n^T$. Proposition \ref{prop-Psym-gens}, proved in Appendix \ref{proof-Psym-gens}, shows that the permutation symmetries are generated by those described in Proposition \ref{prop-psym}. To prove Proposition \ref{prop-Psym-gens}, we use Observation \ref{obs-HCn-pdzero}, which tells us when the second-order derivatives of $HC_n$ vanish.
\begin{proposition} \label{prop-psym}
Let $\sigma \in S_n$. The matrices $P^{(\sigma)}$ and $P^{(T)}$, as below, are permutation symmetries of $HC_n$.
\[
P^{(\sigma)}_{(i,j),(k,\ell)} := \begin{cases} 
          1 & k = \sigma(i), \ell = \sigma(j), \\
          0 & \text{ otherwise} 
       \end{cases},
\quad P^{(T)}_{(i,j),(k,\ell)} := \begin{cases} 
          1 & k = j, \ell = i, \\
          0 & \text{ otherwise} 
        \end{cases}.
\]
Here $i,j,k,\ell \in [n]$, $i \neq j$ and $k \neq \ell$. Also, $P^{(\sigma)}$ and $P^{(T)}$ commute with one another.
\end{proposition}

\begin{proposition} \label{prop-Psym-gens}
    If $P \in \mathcal{G}_{HC_n}$ is a permutation symmetry, then $P = P^{(\sigma)}$ or $P = P^{(\sigma)}P^{(T)}$ for some $\sigma \in S_n$.
\end{proposition}
\begin{observation} \label{obs-HCn-pdzero}
Let $i,j \in [n]$, $i \neq j$ and $R_{i,j}$ be the set of variables $x_{k,\ell}$, other than $x_{i,j}$, such that $\frac{\partial^2 HC_n}{\partial x_{i,j} \partial x_{k,\ell}} = 0$. Then, 
\[x_{k,\ell} \in R_{i,j} \iff i = k \text{, or } j = \ell, \text{ or } i = \ell \text{ and } j = k.\]
Further, we can partition $R_{i,j}$ as
\[R_{i,j} := Q_{i,j} \sqcup T_{i,j} \sqcup \{x_{j,i}\}, \ Q_{i,j} := \{x_{k,j} \ | \ k \in [n] \backslash \{i,j\}\} \text{ and } T_{i,j} := \{x_{i,k} \ | \ k \in [n]\backslash \{i,j\}\}.\]
The partitions also satisfy:
\begin{enumerate}
    \item \label{PDset-prop-1} $\frac{\partial^2 HC_n}{\partial x_{i_1,j_1}\partial x_{i_2,j_2}} \neq 0$ if $x_{i_1,j_1} = x_{j,i}$ and  $x_{i_2,j_2} \in Q_{i,j} \sqcup T_{i,j}$, or $x_{i_1,j_1} \in T_{i,j} ,x_{i_2,j_2} \in Q_{i,j}$.
    \item \label{PDset-prop-2} $\frac{\partial^2 HC_n}{\partial x_{i_1,j_1}\partial x_{i_2,j_2}} = 0$ if $x_{i_1,j_1},x_{i_2,j_2} \in Q_{i,j}$ or $x_{i_1,j_1},x_{i_2,j_2} \in T_{i,j}$. 
\end{enumerate}
\end{observation} 

If $f = HC_n(P\vecx)$, where $P$ is a permutation matrix, then $\frac{\partial^2 f}{\partial x_{i,j} \partial x_{k,\ell}} = 0$ if and only if $P^{-1}(x_{i,j})$ and $P^{-1}(x_{k,\ell})$ are as specified in Observation \ref{obs-HCn-pdzero}. Thus, $P$ creates a bijection on $R_{i,j}$ which also maps the partitions of $R_{i,j}$ in Observation \ref{obs-HCn-pdzero} appropriately. If $P \in \mathcal{G}_{HC_n}$ is a permutation symmetry, then it creates, for each variable $x_{i,j}$, a bijection between the set $R_{i,j}$ and $R_{k,\ell}$, where $x_{k,\ell} = P(x_{i,j})$, such that the bijection also maps the partitions of $R_{i,j}$ as described in Observation \ref{obs-HCn-pdzero} accordingly to those of $R_{k,\ell}$. analysing these bijections for variables $x_{1,j}$ and $x_{i,1}$, and noting that the image of $R_{1,j} \cap R_{i,1}$ is a singleton shows how $P$ is generated by the permutation symmetries described in Proposition \ref{prop-psym}. 

\subsection{$HC_n$ is not characterised by its symmetries} \label{subsec-non-char}
Proposition \ref{prop-nonchar}, proved in Appendix \ref{proof-nonchar}, shows that for $n \geq 5$, $HC_n$ is \emph{not} characterised by its symmetries, unlike the Permanent polynomial, proving Theorem \ref{thm-nonsymchar}. The proof involves defining a polynomial $g(\vecx)$ as the sum over the image of a monomial, $m_{\sigma}$, under all the permutation symmetries (Proposition \ref{prop-Psym-gens}) of $HC_n$. Here, $\sigma \in S_n \backslash C_n$ such that $\sigma(i) \neq i$ for all $i \in [n]$. It is easy to see that such a $\sigma$ exists for $n \geq 5$, for example $\sigma = (1 \ 2)(3 \ 4 \ \dots \ n)$. Then, it is not hard to observe that every scaling symmetry of $HC_n$ (Proposition \ref{prop-scalesymcont}) is also a scaling symmetry of $g$, but $g \neq c\cdot HC_n$ for any $c \in \F^{\times}$. 
\begin{proposition}[Non-characterisation by symmetries] \label{prop-nonchar}
    Let $n \geq 5$ and $\F$ be any field such that $|\F| > \binom{n^2-n}{2}$. There exists a non-zero polynomial $f$ such that $\mathcal{G}_{HC_n} \subseteq \mathcal{G}_{f} $ but $f \neq c\cdot HC_n$ for any $c \in \F^{\times}$.
\end{proposition}
More generally, in Appendix \ref{sec-scalesymrelate-hcperm} we show how to obtain a scaling symmetry of the Permanent polynomial from that of $HC_n$. In contrast, for $n = 3$ and $n = 4$, we show $HC_n$ is characterised by its symmetries over appropriate fields in Appendix \ref{subsec-charsym-HC4}.

\subsection{Circuit Identities, Circuit Testing and Flip Theorem} \label{subsec-Ckt-id}
Though $HC_n$ is not characterised by its symmetries, it is characterised by circuit identities (see Definition \ref{def-char-cktid}), which follows from $HC_n$ being downward self-reducible, as shown in Lemma 4 of \cite{BZ08}. We give a proof of downward self-reducibility in the proof of Lemma \ref{lemma-HCn-downself} for completeness (see Appendix \ref{proof-HCn-downself}). The proof given here was discovered by the author independently before coming across \cite{BZ08}. To prove the circuit identities, we use Lemma \ref{lemma-HCn-monotone}, an adaptation of Lemma 7.13 in \cite{Burgisser2000}, which shows that $HC_n$ can be expressed as $HC_m$ for all $m > n \geq 2$. Lemma \ref{lemma-HCn-downself} then shows that $HC_n$ is downward self-reducible over any field, and can be proved using induction on $n$. Note that we use the matrix $X_n$ with $x_{i,i}$ replaced by $0$'s in the statement and proof of Lemmas \ref{lemma-HCn-monotone} and \ref{lemma-HCn-downself} for ease of exposition. 

\begin{lemma} \label{lemma-HCn-monotone}
    Let $n,m \in \N$ with $m > n \geq 2$. In $O(m^2)$ time, we can construct a $m \times m$ matrix $Y^{(m,n)}$ from $X_n$ such that $HC_m(Y^{(m,n)}) = HC_n(\vecx)$.
\end{lemma}

\begin{lemma} \label{lemma-HCn-downself}
    Let $f \in \F[\vecx]$. Then, $f = HC_n$ if and only if $f$ satisfies the identities
    \[ f(Y^{(n,k)}) = \sum_{i=2}^{k}x_{1,i}f(Y^{(n,k-1)}_{i}) \quad k \in [3,n], \text{ and }f(Y^{(n,2)}) = x_{1,2}x_{2,1},\]
    where $Y^{(n,k)}$ is constructed from $X_k$ as described in the proof of Lemma \ref{lemma-HCn-monotone}. The matrix $Y^{(n,k-1)}_i$ is obtained from $X_k$ by first swapping row $i$ with row $2$ and column $i$ with column $2$ in $X_k$, removing the 1st row and 2nd column to obtain the $(k-1) \times (k-1)$ matrix $X^{(i)}_{k}$, and constructing $Y^{(n,k-1)}$ using $X^{(i)}_{k}$ by Lemma \ref{lemma-HCn-monotone}.
\end{lemma}
\paragraph{Proof of Theorem \ref{thm-ckt-id}.} Represent the $n-1$ identities in Lemma \ref{lemma-HCn-downself} as polynomials $g_1(\vecx,\vecz), \dots, g_{n-1}(\vecx,\vecz)$, where $|\vecz| = n-1$, $g_i$'s are  degree $2$ polynomials over $\Z$ and have $n^{O(1)}$ size circuit over $\Z$. This establishes circuit identities for $HC_n$ and proves Theorem \ref{thm-ckt-id}. 

\paragraph{Proof of Theorems \ref{thm-ckt-test} and \ref{thm-flip-theorem}.} 
Algorithm \ref{alg4}, when given black-box access to a size-$s$ circuit $\mathrm{C}$ in $\vecx$ variables and of degree $d = s^{O(1)}$, tests $\mathrm{C}$ on the $n-1$ identities in Lemma \ref{lemma-HCn-downself}. The degree of each identity, when we use $\mathrm{C}$ in it, is at most $d$. If $\mathrm{C}$ computes $HC_n$, then  Algorithm \ref{alg4} will always return ``Yes''. If $\mathrm{C}$ does \emph{not} compute $HC_n$, then some identity fails to hold for $\mathrm{C}$. By the Polynomial Identity Lemma and the fact $|U| > 2d$, the probability that some identity fails to hold is at least $\frac{1}{2}$. By repeating Algorithm \ref{alg4} $s^{O(1)}$ times, we can boost the success probability to $1 - (\frac{1}{2})^{s^{O(1)}}$. This proves Theorem \ref{thm-ckt-test}. Theorem \ref{thm-flip-theorem} follows easily because the random matrices sampled by Algorithm \ref{alg4} provide the list of counterexamples against polynomial-size circuits, assuming $HC_n$ is not computable by polynomial-size circuits.

\begin{algorithm} 
\caption{Circuit Testing for $HC_n$} \label{alg4}
\textbf{Input}: B.b.a to size-$s$ circuit $\mathrm{C}$ in $\vecx$ variables and of degree $s^{O(1)}$. \\
\textbf{Output}: ``Yes'' if $\mathrm{C}(\vecx) = HC_n(\vecx)$, else ``No''. 

\begin{algorithmic}[1]
\State Let $U \subseteq \F$ with $|U| > 2d$. Choose $n-1$ matrices $M_2, M_3, \dots, M_{n}$, where $M_i$ are $i \times i$ matrices with entries chosen uniformly at random from $U$.
\State For $i \in [2,n]$, check if the $i$'th identity in Lemma \ref{lemma-HCn-downself} does not hold for $\mathrm{C}$ when evaluated at $M_i$. Return ``No'' if so. 
\State Return ``Yes''.
\end{algorithmic}
\end{algorithm}

\section{Equivalence Testing for $HC$} \label{sec-ETalgo}
In this section, we present the ET algorithm for $HC_n$, and argue its correctness and efficiency by using the structural results proven in Section \ref{Scn: Lie_alg}. In Section \ref{subsec-PSreduction}, we show the reduction to $PS$-equivalence; the correctness of the reduction follows from Lemma \ref{prop-dist-eigen}. In Section \ref{subsec-Pequiv-test}, we show how to recover a permutation; the correctness follows from Observation \ref{obs-HCn-pdzero}. The problem then reduces to checking scaling equivalence with $HC_n$. In Section \ref{subsec-Sequiv-test}, we show how to recover a scaling matrix by using Proposition \ref{prop-scalesymcont} and Lemma \ref{lemma-abeliangroup-soln}. All missing proofs can be found in Appendix \ref{sec-ETalgo-proofs}. We present the analysis assuming $f = HC_n(A\vecx)$ and that all steps which employ randomisation execute correctly. Every step where randomisation is used has a small probability of failing due to the assumption on the field size. We run Algorithm \ref{alg4} at the end to verify the correctness of the recovered transform. If $f \neq HC_n(A\vecx)$, Algorithm \ref{alg4} will output ``No'' with high probability.  If $f = HC_n(A\vecx)$, Algorithm \ref{alg4} always accepts.  

\subsection{Reduction to PS-equivalence} \label{subsec-PSreduction}
Algorithm \ref{alg1} shows the reduction to $PS$-equivalence testing. If $f = HC_n(A\vecx)$, then by Lemma \ref{lemma-lieconjug}, $\dim_{\F}(\mathfrak{g}_f)$ is $2n-2$.We then compute a random element of $\mathfrak{g}_f$ and diagonalise it. The algorithm outputs the diagonalising matrix if it exists. The correctness and complexity analysis of Algorithm \ref{alg1} is given by Proposition \ref{prop:PS-equiv-reduction}.
\begin{proposition} \label{prop:PS-equiv-reduction}
    If $f = HC_n(A\vecx)$ for some $A \in GL_{n^2-n}(\F)$, with high probability Algorithm \ref{alg1} computes in randomised $n^{O(1)}$ time a $D \in GL_{n^2-n}(\F)$ such that $f(D\vecx) = HC_n(PS\vecx)$ .
\end{proposition}

\begin{algorithm}
\caption{Reduction to PS-equivalence Test} \label{alg1}

\textbf{Input}: B.b.a to $f(\vecx)$.\\
\textbf{Output}: $D \in \GL_{n^2-n}(\F)$ s.t. $f(D\vecx) = HC_n(PS\vecx)$  if $f(\vecx) = HC_n(A\vecx)$.

\begin{algorithmic}[1]    
\State Compute a basis $\{B_1, B_2 , \dots, B_k \}$ of $\mathfrak{g}_f$ from b.b.a to $f$. If $k \neq 2n-2$, abort and output ``$f$ not equivalent to $HC_n$''.
\State Let $U \subseteq \F$ be such that $|U| = 3n^5$. Let $a_1, a_2 \dots a_k \in_r U$ and $C := \sum_{i \in [k]}a_i B_i$. Compute and output a matrix $D$ such that $D^{-1}CD$ is a diagonal matrix. If such a $D$ does not exist, abort output ``$f$ not equivalent to $HC_n$''.
\end{algorithmic}
\end{algorithm}
\subsection{P-equivalence test} \label{subsec-Pequiv-test}
Let $f(\vecx) = HC_n(P\vecx)$. Using the permutation symmetries and second-order partial derivatives of $HC_n$, Algorithm \ref{alg2} recovers a permutation $P'$ such that $f(\vecx) = HC_n(P'\vecx)$, where $P'$ is $P$ up to a permutation symmetry. Claim \ref{claim-Ptest-assum} and Proposition \ref{prop-ptest-correct} argue the correctness.
\begin{claim} \label{claim-Ptest-assum}
    Without loss of generality, $P(x_{1,2}) = x_{1,2}$.
\end{claim}

\begin{proposition} \label{prop-ptest-correct}
    If $f(\vecx) = HC_n(P\vecx)$, then with high probability Algorithm \ref{alg2} computes $P'$ in randomised $n^{O(1)}$ time such that $P' = \tilde{P}P$, where $\tilde{P}$ is as described in Proposition \ref{prop-psym}.  
\end{proposition}

\begin{algorithm} 
\caption{P-equivalence Test} \label{alg2}

\textbf{Input}: B.b.a to $f(\vecx)$.\\
\textbf{Output}: $P' \in \GL_{n^2-n}(\F)$ s.t. $HC_n(P'\vecx) = f(\vecx)$, if $f(\vecx) = HC_n(P\vecx)$. 

\begin{algorithmic}[1]    
\State Compute b.b.a to $\frac{\partial^2 HC_n}{\partial x_{i,j} \partial x_{k,\ell}}$, where $i \neq j$ and $k \neq \ell$. Check if $\frac{\partial^2 HC_n}{\partial x_{i,j}\partial x_{k,\ell}}$ is identically zero or not, and store this information.

\State Compute the set $R'_{1,2} := \{x_{k,\ell} \ | \frac{\partial^2 HC_n}{\partial x_{1,2} \partial x_{k,\ell}} = 0 \text{ and } (k,\ell) \neq (1,2) \}$. Partition $R'_{1,2}$ as $\{x_{i,j}\}\sqcup Q'_{1,2}\sqcup T'_{1,2}$ such that items \ref{PDset-prop-1} and \ref{PDset-prop-2} in Observation \ref{obs-HCn-pdzero} hold. If no such partition exists, abort and output ``$f$ not equivalent to $HC_n$''.
    
\State Set $P'(x_{1,2}) := x_{1,2}$, $P'(x_{2,1}) := x_{i,j}$ and $P'(x_{1,t}) := x_{i_t,j_t}$, where $x_{i_t,j_t} \in T'_{1,2}$ and $t \in [3,n]$.

\State For all $x_{i_t,j_t} \in T'_{1,2}$, compute $R'_{i_t,j_t}$ as in Step 2 and partition it. Let $\{x_{k_t,\ell_t}\}$ be the singleton set in the partitioning. Set $P'(x_{t,1}) := x_{k_t,\ell_t}$. Abort and output ``$f$ not equivalent to $HC_n$'' if the partition does not exist. 
        
\State For $a,b \in [2,n]$, $a \neq b$, compute $R'_{k_a,\ell_a} \cap R'_{{i_b,j_b}}$, where $P'(x_{a,1}) = x_{k_a,\ell_a}$ and $P'(x_{1,b}) = x_{i_b,j_b}$, as computed in earlier steps. If $R'_{k_a,\ell_a} \cap R'_{{i_b,j_b}}$ is a singleton $ \{x_{c,d}\}$,  set $P'(x_{a,b}) := x_{c,d}$, else abort and output ``$f$ not equivalent to $HC_n$''.

\State Output $P'$.
\end{algorithmic}
\end{algorithm}

\subsection{S-equivalence test} \label{subsec-Sequiv-test}
We assume $n \neq 4$. Suppose $f(\vecx) = HC_n(S\vecx)$, where $S \in GL_{n^2-n}(\F)$ is a scaling matrix. Algorithm \ref{alg3} gives the scaling equivalence test over $\F_q$. The same algorithm, with some changes, will also work over other fields $\F$, see Appendix \ref{proof-Sequiv-allfields} for details. The correctness of Step 1 follows from Claim \ref{claim-Ssym-assum}, while that of the remaining steps follows from Proposition \ref{prop-stest-correct}. In Appendix \ref{subsec-Sequiv-HC4}, we show an S-equivalence test for $HC_4$ over $\Q$, $\R$, and finite fields.

\begin{claim} \label{claim-Ssym-assum}
    Without loss of generality, $S'_{1,j} = 1$, $S'_{2,3} = 1$ and $S'_{i,1} = 1$, where $i \in [2,n-1], j \in [2,n]$. 
\end{claim}
\begin{proposition} \label{prop-stest-correct}
    If $f(\vecx) = HC_n(S\vecx)$, then in deterministic $(n \log q)^{O(1)}$ time Algorithm \ref{alg3} outputs an $S'$ such that $S'S$ is a scaling symmetry of $HC_n$.
\end{proposition}

\begin{algorithm} 
\caption{S-equivalence test} \label{alg3}
\textbf{Input}: B.b.a to $f(\vecx)$. 
    
\textbf{Output}: $S' \in \GL_{n^2-n}(\F_q)$ s.t. $HC_n(S'\vecx) = f(\vecx)$ and $S' = \tilde{S}S$, where $\tilde{S}$ is a scaling symmetry, if $f(\vecx) = HC_n(S\vecx)$. 

\begin{algorithmic}[1]
\State Set $S'_{1,j} = 1$, $S'_{2,3} = 1$ and $S'_{i,1} = 1$, where $i \in [2,n-1], j \in [2,n]$.

\State Query $f$ to obtain the coefficients of monomials corresponding to the rows $\sigma_1, \dots ,\sigma_{(n-1)(n-2)}$ of the matrix $M^{(n)}$ constructed in Proposition \ref{prop-lie-dim-lb}. Let $c_k$ be the $k$'th coefficient. If $c_k = 0$ for some $k$, then abort and ouput ``$f$ not equivalent to $HC_n$''. 

\State Let $T := \{(i,j) | \ (i,j) \neq (1,k), \ k \in [2,n] \text{ and } (i,j) \neq (\ell,1), \ \ell \in [2,n-1] \text{ and }(i,j) \neq (2,3) \}$. Consider the matrix $M := M^{(n)}_{\bullet \times T}$. Compute $\beta = (det(M))^{-1}$ in $\Z_{q-1}$.

\State For each row $\sigma_k$, $k \in [(n-1)(n-2)]$ and $(i,j) \in T$, compute the cofactor of $M$ with respect to $\sigma_k$ and $(i,j)$, and denote it as $\alpha_{k,(i,j)}$. Set $S'_{i,j} = \prod_{k = 1}^{(n-1)(n-2)}c^{\beta \alpha_{k,(i,j)} \text{ mod } q-1}_{k}$.
\State Output $S'$.
\end{algorithmic}
\end{algorithm}

\section{Acknowledgements}
The author thanks Chandan Saha for suggesting the problem and Abhiram Aravind for going through the proofs, which helped improve the presentation of the paper. The author also thanks the anonymous reviewers of MFCS 2026 for their constructive and informative feedback, particularly one reviewer for pointing to the Pascal Determinant polynomial family, which is VNP-complete and for which an ET algorithm can be designed using knowledge of its continuous symmetries.
\bibliographystyle{alpha}
\bibliography{bibl}

\appendix
\addtocontents{toc}{\protect\setcounter{tocdepth}{1}}
\section{Comparison with \cite{Kayal12} and \cite{GuptaS19}} \label{sec-prevwork-comp}
Lie Algebra has been used to develop ET algorithms for many polynomial families; refer to Appendix \ref{sec-related} for an overview. As our techniques are based on \cite{Kayal12,GuptaS19}, we compare our techniques with theirs. Since $HC_n$ is a linear combination of the Immanant polynomials \cite{Merris83,Burgisser2000}, and the symmetries and Lie algebra of Immanants are known \cite{Ye11}, it is possible that the symmetries of $HC_n$ can be studied via those of the Immanant polynomials. However, the definition of the Immanant polynomial is representation theoretic, and the proof technique of \cite{Ye11} also involves some representation theory, while $HC_n$ has a simpler and intuitive definition, and our proof technique does \emph{not} need representation theory. Moreover, $HC_n$ is a linear combination of the Immanant polynomials as long as the characteristic of the underlying field does \emph{not} divide $n$, thus the analysis of $HC_n$ via this linear combination may not hold over all fields. It is also unclear whether the proof of \cite{Ye11} will directly yield an ET algorithm or hold for $HC$, because at certain points \emph{non-cyclic} permutations are needed in their proof.

\paragraph{Comparison with \cite{Kayal12}.} The author of \cite{Kayal12} uses the symmetries of the Permanent polynomial to describe an explicit basis for the Lie algebra of the Permanent polynomial. In contrast, the symmetries and the Lie Algebra of $HC_n$ have not been studied before to our knowledge, so we analyse the Lie Algebra first and use it to argue about the symmetries. We use ideas from \cite{GuptaS19} to establish an explicit basis for the Lie algebra of $HC_n$, see \nameref{para-gs19comp} for more details.

In Proposition \ref{prop-lie-alg-struct}, we give a basis for $\mathfrak{g}_{HC_n}$ over \emph{any} field, except for $n=4$. All but one of the elements of this basis are similar to those of the basis for $\mathfrak{g}_{Perm_n}$, the Lie algebra of the Permanent polynomial, given in Proposition 39 in \cite{Kayal12}. The single differing basis element for $HC_n$ ensures the linear independence of the basis elements over \emph{every} field. In contrast, the basis given for $\mathfrak{g}_{Perm_n}$ is a basis as long as the characteristic of the underlying field does \emph{not} divide $n$, the degree of the Permanent polynomial. For $HC_4$, we give a basis over characteristic $2$ fields separately from other fields. This exception arises because $\mathfrak{g}_{HC_4}$ is determined by the \emph{nullspace} of $6$ linear equations in $12$ variables, each corresponding to a permutation in $C_4$, such that all the equations are linearly dependent over characteristic $2$ fields. Five of the linear equations are linearly independent over any field, implying that the sixth equation is a linear combination of the five equations over characteristic $2$ fields. This causes $\mathfrak{g}_{HC_4}$ to be $6$ dimensional over fields of characteristic other than $2$ and $7$ dimensional over characteristic $2$ fields. Using this, we show that $HC_4$ has discrete symmetries over $\Q$, $\R$, $\C$ and 
certain finite fields. See Appendix \ref{sec-HC4-results} for details.

Our implementation of Step \ref{item-Sequivred} differs from that for the Permanent polynomial because the set of vanishing second-order partial derivatives of $HC_n$ differs from that for the Permanent polynomial. As mentioned earlier, our execution of Step \ref{item-Srecov} for $HC_n$ completely differs from that for the Permanent. For the Permanent, there are sufficiently many pairs of monomials such that for each pair, the monomials differ in exactly two variables. For example, the monomials $x_{1,1}x_{2,2}x_{3,3}\dots x_{n,n}$ and $x_{1,2}x_{2,1}x_{3,3}\dots x_{n,n}$ share $n-2$ variables. Assuming that $x_{1,1},x_{1,2}$ and $x_{2,1}$ are unscaled, which follows from the scaling symmetries of the Permanent, we query the scaled Permanent polynomial at these two monomials, which gives two scalars $\lambda_1$ and $\lambda_2$. Then, $\frac{\lambda_1}{\lambda_2}$ is the factor by which $x_{2,2}$ is scaled. In this way, $S$ can be recovered from black-box access to a scaled Permanent polynomial. In contrast, any two monomials of $HC_n$ differ in at least $3$ variables (see Claim \ref{claim: cycle-disjointness}). It is then unclear whether Step \ref{item-Srecov} can be performed in a similar way to the Permanent polynomial. Instead, we implement Step \ref{item-Srecov} by solving a system of linear equations as was done in \cite{GuptaS19}. For the Permanent polynomial, Step \ref{item-verify} can be performed via downward self-reducibility \cite{KI04}, as was done in \cite{Kayal12}, or by using circuit identities arising from the characterisation by symmetries of the Permanent, see Page 81, Proof of Theorem 3.2.1 in \cite{grochowPhD}. The latter identities also yield a \emph{more} efficient circuit testing algorithm. In contrast, though $HC_n$ is not characterised by its symmetries, we use its downward self-reducibility to verify if the given input is equivalent to $HC_n$ under the recovered transform.

\paragraph{Comparison with \cite{GuptaS19}.} \label{para-gs19comp} The authors of \cite{GuptaS19} studied the Lie Algebra and symmetries of the Nisan-Wigderson polynomial, denoted NW, and gave an ET algorithm for NW in a special case where $A$ is a block-permuted permutation scaling matrix. The monomials of NW correspond to the evaluations of a low-degree univariate polynomial over a finite field, while those of $HC_n$ correspond to cyclic permutations. This leads to many structural differences between the symmetries and the Lie algebras of $HC_n$ and NW.

The main idea we use from \cite{GuptaS19} is their argument for establishing a basis for $\mathfrak{g}_{NW}$, the Lie algebra of NW. They argue an upper bound on the nullity of a matrix $M$, where the rows of $M$ are the linear equations that determine $\mathfrak{g}_{NW}$, and the nullspace of $M$ is $\mathfrak{g}_{NW}$. To establish this upper bound, they lower bound the rank of $M$ by constructing an appropriate full rank submatrix $M'$ of $M$. Then, they show an explicit set of linearly independent vectors that lies in the nullspace of $M$ and has the same size as the upper bound, thus getting a basis of $\mathfrak{g}_{NW}$. They describe $M'$ explicitly and use it to describe all scaling symmetries of NW over $\Q$, $\R$, and certain finite fields, and give a scaling equivalence test for NW over $\R$ and certain finite fields. 

While we follow the same idea at a high level, the crucial difference is that our argument for $HC_n$ holds over \emph{any} field, as shown in Proposition \ref{prop-lie-basis}. This is in contrast to the case of NW, where the nullity upper bound argument holds as long as the characteristic of the field does not divide $d$, the degree of NW. While the authors of \cite{GuptaS19} describe $M'$ explicitly, we construct $M'$ efficiently via induction in our case, as shown in Proposition \ref{prop-lie-dim-lb}. We also leverage $M'$ to describe all scaling symmetries of $HC_n$ and give a scaling equivalence test for it over \emph{all} fields.

\section{Other related works} \label{sec-related}
\subsection{Results on PE} \label{sec-PEres}
As stated in Section \ref{Section: Introduction}, the complexity of PE is not well understood. PE is known to be in PSPACE over $\C$, in EEXP over $\R$, in NP $\cap$ co-AM over finite fields, and over $\Q$ it is not even known to be decidable \cite{Saxenaphd}. Even when both the inputs are restricted, limited results are known. The variant in which both inputs are degree two polynomials is called Quadratic form equivalence (QFE). Due to well-known classification results for such polynomials \cite{Lam04,Ara11}, QFE has polynomial time algorithms over $\R$, $\C$, finite fields, and $\Q$ (assuming access to an integer factoring oracle) \cite{Saxenaphd,Wallenborn13}. In contrast, Cubic form equivalence (CFE), where both inputs are degree three polynomials, is at least as hard as graph isomorphism \cite{AgrawalS05,Kayal11}. Further, CFE is polynomial time equivalent to isomorphism/equivalence problems like algebra isomorphism, trilinear form-equivalence, etc., as shown by \cite{GrochowQ23a}. Recently, in \cite{GuptaST23}, the authors studied PE for orbits of ROFs, a significant generalization of QFE, and gave a randomised polynomial time algorithm for PE for orbits of additive-constant-free ROFs, a mild restriction of general ROFs.

A variant of PE is the Shift Equivalence Test problem (SET), as referred to by \cite{DvirOS14}, where we have to decide for given $f,g \in \F[\vecx]$ whether $f(\vecx) = g(\vecx + \vecb)$ for some $\vecb \in \F^{|\vecx|}$. The author of \cite{Grigoriev97} showed that when $f$ and $g$ are given verbosely as a list of coefficients of all $\binom{n+d}{d}$ monomials, there is a deterministic algorithm over characteristic $0$ fields, a randomised algorithm over prime residue fields and a quantum algorithm for characteristic $2$ fields, all with running time polynomial in $\binom{n+d}{d}$. Later, \cite{DvirOS14} gave a randomised $(nds)^{O(1)}$ algorithm when $f$ and $g$ are given as black-boxes with degree bound $d$ and circuit size bound $s$. The author of \cite{Kayal12} also gave a randomised polynomial-time algorithm for SET; see Theorem 28 in the paper, and Appendix B of \cite{KayalNST17} for another proof.

In \cite{BlaserRS17}, the authors studied the Scaling equivalence problem, yet another variant of PE, which involves checking for two given polynomials $f$ and $g$ whether there exists a scaling matrix $S \in \GL_{|\vecx|}(\F)$ such that $f(\vecx) = g(S\vecx)$.  They gave a randomised polynomial-time algorithm for the scaling equivalence problem over $\R$ when the inputs are given as black boxes. 

\subsection{Results on ET} \label{sec-ETres}
As mentioned in Section \ref{Section: Introduction}, the study of ET for polynomial families was initiated in \cite{Kayal11} where randomised polynomial time ET algorithms were given for the Power Symmetric and the Elementary Symmetric polynomials. Following this, efficient ET algorithms were given for several other important polynomial families and circuit classes, see Table \ref{tab:ET-res}. The authors of \cite{KayalNS19} used ET for the determinant to give an efficient average-case reconstruction algorithm for low-width Algebraic Branching Programs (ABPs). Thus, ET algorithms have been used to design efficient reconstruction algorithms. 

It is clear from Table \ref{tab:ET-res} that studying ET for polynomial families has led to efficient algorithms in every case so far. In contrast, read-once formulas (ROFs) and $t$-design polynomials for low $t$, a subclass of sparse polynomials (depth-2 circuits), are the only circuit classes for which we have efficient ET algorithms so far. Sparse polynomials form a natural circuit class for which ET was shown to be NP-hard \cite{BDSS24}. Later, the authors of \cite{BDGT24} showed the NP-hardness of testing permutation equivalence to read-once oblivious algebraic branching programs (ROABPs). The authors of \cite{RS25} adapted the techniques of \cite{BDSS24,BDGT24} to show ET for ROABPs is NP-hard. 

\begin{table}
    \centering
    \begin{tabular}{|c|c|c|}
        \hline
       &   \makecell{\textbf{\emph{Algorithmic results}} \\ \textbf{\emph{(randomised polynomial time)}}}   & \\
            \hline
         \textbf{\emph{Family/Circuit class}} & \textbf{\emph{Technique}} & \textbf{\emph{Reference}} \\
         \hline
        \hline
         Determinant & \text{Lie Algebra} & \cite{Kayal12,grochowPhD,GargGK019} \\
         \hline
         Permanent  &  \text{Lie Algebra}& \cite{Kayal12}  \\
         \hline
         IMM  &  \text{Lie Algebra} & \cite{KayalNST17} \\ 
         \hline
         Trace-IMM  &  \text{Lie Algebra} & \cite{MurthyNS20} \\
         \hline
         Sum-Product  &  \makecell{Hessian,Lie Algebra, \\ Vector Space Decomposition} & \cite{Kayal11,GuptaST23,BDS24,MS21}  \\
         \hline
         Power Symmetric  &  \text{Hessian}&  \cite{Kayal11} \\
         \hline
         Elementary Symmetric  & \text{Second-order partial derivative} & \cite{Kayal11}  \\
         \hline
         Vandermonde  &  \text{Analysis of product of linear forms} & \cite{RR19} \\
         \hline
         Continuant & \makecell{Interpolating sets, Directional derivatives \\ analysis of product of linear forms} &  \cite{MS21} \\
         \hline
         Design Polynomials  & \text{Lie Algebra, Vector Space Decomposition} &  \cite{GuptaS19,BDS24} \\
         \hline
         ROF & \text{Hessian} & \cite{GuptaST23} \\
         \hline
         Hamiltonian Cycle polynomial & \text{Lie Algebra} & \textbf{This work}\\
        \hline \hline
       &  \textbf{\emph{Hardness results}} & \\
                    \hline
         \text{Sparse polynomials} & NP-hardness  & \cite{BDSS24}\\
         \hline
     ROABP  & NP-hardness & \cite{BDGT24,RS25} \\    
    \hline
    \end{tabular}
    \caption{A list of ET results}
    \label{tab:ET-res}
\end{table}
\paragraph{Other hardness results.}The author of \cite{Kayal12} showed that the more general problem of checking if a polynomial is an affine projection\footnote{An $n$-variate polynomial $f(\vecx)$ is an affine projection of $m$-variate $g(\vecy)$ if there exist an $A \in \F^{m \times n}$ and a $\vecb \in \F^{m}$ such that $f(\vecx) = g(A\vecx + \vecb)$.} of another polynomial is $\NP$-hard via a reduction from Graph $3$-Colorability. More recently, the authors of \cite{BDJ25} improved upon the result in \cite{Kayal12} by showing a reduction from polynomial solvability over any field to checking affine equivalence of polynomials. They also showed that the SparseShift problem\footnote{The problem involves checking if for a given $f(\vecx)$ whether $f(\vecx + \vecb)$ has strictly fewer monomials than $f$ for some shift $\vecb \in \F^{|\vecx|}$.}, a variant of the Shift Equivalence Testing problem, is at least as hard as polynomial solvability over any integral domain including fields. This result is an improvement over \cite{ChillaraGS23}, where the authors showed the same result but over integral domains which are \emph{not} fields. 

\subsection{Related works involving $HC$} \label{sec-HCnmisc}

\paragraph{Connections and comparisons between Permanent and $HC$.} The authors of \cite{GJ81} showed an identity expressing $HC_n$ as a sum of the product of the evaluation of the Permanent and the Determinant polynomials at various submatrices of $X_n$; see \cite{SG25} for an accessible proof. The author of \cite{Merris83} showed that $HC_n$ can also be written as a linear combination of Immanant polynomials, which are a generalization of the Permanent and the Determinant polynomials. In \cite{Burgisser2000}, the author showed that $HC_{n-2}$ can be written as a rational linear combination of $p$-projections of $HC_n$ and the Determinant polynomial.\footnote{An $n$-variate polynomial $f(\vecx)$ is a $p$-projection of $m$-variate $g(\vecy)$ if $m = n^{O(1)}$ and there exists a mapping $\phi: \vecy \mapsto \F \sqcup \vecx$ such that $g(\phi(\vecy)) = f(\vecx)$. The projection is monotone if the scalars lie in $\F^{\geq 0}$, the set of non-negative elements of $\F$.}

The Permanent polynomial and $HC_n$ have been studied with respect to lower bounds and upper bounds, with some results as presented in Table \ref{tab:comparison-lb}. The authors of \cite{JS82} studied the computation of polynomials by circuits over semi-rings\footnote{A semi-ring $(S,+,\times,0,1)$ comprises a set $S$ such that $(S,+,0)$ and $(S,\times,1)$ are commutative monoids, $\times$ distributes over $+$, and $a\times 0 $ is $0$ for all $a \in S$.}, which led to monotone arithmetic circuit lower bounds. They showed a $2^{\Omega(n)}$ monotone arithmetic circuit lower bound for $HC_n$ and the Permanent over the semi-ring $(\R^{\geq 0},+,\times,0,1)$. They also showed that for a polynomial $f$, monotone arithmetic circuit lower bounds over the Boolean semi-ring $\B = (\{0,1\},\lor,\land,0,1)$ imply monotone arithmetic circuit lower bounds over $(\R^{\geq 0},+,\times,0,1)$. Note, monotone Boolean circuit lower bounds for the Boolean function defined by $f$ imply monotone arithmetic circuit lower bounds for $f$ over $\B$.

\begin{table}
    \centering
    \begin{tabular}{|c|c|c|}
        \hline
          \textbf{\emph{Result}}  & \textbf{\emph{Permanent}} & $HC_n$ \\
        \hline
        \hline
        Monotone circuit lower bound  & $n(2^{n-1}-1)$ \cite{JS82}  & $(n-1)((n-2)2^{n-3} + 1)$ \cite{JS82} \\
        \hline
        Tropical circuit lower bound & $2^{\Omega(n)}$ \cite{Jukna15} & $2^{\Omega(n)}$ \cite{Jukna15} \\
        \hline
         Monotone Boolean circuit lower bound & $2^{n^{1/3 - o(1)}}$\cite{CGRSS26} & $2^{\tilde{\Omega}(n^{1/4})}$ \cite{BM25} \\
        \hline
         Determinantal complexity lower bound & $\frac{n^2}{2}$ \cite{MR04} &  \textbf{Unknown}\\
        \hline
        \makecell{Binary determinantal complexity \\ upper bound} & $2^{n} - 1$ \cite{Gre11} & $(n-1)2^{n-2}+1$ \cite{HI16}\\
        \hline
    \end{tabular}
    \caption{Some bounds for the Permanent and $HC_n$}
    \label{tab:comparison-lb}
\end{table}

A monotone Boolean circuit lower bound for the Hamiltonian Cycle function follows from the fact that the Clique function is a monotone Boolean projection\footnote{An $n$-variate monotone Boolean function $f(\vecx)$ is a monotone-projection of the $m$-variate monotone Boolean function $g(\vecy)$ if there exists a mapping $\phi: \vecy \mapsto \{0,1\}\sqcup \vecx$ such that $g(\phi(\vecy)) = f(\vecx)$.} of $HC_n$, as observed by \cite{AB87,Val79a}, and such lower bounds for the Clique function have been studied \cite{AB87,CKR22,BM25}, with the most recent lower bound being $2^{\tilde{\Omega}(n^{1/2})}$. Prior to the improved monotone Boolean circuit lower bound for the Permanent function by \cite{CGRSS26}, the best known bound was $n^{\Omega(\log n)}$ by \cite{Razborov85}, which also implied a $n^{\Omega(\log n)}$ monotone arithmetic circuit lower bound for the Permanent polynomial over $\B$. Before this improvement, \cite{Jukna14} posed the question whether $HC$ is a monotone $p$-projection of the Permanent. If the answer were yes, then we would get a $2^{n^{\Omega(1)}}$ monotone arithmetic circuit lower bound for the Permanent polynomial over $\B$, because the Clique polynomial family is a monotone $p$-projection of $HC$ \cite{Val79a}. However, the author of \cite{Gro17} showed that $HC$ is \emph{not} a monotone $p$-projection of the Permanent over $\B$, $\R$ and other semi-rings, which also shows that the Permanent is \emph{not} VNP-complete under monotone $p$-projections over $\R^{\geq 0}$ for polynomials with non-negative coefficients.\footnote{Later, \cite{MS18} showed that even the Clique polynomial is not a monotone affine projection of the Permanent polynomial.} It is \emph{unknown} whether $HC$ is VNP-complete under monotone $p$-projections over $\R$ for such polynomials. 

The lower bounds established by \cite{JS82} for the Permanent polynomial and $HC_n$ also hold over the tropical semi-rings $(\R,\text{min},+,+\infty,0)$ and $(\R^{\geq 0},\text{min},+,+\infty,0)$. Circuits over tropical semi-rings are called tropical circuits. The author of \cite{Jukna15} studied tropical circuits, motivated by the observation that many dynamic programming algorithms correspond to tropical circuits over tropical semi-rings. He showed that the power of tropical circuits lies between that of monotone Boolean circuits and monotone arithmetic circuits, and also showed $2^{\Omega(n)}$ tropical circuit lower bounds for the Permanent and $HC_n$, among other polynomials, over the tropical semi-rings $(\N,\text{min},+,+\infty,0)$ and $(\N,\text{max},+,0,0)$, where the former is a proper sub semi-ring of $(\R,\text{min},+,+\infty,0)$.

In \cite{Malod03}, the classes VP$^{0}$ and VNP$^{0}$ were defined using constant-free circuits, where the only constants that appear as inputs are $1$, $0$ or $-1$. In the work, $HC$ was shown to be VNP$^{0}$-complete, which also implies $HC$ is VNP-complete over \emph{any} ring. The VP$^{0}$ vs VNP$^{0}$ question (equivalently, is $HC$ in VP$^{0}$?) is connected to the $\tau$-conjecture concerning the number of integer roots of univariate integer polynomials, see \cite{Burg24} for an overview. With this context, showing VP $\neq$ VNP (Valiant's conjecture) implies a superpolynomial lower bound on the determinantal complexity of the Permanent polynomial, and VP$^{0}$ $\neq$ VNP$^{0}$ will show a superpolynomial lower bound on the binary determinantal complexity of $HC_n$.\footnote{ The determinantal complexity of a polynomial $f$, dc$(f)$, is the smallest $m \in \N$ such that $f = det(M)$, where $M$ is a $m \times m$ matrix with entries as affine forms. The binary determinantal complexity, bdc$(f)$, is defined similarly, with $M$ being a matrix with entries as $0$, $1$ or variables. Clearly, bdc$(f) \geq$ dc$(f)$.} Table \ref{tab:comparison-lb} shows that the Permanent polynomial has a quadratic determinantal complexity lower bound, while for $HC_n$ \emph{no} such lower bound is known.

\paragraph{Usefulness of $HC$ as a VNP-complete family.} In various works \cite{Malod07,KoiranPTT15,GMQ16,IM18,IS22}, $HC$ has been the primary choice of a VNP-complete family over all fields for generalizing results to all fields.\footnote{In \cite{GMQ16}, the authors work over algebraically closed fields, and using $HC$, their result generalizes to all closed fields.} The author of \cite{Hrubes16} noted the lack of VNP-complete families over characteristic $2$ fields and showed VNP-completeness of multiple graph-based polynomial families over such fields. He used $HC$ to show the completeness of Clique$^{*}$ polynomial family, a variant of the Clique polynomial family known to be VNP-complete over fields of characteristic other than $2$ \cite{Burgisser2000}, over characteristic $2$ fields. He then showed that four other graph-based polynomial families are VNP-complete over all fields using the Clique$^{*}$ polynomial family. Thus, the VNP-completeness of all five polynomial families relies on that of $HC$.

In \cite{DRS22}, the authors used $HC$ to show that a variant of the monomial prediction problem, where in this variant the input polynomial is a composition of ``easy'' functions, is $\oplus$P-complete over finite fields and NP-hard over $\Z$. The motivation for studying this variant comes from cryptography. They also adapt the proof to show the $\#$P-completeness of a problem motivated by machine learning.

\section{Missing Proofs from Section \ref{Scn: Lie_alg}} \label{proofs-lie}
\subsection{Proof of Proposition \ref{prop-lie-alg-struct}} \label{proof-lie-alg-struct}
We prove the forward direction first. If $A \in \mathfrak{g}_{HC_n}$, then, by definition, $A$ satisfies 
    \begin{equation} \label{eqn: HCn-lie}
     \sum_{\substack{i_1,j_1,i_2,j_2 \in [n],  \\ i_1 \neq j_1, i_2 \neq j_2}} A_{(i_1,j_1),(i_2,j_2)}x_{i_2,j_2}\frac{\partial HC_n}{\partial x_{i_1,j_1}}  = 0.   
    \end{equation}

    Let $m_{\sigma} := \prod_{i \in [n]}x_{i,\sigma(i)}$ for all $\sigma \in C_n$. By linearity of the derivative operator, it suffices to analyse the monomials generated from $m_\sigma$ in \eqref{eqn: HCn-lie}.
    For any $\sigma$, we have that,
    \begin{equation} \label{eqn:deriv-equal}
    \begin{split}
        x_{i_1,j_1} \frac{\partial m_{\sigma}}{\partial x_{i_2,j_2}} = x_{i_3,j_3} \frac{\partial m_{\sigma}}{\partial x_{i_4,j_4}} \iff
        (i_1,j_1) = (i_2,j_2) \text{ and } (i_3,j_3) = (i_4,j_4).
    \end{split}            
    \end{equation}
    assuming $(i_2,j_2) \neq (i_4,j_4)$ and the derivatives are non-zero.  From Observation \ref{obs: disj-deriv}, we get that the coefficient of the monomial $x_{i_1,j_1} \frac{\partial m_{\sigma}}{\partial x_{i_2,j_2}}$, where $(i_1,j_1) \neq (i_2,j_2)$ and $j_2 = \sigma(i_2)$, in \eqref{eqn: HCn-lie} is just $A_{(i_1,j_1),(i_2,j_2)}$. From \eqref{eqn:deriv-equal} and Observation \ref{obs: disj-deriv}, the coefficient of $m_{\sigma}$ is $\sum_{i \in [n]}A_{(i,\sigma(i)),(i,\sigma(i))}$. Thus, we get that all the off-diagonal entries $A_{(i_1,j_1),(i_2,j_2)}$ are $0$ while the diagonal entries $A_{(i,j),(i,j)}$ satisfy  
    
    \[\sum_{i \in [n]}A_{(i,\sigma(i)),(i,\sigma(i))} = 0 \quad \forall  \sigma \in C_n.\]
    The reverse direction can be verified easily by using Observation \ref{claim: cycle-disjointness} and the condition in \eqref{eqn:deriv-equal}.

\subsection{Proof of Observation \ref{obs: disj-deriv}} \label{proof-disj-deriv}
    We assume throughout the proof that the derivatives we consider are non-zero. For $n = 3$, $|C_3| = 2$ with $\sigma_1 = (1 \ 2 \ 3)$ and $\sigma_2 = (1 \ 2 \ 3)$. Then, we have
        \begin{equation*}
        \begin{split}
            x_{i_1,j_1} \frac{\partial x_{1,2}x_{2,3}x_{3,1}}{\partial x_{i_2,j_2}} &= x_{i_3,j_3} \frac{\partial x_{1,3}x_{3,2}x_{2,1}}{\partial x_{i_3,j_3}} \\
            \iff x_{i_1,j_1} x_{i_4,j_4} x_{1,2}x_{2,3}x_{3,1} &= x_{i_2,j_2} x_{i_3,j_3} x_{1,3}x_{3,2}x_{2,1}. 
        \end{split}
        \end{equation*}
    It can be verified that the last equality cannot hold for any choice of $x_{i,j}$ variables, proving the observation statement for $n=3$. Now assume $n > 3$. As $\sigma_1 \neq \sigma_2$, by Claim \ref{claim: cycle-disjointness}, there exists $S = \{k_1,k_2,k_3\} \subseteq [n]$ such that $\sigma_1(k_{\ell}) \neq \sigma_2(k_{\ell})$ where $\ell \in [3]$. Consider the following cases:
    \begin{enumerate}
        \item $i_2 \in S$ and $i_4 \in S$: In this case, we have that

        \begin{equation*}
            x_{i_1,j_1} \frac{\partial m_{\sigma_1}}{\partial x_{i_2,j_2}} = x_{i_1,j_1}\frac{\prod_{i \in S}x_{i,\sigma_1(i)}}{x_{i_2,j_2}} \dot \prod_{i \notin S}x_{i,\sigma_1(i)}
        \end{equation*}
        and
        \begin{equation*}
            x_{i_3,j_3} \frac{\partial m_{\sigma_2}}{\partial x_{i_4,j_4}} = x_{i_3,j_3}\frac{\prod_{i \in S}x_{i,\sigma_2(i)}}{x_{i_4,j_4}} \dot \prod_{i \notin S}x_{i,\sigma_2(i)}
        \end{equation*}
        
        \item $i_2 \in S$ and $i_4 \notin S$: In this case we have,
        \begin{equation*}
            x_{i_1,j_1} \frac{\partial m_{\sigma_1}}{\partial x_{i_2,j_2}} = x_{i_1,j_1} \frac{\prod_{i \in S}x_{i,\sigma_1(i)}}{x_{i_2,j_2}} \prod_{i \notin S}x_{i,\sigma_1(i)}
        \end{equation*}
        and
        \begin{equation*}
            x_{i_3,j_3} \frac{\partial m_{\sigma_2}}{\partial x_{i_4,j_4}} = x_{i_3,j_3}\prod_{i \in S}x_{i,\sigma_2(i)}\frac{\prod_{i \notin S}x_{i,\sigma_2(i)}}{x_{i_4,j_4}} 
        \end{equation*}

        Note that the case $i_2 \notin S$ and $i_4 \in S$ is the same as this case with the appropriate changes in the resulting monomials.
        
        \item $i_2 \notin S$ and $i_4 \notin S$
            \begin{equation*}
            x_{i_1,j_1} \frac{\partial m_{\sigma_1}}{\partial x_{i_2,j_2}} = x_{i_1,j_1}  \prod_{i \in S}x_{i,\sigma_1(i)}\frac{\prod_{i \notin S}x_{i,\sigma_1(i)}}{x_{i_2,j_2}}
        \end{equation*}
        and
        \begin{equation*}
            x_{i_3,j_3} \frac{\partial m_{\sigma_2}}{\partial x_{i_4,j_4}} = x_{i_3,j_3}\prod_{i \in S}x_{i,\sigma_2(i)}\frac{\prod_{i \notin S}x_{i,\sigma_2(i)}}{x_{i_4,j_4}} 
        \end{equation*}
        
        By the unique factorization theorem for multivariate polynomials, if the resulting monomials were equal, then their factors must be the same. However, as $|S| = 3$, it is not hard to observe in all of the above cases that regardless of the choice of $x_{i_3,j_3}$, there will always exist a $k \in S$, such that $x_{k,\sigma_1(k)}$ is not divisible by any $x_{\ell,\sigma_2(\ell)}$ for $\ell \in [n]$ or $x_{i_3,j_3}$. Similarly, regardless of any choice of $x_{i_1,j_1}$, there exists a $k' \in S$ such that $x_{k',\sigma_2(k')}$ is not divisible by any $x_{\ell,\sigma_1(\ell)}$ for $\ell \in [n]$ or $x_{i_1,j_1}$.
    \end{enumerate}
\subsection{Proof of Claim \ref{claim: cycle-disjointness}} \label{proof-cycle-disjoint}
For $n=3$, there are only two cyclic permutations $(1 \ 2 \ 3)$ and $(1 \ 3 \ 2)$. It is then easy to see that the claim holds for $n=3$ with $S = \{1,2,3\}$. Hence, we assume $n \geq 4$. As $\sigma_1 \neq \sigma_2$ and both are $n$-cycles, we can write $\sigma_1$ and $\sigma_2$ as

\begin{equation*}
\begin{split}
    \sigma_1 : (1 \ m_1 \ m_2 \ \dots \ m_i \ m_{i+1} \ \dots \ m_n), \\
    \sigma_2 : (1 \ m_1 \ m_2 \ \dots \ m_i \ m'_{i+1} \ \dots \ m'_n),   
\end{split}
\end{equation*}
where $i < n-1$ (otherwise $\sigma_1 = \sigma_2$), $m_j = \sigma_1^{j-1}(1)$ for $j \in [n]$, $m'_j = \sigma_2^{j-1}(1)$ for $j \in [i+1,n]$, with $m'_{i+1} \neq m_{i+1}$ while $m'_j = m_j$ for all $j \in [i]$. Thus, $\sigma_1({m_i}) \neq \sigma_2(m_i)$.

As $\sigma_2$ is an $n$-cycle, there exists $k_2 \in [i+1, n-1]$ such that $\sigma_2(m'_{k_2}) = m_{i+1}$. As $\sigma_1$ is an $n$-cycle, there exists $k_1 \in [i+2,n]$ such that $m_{k_1} = m'_{k_2}$ and $\sigma_1(m_{k_1}) \neq m_{i+1}$. Thus, $\sigma_1(m_{k_1}) \neq \sigma_2(m_{k_1})$. 

We now further leverage the fact that $\sigma_1,\sigma_2$ are $n$-cycles to find an $m_j$ such that $\sigma_1(m_j) \neq \sigma_2(m_j)$. Consider the following cases:

\begin{enumerate}
    \item $m'_{i+1} \neq m_{k_1+1}$: This implies $\sigma_{1}(m_{k_1}) \neq m'_{i+1}$. As $\sigma_1$ and $\sigma_2$ are permutations, there exists a $j_1 \in [i+1,n-1] \backslash \{k_1\}$ such that $\sigma_1(m_{j_1}) = m'_{i+1}$ and $\sigma_2(m_{j_1}) \neq m'_{i+1}$.

    \item $m'_{i+1} = m_{k_1+1}$: In this case, we show that there exists a $j_1 \in [i+1,k_1-1]$ such that $\sigma_1(m_{j_1}) \neq \sigma_2(m_{j_1})$. 
    
    Suppose to the contrary that there did not exist such a $j_1$, then $\sigma_1(m_j) = \sigma_2(m_j)$ for all $j \in [i+1,k_1-1]$. Since $m'_{k_2+1} = \sigma_2(m'_{k_2}) = m_{i+1}$, we also get from the previous assumption that $m_{i + j} = m'_{k_2+j}$ for $j \in [1,k_1-i]$. Based on this, we have the following subcases, each of which leads to a contradiction. Thus, the claimed $j_1$ exists.
    \begin{enumerate}
        \item $k_1 - i > n - k_2$: Taking $j = n-k_2+1$, we get that $m_{i+n-k_2+1} = m'_{k_2+n-k_2+1} = 1$. This implies $\sigma_1^{i+n-k_2}(1) = 1$ and hence $i = k_2 \mod n$, a contradiction as $i < k_2 < n$. 

        \item $k_1 - i < n - k_2$: Taking $j = k_1-i$, we get that $m'_{k_2+k_1-i} = m_{k_1}$. As $m_{k_1} = m'_{k_2}$, we get that, $m'_{k_2+k_1-i} = m'_{k_2}$. This implies $\sigma_2^{k_2+k_1-i-1}(1) = \sigma_2^{k_2 - 1}(1)$ and hence $i = k_1 \mod n$, a contradiction as $0 < k_1 - i < n$. 

        \item $k_1 - i = n - k_2$: Taking $j = k_1-i$, we get that $m_{i+k_1-i} = m'_{k_2+k_1-i}$ implying $m_{k_1} = m'_{k_2 + k_1 - i} = m'_{k_2 + n - k_2} = m'_n$. Thus, we get that $m'_n = m_{k_1} = m'_{k_2}$, a contradiction as $0 < k_2 < n$.
    \end{enumerate}
\end{enumerate}
Take $S$ to be $\{m_i,m_{k_1},m_{j_1}\}$.

\subsection{Proof of Proposition \ref{prop-lie-basis}} \label{proof-lie-basis}
Define the matrix $M^{HC_n} \in \F^{(n-1)! \times (n^2-n)}$, with rows indexed by $\sigma \in C_n$ and columns by $(i,j)$ in lexicographic order of the variables $x_{i,j}$, as
     \begin{equation} \label{eqn-LieMat-def}
         M^{HC_n}_{(\sigma,(i,j))} = \begin{cases} 
          1 & \sigma(i) = j, \\
          0 & \text{ otherwise}
       \end{cases}
     \end{equation}
where $\sigma \in C_n$. The matrix $M^{HC_n}$ represents the system of linear equations described by Proposition \ref{prop-lie-alg-struct}, and hence its nullspace is $\mathfrak{g}_{HC_n}$. We analyse $M^{HC_n}$ for $n \neq 4$ here to show the set $\mathcal{B}_n$ as described in the proposition statement is a basis for $\mathfrak{g}_{HC_n}$. For $n = 3$, $M^{HC_3}$ is as in \eqref{eqn-Lie-HC3}. 
\begin{equation} \label{eqn-Lie-HC3}                        
    \begin{blockarray}{ccccccc}
        & \text{\small{(1,2)}} &  \text{\small{(1,3)}} &  \text{\small{(2,1)}} &  \text{\small{(2,3)}} &  \text{\small{(3,1)}} &  \text{\small{(3,2)}}  \\
        \begin{block}{c(cccccc@{\hspace*{1pt}})}
            (1,2,3) & 1 & 0 & 0 & 1 & 1 & 0  \\
            (1,3,2) & 0 & 1 & 1 & 0 & 0 & 1  \\
        \end{block}
    \end{blockarray}  
\end{equation}
Clearly, $M^{HC_3}$ is full row rank over any $\F$. Hence, $\dim (\mathfrak{g}_{HC_3}) = 4$ over any $\F$. The set $\mathcal{B}_3$ described in the proposition statement can be easily verified to lie in $\mathfrak{g}_{HC_3}$. Stacking the elements of $\mathcal{B}_3$ together, we get the following matrix 
\begin{equation} \label{eqn-basis-lin-indep-3}                       
    \begin{blockarray}{cccccc|ccccc}
        & A^{(2)} & A^{(3)} & B^{(2)} & C  \\
        \begin{block}{c(ccccc|ccccc@{\hspace*{5pt}})}
            (1,2) & 1 & 1 & -1 & 1 \\
            (1,3) & 1 & 1 & 0 & 0  \\
            (2,1) & -1 & 0 & 1 & -1  \\
            (2,3) & -1 & 0 & 0 & -1  \\
            (3,1) & 0 & -1 & 1 & 0  \\
            (3,2) & 0 & -1 & -1 & 1  \\
        \end{block}
    \end{blockarray}  
\end{equation}
It can be easily verified that the matrix in \eqref{eqn-basis-lin-indep-3} has rank $4$ by noting that the submatrix indexed by the first four rows has determinant $1$, proving the statement for $n = 3$. Now, we assume $n \geq 5$. Note that for any $\sigma \in C_n$, $k \in [2,n]$ and $\ell \in [2,n-1]$, we have
\[\sum_{i=1}^{n}A^{(k)}_{(i,\sigma(i))} = A^{(k)}_{(1,\sigma(1))} + A^{(k)}_{(k,\sigma(k))} = 1 - 1 = 0,\]
\[\sum_{i=1}^{n}B^{(\ell)}_{(i,\sigma(i))} = B^{(\ell)}_{(\sigma^{-1}(1),1)} + B^{(\ell)}_{(\sigma^{-1}(\ell),\ell)} = 1 - 1 = 0 \text{  and}\]
\[\sum_{i=1}^{n}C_{(i,\sigma(i))} = C_{(2,\sigma(2))} + C_{(\sigma^{-1}(2),2)} = -1 + 1 = 0.\]
Thus, $\mathcal{B}_n \subset \mathfrak{g}_{HC_n}$. We now show that $\mathcal{B}_n$ is $\F$-linearly independent, implying $\dim_{\F}(\mathfrak{g}_{HC_n}) \geq 2n-2$. Let $M \in \F^{(n^2-n) \times (2n-2)}$ be the matrix formed by stacking the elements of $\mathcal{B}_n$ as columns, and let $R$ be the set 
\[R = \{(1,j) \ | \ j \in [2,n-1] \} \sqcup \{ (2,3) \} \sqcup \{(i,1) \ | \ i \in [2,n-1]\} \sqcup \{(1,n)\}.\] 
Then, upto reordering of rows, the submatrix $M_{R \times \bullet}$ is as follows
\begin{equation} \label{eqn-basis-lin-indep}                       
    \begin{blockarray}{ccccccccccc}
        & A^{(2)} & A^{(3)} & \dots & A^{(n-1)} & A^{(n)} & B^{(2)} & B^{(3)} & \dots & B^{(n-1)} & C \\
        \begin{block}{c(ccccc|ccccc@{\hspace*{5pt}})}
            (1,2) & 1 & 1 & \dots & 1 & 1 & -1 & 0 & \dots & 0 & 1 \\
            (1,3) & 1 & 1 & \dots & 1 & 1 & 0 & -1 & \dots & 0 & 0 \\
            \vdots & \vdots & \vdots & \vdots & \vdots & \vdots & \vdots & \vdots & \vdots & \vdots & \vdots  \\
            (1,n-1) & 1 & 1 & \dots & 1 & 1 & 0 & 0 & \dots & -1 & 0 \\
            (2,3) & -1 & 0 & \dots & 0 & 0 & 0 & -1 & \dots & 0 & -1 \\
        \cline{2-11}
            (2,1) & -1 & 0 & \dots & 0 & 0 & 1 & 1 & \dots & 1 & -1 \\ 
            (3,1) & 0 & -1 & \dots & 0 & 0 & 1 & 1 & \dots & 1 & 0 \\
            \vdots & \vdots & \vdots & \vdots & \vdots & \vdots & \vdots & \vdots & \vdots & \vdots & \vdots  \\
            (n-1,1) & 0 & 0 & \dots & -1 & 0 & 1 & 1 & \dots & 1 & 0 \\
            (1,n) & 1 & 1 & \dots & 1 & 1 & 0 & 0 & \dots & 0 & 0 \\
        \end{block}
    \end{blockarray}  
\end{equation}
We show that $det(M_{R \times \bullet}) = \pm 1$. On $M_{R \times \bullet}$, apply the elementary row operations $R_i \mapsto R_i - R_{2n-2}$, where $R_i$ refers to the $i$'th row, for all $i \in [1,n-2]$ and $R_{n-1} \mapsto R_{n-1} - R_n$ to get the following matrix: 

\begin{equation*}                      
    \begin{blockarray}{ccccccccccc}
        & A^{(2)} & A^{(3)} & \dots & A^{(n-1)} & A^{(n)} & B^{(2)} & B^{(3)} & \dots & B^{(n-1)} & C \\
        \begin{block}{c(ccccc|ccccc@{\hspace*{5pt}})}
            (1,2) & 0 & 0 & \dots & 0 & 0 & -1 & 0 & \dots & 0 & 1 \\
            (1,3) & 0 & 0 & \dots & 0 & 0 & 0 & -1 & \dots & 0 & 0 \\
            \vdots & \vdots & \vdots & \vdots & \vdots & \vdots & \vdots & \vdots & \vdots & \vdots & \vdots  \\
            (1,n-1) & 0 & 0 & \dots & 0 & 0 & 0 & 0 & \dots & -1 & 0 \\
            (2,3) & 0 & 0 & \dots & 0 & 0 & -1 & -2 & \dots & -1 & 0 \\
        \cline{2-11}
            (2,1) & -1 & 0 & \dots & 0 & 0 & 1 & 1 & \dots & 1 & -1 \\ 
            (3,1) & 0 & -1 & \dots & 0 & 0 & 1 & 1 & \dots & 1 & 0 \\
            \vdots & \vdots & \vdots & \vdots & \vdots & \vdots & \vdots & \vdots & \vdots & \vdots & \vdots  \\
            (n-1,1) & 0 & 0 & \dots & -1 & 0 & 1 & 1 & \dots & 1 & 0 \\
            (1,n) & 1 & 1 & \dots & 1 & 1 & 0 & 0 & \dots & 0 & 0 \\
        \end{block}
    \end{blockarray}  
\end{equation*}
As the resulting matrix is block triangular, $det(M_{R \times \bullet})$ is, up to a sign, the product of the determinants of the lower-left block matrix and the upper-right block matrix in the above matrix. It is easily verified that the lower-left block matrix has determinant $(-1)^{n-2}$, while the upper-right block matrix has determinant $(-1)^{n-1}$. Hence, $det(M_{R \times \bullet}) =\pm 1$ and $\mathcal{B}_n$ is $\F$-linearly independent. 

Proposition \ref{prop-lie-dim-lb} shows a rank lower bound of $(n-1)(n-2)$ for $M^{HC_n}$ implying, by the Rank-Nullity Theorem, that $\text{dim}_{\F}(\mathfrak{g}_{HC_n}) \leq 2n-2$, proving $\mathcal{B}_n$ is a basis and hence Proposition \ref{prop-lie-basis}. Note that Proposition \ref{prop-lie-dim-lb} also shows how to efficiently construct a row basis $M^{(n)}$ of $M^{HC_n}$, which we leverage in Algorithm \ref{alg3}. For $n = 3$, we have $M^{(3)} = M^{HC_3}$.

\begin{proposition} \label{prop-lie-dim-lb}
 Let $n \geq 5$. We can construct in time $n^{O(1)}$ a submatrix $M^{(n)} \in \F^{(n-1)(n-2) \times (n^2-n)}$ of $M^{HC_n}$ such that $M^{(n)}$ contains a $(n-1)(n-2) \times (n-1)(n-2)$ submatrix with determinant $\pm 1$. 
\end{proposition}

 \begin{proof}
    This will be a proof by induction on $n$.
\paragraph{Base case.} For $n = 5$, let $R$ be the following set of $5$-cycles
     \[R = \{(1 \ 2 \ 3 \ 4 \ 5), (1 \ 3 \ 2 \ 4 \ 5), (1 \ 4 \ 2 \ 3 \ 5), (1 \ 5 \ 2 \ 3 \ 4), (1 \ 3 \ 4 \ 5 \ 2), (1 \ 2 \ 4 \ 3 \ 5), \]\[(1 \ 2 \ 5 \ 3 \ 4), (1 \ 2 \ 4 \ 5 \ 3),(1 \ 3 \ 4 \ 2 \ 5), (1 \ 2 \ 3 \ 5 \ 4), (1 \ 3 \ 2 \ 5 \ 4), (1 \ 3 \ 5 \ 2 \ 4).
      \} \]
Then, we set $M^{(5)} = M^{HC_5}_{R \times \bullet }$, where each row is indexed by a permutation $\sigma \in R$ and column by $(i,j)$, corresponding to the variable $x_{i,j}$, ordered lexicographically.
     \begin{equation} \label{eqn-base-case-matrix}
          \begin{blockarray}{ccccc|cccc|cccc|cccc|cccc}
        \begin{block}{c(cccc|cccc|cccc|cccc|cccc@{\hspace*{5pt}})}
      (1 \ 2 \ 3 \ 4 \ 5) & 1 & 0 & 0 & 0 & 0 & 1 & 0 & 0 & 0 & 0 & 1 & 0 & 0 & 0 & 0 & 1 & 1 & 0 & 0 & 0\cr
      (1 \ 3 \ 2 \ 4 \ 5) & 0 & 1 & 0 & 0 & 0 & 0 & 1 & 0 & 0 & 1 & 0 & 0 & 0 & 0 & 0 & 1 & 1 & 0 & 0 & 0\cr      
      (1 \ 4 \ 2 \ 3 \ 5) & 0 & 0 & 1 & 0 & 0 & 1 & 0 & 0 & 0 & 0 & 0 & 1 & 0 & 1 & 0 & 0 & 1 & 0 & 0 & 0\cr 
      (1 \ 5 \ 2 \ 3 \ 4) & 0 & 0 & 0 & 1 & 0 & 1 & 0 & 0 & 0 & 0 & 1 & 0 & 1 & 0 & 0 & 0 & 0 & 1 & 0 & 0\cr 
        \cline{2-21}
      (1 \ 3 \ 4 \ 5 \ 2) & 0 & 1 & 0 & 0 & 1 & 0 & 0 & 0 & 0 & 0 & 1 & 0 & 0 & 0 & 0 & 1 & 0 & 1 & 0 & 0\cr 
      (1 \ 2 \ 4 \ 3 \ 5) & 1 & 0 & 0 & 0 & 0 & 0 & 1 & 0 & 0 & 0 & 0 & 1 & 0 & 0 & 1 & 0 & 1 & 0 & 0 & 0\cr 
      (1 \ 2 \ 5 \ 3 \ 4) & 1 & 0 & 0 & 0 & 0 & 0 & 0 & 1 & 0 & 0 & 1 & 0 & 1 & 0 & 0 & 0 & 0 & 0 & 1 & 0\cr 
      (1 \ 2 \ 4 \ 5 \ 3) & 1 & 0 & 0 & 0 & 0 & 0 & 1 & 0 & 1 & 0 & 0 & 0 & 0 & 0 & 0 & 1 & 0 & 0 & 1 & 0\cr 
        \cline{2-21}
      (1 \ 3 \ 4 \ 2 \ 5) & 0 & 1 & 0 & 0 & 0 & 0 & 0 & 1 & 0 & 0 & 1 & 0 & 0 & 1 & 0 & 0 & 1 & 0 & 0 & 0\cr 
      (1 \ 2 \ 3 \ 5 \ 4) & 1 & 0 & 0 & 0 & 0 & 1 & 0 & 0 & 0 & 0 & 0 & 1 & 1 & 0 & 0 & 0 & 0 & 0 & 0 & 1\cr 
      (1 \ 3 \ 2 \ 5 \ 4) & 0 & 1 & 0 & 0 & 0 & 0 & 0 & 1 & 0 & 1 & 0 & 0 & 1 & 0 & 0 & 0 & 0 & 0 & 0 & 1\cr 
      (1 \ 3 \ 5 \ 2 \ 4) & 0 & 1 & 0 & 0 & 0 & 0 & 1 & 0 & 0 & 0 & 0 & 1 & 1 & 0 & 0 & 0 & 0 & 1 & 0 & 0\cr
        \end{block}
        \end{blockarray}  
     \end{equation}
    By Gaussian elimination, it can be verified that $M^{(5)}$ has full row rank (see \eqref{eqn-base-case-matrix-reduced}). Moreover, the pivot entries are $\pm 1$, hence $M^{(5)}$ has a submatrix with determinant $\pm 1$.

\begin{equation} \label{eqn-base-case-matrix-reduced}
          \begin{blockarray}{ccccc|cccc|cccc|cccc|cccc}
        \begin{block}{c(cccc|cccc|cccc|cccc|cccc@{\hspace*{5pt}})}
      (1 \ 2 \ 3 \ 4 \ 5) & 1 & 0 & 0 & 0 & 0 & 0 & 0 & 1 & 0 & 0 & 0 & 1 & 0 & 0 & 2 & -1 & 1 & 0 & -1 & 1\cr
      (1 \ 3 \ 2 \ 4 \ 5) & 0 & 1 & 0 & 0 & 0 & 0 & 0 & 1 & 0 & 0 & 0 & 1 & 0 & 0 & 2 &-1 & 1 & 1 &-2 & 1 \cr      
      (1 \ 4 \ 2 \ 3 \ 5) & 0 & 0 & 1 & 0 & 0 & 0 & 0 & 1 & 0 & 0 & 0 & 1 & 0 & 0 & 2 & -1 & 1 & 1 & -1 & 0\cr 
      (1 \ 5 \ 2 \ 3 \ 4) & 0 & 0 & 0 & 1 & 0 & 0 & 0 & 1 & 0 & 0 & 0 & 1 & 0 & 0 & 3 & -2 & 1 & 1 & -2 & 1\cr 
        \cline{2-21}
      (1 \ 3 \ 4 \ 5 \ 2) & 0 & 0 & 0 & 0 & 1 & 0 & 0 & -1 & 0 & 0 & 0 & 0 & 0 & 0 & -1 & 1 & -1 & 0 & 1 & 0\cr 
      (1 \ 2 \ 4 \ 3 \ 5) & 0 & 0 & 0 & 0 & 0 & 1 & 0 & -1 & 0 & 0 & 0 & 0 & 0 & 0 & -1 & 1 & 0 & 0 & 0 & 0\cr 
      (1 \ 2 \ 5 \ 3 \ 4) & 0 & 0 & 0 & 0 & 0 & 0 & 1 & -1 & 0 & 0 & 0 & 0 & 0 & 0 & -1 & 1 & 0 & 0 & 1 & -1\cr 
      (1 \ 2 \ 4 \ 5 \ 3) & 0 & 0 & 0 & 0 & 0 & 0 & 0 & 0 & 1 & 0 & 0 & -1 & 0 & 0 & -1 & 1 & -1 & 0 & 1 & 0\cr 
        \cline{2-21}
      (1 \ 3 \ 4 \ 2 \ 5) & 0 & 0 & 0 & 0 & 0 & 0 & 0 & 0 & 0 & 1 & 0 & -1 & 0 & 0 & -1 & 1 & 0 & -1 & 1 & 0\cr 
      (1 \ 2 \ 3 \ 5 \ 4) & 0 & 0 & 0 & 0 & 0 & 0 & 0 & 0 & 0 & 0 & 1 & -1 & 0 & 0 & -1 & 1 & 0 & 0 & 1 &-1\cr 
      (1 \ 3 \ 2 \ 5 \ 4) & 0 & 0 & 0 & 0 & 0 & 0 & 0 & 0 & 0 & 0 & 0 & 0 & 1 & 0 & -1 & 0 & -1 & 0 & 1 & 0\cr 
      (1 \ 3 \ 5 \ 2 \ 4) & 0 & 0 & 0 & 0 & 0 & 0 & 0 & 0 & 0 & 0 & 0 & 0 & 0 & 1 & -1 & 0 & 0 & -1 & 1 & 0\cr
        \end{block}
        \end{blockarray}  
     \end{equation}

\paragraph{Inductive hypothesis.} Suppose that we can construct $M^{(n-1)}$ in $O((n-1)^5)$ time, where $n > 5$. By Claim \ref{claim-Sn-to-Sn+1}, which we prove later, we can extend $M^{(n-1)}$ to a submatrix of $M^{HC_n}$. Thus, we treat $M^{(n-1)}$ as a $(n-2)(n-3) \times (n^2-n)$ submatrix of $M^{HC_n}$.  
\begin{claim} \label{claim-Sn-to-Sn+1}
        Let $n \geq 3$ and $R' = \{\sigma \in C_n | \ \sigma(1) = 2\}$ and $T = \{(i,j) | \ i \neq 1, j \neq 2 \text{ and }(i,j) \neq (2,1) \}$. Then $M^{HC_n}_{R' \times T} = M^{HC_{n-1}}$. In particular, $M^{(n-1)}$ can be extended to a submatrix $M'$ of $M^{HC_n}$ in $O(n^4)$ time such that $M'$ has a $(n-2)(n-3) \times (n-2)(n-3)$ submatrix with determinant $\pm 1$.
\end{claim}

Now, let $R_1 = \{ \tau_3, \tau_4, \dots, \tau_{n-1}, \tau_n\} \subseteq C_n$ such that for $i \in [3,n]$, $\tau_i(1) = i$, $\tau_i(2) \neq 1$ and $\tau_i(n) = 2$. Let $R_2 = \{ \sigma_3, \sigma_4, \dots, \sigma_{n-1}, \sigma_n\} \subseteq C_n$ such that for $i \in [3,n-2]$, $\sigma_i(1) = i$, $\sigma_i(i+1) = 2$ and $\sigma_i(2) = 1$, $\sigma_{n-1}(1) = n-1$, $\sigma_{n-1}(2)  = 1$, $\sigma_{n-1}(3) = 2$, and  $\sigma_{n}(1) = n$, $\sigma_{n}(2)  \neq 1$, $\sigma_{n}(3) = 2$. Note that $|R_1| = |R_2| = n-2$ and $n > 5$ ensures $R_1$ and $R_2$ exist. An explicit choice of $R_1$ and $R_2$, which can be constructed in $O(n^2)$ time, is as follows 
\begin{equation} \label{eqn-R1-induct}
\begin{split}
    R_1 = \{(1 \ 3 \ n \ 2 \ n-1 \ n - 2 \ \dots \ 7 \ 6 \ 5 \ 4), \\ (1 \ 4 \ n \ 2 \ n-1 \ n - 2 \ \dots \ 7 \ 6 \ 5 \ 3), \\ (1 \ 5 \ n \ 2 \ n-1 \ n - 2 \ \dots \ 7 \ 6 \ 4 \ 3), \\ (1 \ 6 \ n \ 2 \ n-1 \ n-2  \ \dots \ 7 \ 5 \ 4 \ 3), \\ \vdots \\ (1 \ n-3 \ n \ 2 \ n-1  \ n-2 \ \dots \ 6 \ 5 \ 4 \ 3),  \\ (1 \ n-2 \ n \ 2 \ n-1  \ n-3 \ \dots \ 6 \ 5 \ 4 \ 3), \\ (1 \ n-1 \ n \ 2 \ n-2 \ n - 3 \ \dots \ 6 \ 5 \ 4 \ 3), \\ (1 \ n \ 2 \ n-1 \ n-2 \ n - 3 \ \dots \ 6 \ 5 \ 4 \ 3)\}    
\end{split}
\end{equation}
\begin{equation}\label{eqn-R2-induct}
\begin{split}
    R_2 = \{(1 \ 3 \ n \ n-1 \ n-2  \ \dots \ 7 \ 6 \ 5 \ 4 \ 2), \\ (1 \ 4 \ n \ n-1 \ n-2  \ \dots \ 7 \ 6 \ 3 \ 5 \ 2), \\ (1 \ 5 \ n \ n-1 \ n-2  \ \dots \ 7 \ 4 \ 3 \ 6 \ 2), \\ (1 \ 6 \ n \ n-1 \ n-2  \ \dots \ 5 \ 4 \ 3 \ 7 \ 2), \\ \vdots \\ (1 \ n-3 \ n \ n-1 \ n-4  \ \dots \ 5 \ 4 \ 3 \ n-2 \ 2), \\ (1 \ n-2 \ n \ n-3 \ n-4  \ \dots \ 5 \ 4 \ 3 \ n-1 \ 2), \\ (1 \ n-1  \ n \ n-2 \ n-3  \ \dots \ 6 \ 5 \ 4 \ 3 \ 2), \\(1 \ n \ 3 \ 2 \ n-1 \ n-2  \ \dots \ 7 \ 6 \ 5 \ 4)\}    
\end{split}
\end{equation}
Consider the matrix $M^{(n)} = M^{HC_n}_{\tilde{R} \times \tilde{T}}$, where $\tilde{R} = R \sqcup R_1 \sqcup R_2$ with $R$ as the set of permutations in $C_n$ corresponding to the rows of $M^{(n-1)}$, $R_1$ as in \eqref{eqn-R1-induct} and $R_2$ as in \eqref{eqn-R2-induct}, and $\tilde{T} = T \sqcup T_1  \sqcup T_2$ with $T$ as per Claim \ref{claim-Sn-to-Sn+1}, $T_1 = \{(1,j) \ | \ j \in [3,n]\}$ and $T_2 = \{(2,1)\} \sqcup \{(i,2) \ | \ i \in [3,n] \}$. The matrix $M^{(n)}$ is as: 
    \begin{equation} \label{eqn-induct-step}
    M^{(n)} = \begin{blockarray}{ccc}
        & T & T_1 \sqcup T_2  \\
        \begin{block}{c(c|c@{\hspace*{5pt}})}
        R & M^{(n-1)} & 0_{(n-2)(n-3) \times (2n-3)}  \\
        \cline{2-3}
        R_1 \sqcup R_2 & B & C \\
        \end{block}
        \end{blockarray}  
    \end{equation}
where $C \in \F^{(2n-4) \times (2n-3)}$ is a $0/1$ matrix. Claim \ref{claim-Sn-to-Sn+1} ensures that $M^{(n-1)}$ has only zero entries corresponding to the columns indexed by $T_1 \sqcup T_2$. By the inductive hypothesis, there exists a $(n-2)(n-3) \times (n-2)(n-3)$ submatrix $N$ of $M^{(n-1)}$ such that $det(N) = \pm 1$. Let $T' \subset T$ be the set of $(n-2)(n-3)$ columns, such that $N = M^{(n-1)}_{\bullet \times T'}$. Let $C' = C_{\bullet \times T_1 \sqcup (T_2 \backslash \{(n,2)\})}$. Note that $C'$ is a $(2n-4) \times (2n-4)$ matrix. Since $M^{(n)}$ is block-diagonal, we get that
\[det(M^{(n)}_{\bullet \times T' \sqcup T_1 \sqcup (T_2 \backslash \{(n,2)\})}) = det(N)\cdot det(C').\] 
The matrix $C'$, up to a reordering of rows and columns, is as in \eqref{eqn-induct-step2}, where the columns are ordered lexicographically with respect to the elements of $T_1 \sqcup (T_2 \backslash \{(n,2)\})$.
    \begin{equation} \label{eqn-induct-step2}
    C = \begin{blockarray}{cccccccccccccc}
        & (1,3) &  & \dots  & & & (1,n) & (2,1)& (3,2) & (4,2) &  & \dots &  & (n-1,2)  \\
        \begin{block}{c(cccccc|ccccccc@{\hspace*{5pt}})}
            \tau_{3} & 1 & 0 & \dots & 0 & 0 & 0 & 0 & 0 & 0 & 0 & \dots & 0 & 0 \\ 
            \tau_{4} & 0 & 1 & \dots & 0 & 0 & 0 & 0 & 0 & 0 & 0 & \dots & 0 & 0 \\
            \vdots & \vdots & \vdots & \vdots & \vdots & \vdots & \vdots & \vdots & \vdots & \vdots & \vdots & \vdots & \vdots & \vdots  \\
            \tau_{n-2} & 0 & 0 & \dots & 1 & 0 & 0 & 0 & 0 & 0 & 0 &\dots & 0 & 0 \\
            \tau_{n-1} & 0 & 0 & \dots & 0 & 1 & 0 & 0 & 0 & 0 & 0 &\dots & 0 & 0 \\
            \tau_{n} & 0 & 0 & \dots & 0 & 0 & 1 & 0 & 0 & 0 & 0 &\dots & 0 & 0 \\
        \cline{2-14}
            \sigma_{3} & 1 & 0 & \dots & 0 & 0 & 0 & 1 & 0 & 1 & 0 &\dots & 0 & 0 \\ 
            \sigma_{4} & 0 & 1 & \dots & 0 & 0 & 0 & 1 & 0 & 0 & 1 &\dots & 0 & 0 \\
            \vdots & \vdots & \vdots & \vdots & \vdots & \vdots & \vdots & \vdots & \vdots & \vdots & \vdots & \vdots & \vdots & \vdots  \\
            \sigma_{n-2} & 0 & 0 & \dots & 1 & 0 & 0 & 1 & 0 & 0 & 0 & \dots & 0 & 1 \\
            \sigma_{n-1} & 0 & 0 & \dots & 0 & 1 & 0 & 1 & 1 & 0 & 0 & \dots & 0 & 0 \\
            \sigma_{n} & 0 & 0 & \dots & 0 & 0 & 1 & 0 & 1 & 0 & 0 & \dots & 0 & 0 \\
            \end{block}
        \end{blockarray}  
    \end{equation}
It is then not hard to see that $det(C) = \pm 1$. Since $det(N) = \pm 1$, this completes the inductive step. 

The entire argument can be converted to a recursive algorithm to compute $M^{(n)}$ in $O(n^5)$ time. The matrix $M^{(5)}$ in \eqref{eqn-base-case-matrix} is the matrix for $n = 5$, the base case. By using Claim \ref{claim-Sn-to-Sn+1} we can map an already constructed $M^{(n-1)}$ to an $(n-2)(n-3) \times (n^2-n)$ matrix in $O(n^4)$ time. The rows corresponding to the sets in \eqref{eqn-R1-induct} and \eqref{eqn-R2-induct} can then be added to the resulting matrix in $O(n^4)$ time to get $M^{(n)}$. Since the recursion runs for $O(n)$ many steps and each step needs $O(n^4)$ time, $M^{(n)}$ can be consturcted in $O(n^{5})$ time.
\end{proof}
\begin{proof}[Proof of Claim \ref{claim-Sn-to-Sn+1}] We first show an injective map $\phi_{2}: C_{n-1} \rightarrow C_n$ which maps $\pi \in C_{n-1}$ to a $\sigma \in C_n$ such that $\sigma(1) = 2$. Let $\pi \in C_{n-1}$ where $\pi = (1 \ j_1 \ j_2 \ \dots \ j_{n-2})$ with $j_i = \pi^{i}(1)$, $i \in [n-2]$. Then $\phi_2$ acts as follows: 
\[ \phi_2:  (1 \ j_1 \ j_2 \ \dots \ j_{n-2}) \mapsto (1 \ 2 \ j_1+1 \ j_2+1 \ \dots \ j_{n-2}+1).\]          
Formally, $\sigma = \phi_2(\pi)$ is defined as
\[ \sigma(1) = 2, \ \sigma(i+1) = \pi(i)+1 \text{ and } \sigma(k+1) = 1.\]   
where $i,k \in [n-1]$ and $\pi(k) = 1$. Clearly, $\sigma$ is an $n$-cycle. We show $\phi_2$ is injective.
                
Suppose $\pi_1,\pi_2 \in C_{n-1}$ such that their respective images $\sigma_1$ and $\sigma_2$ under $\phi_2$ satisfy $\sigma_1 = \sigma_2$ . Thus, $\sigma_1(i) = \sigma_2(i)$ for all $i \in [n]$. Let $\pi_1(k) = 1$ with $k \in [2,n-1]$, then from the definition of $\phi_2$, and $\sigma_2 = \sigma_1$, we get that $\sigma_2(k+1) = \sigma_1(k+1) = 1$. This implies $\pi_2(k) = 1$, for otherwise $\sigma_2(k+1) = \pi_2(k)+1 \neq 1$ a contradiction. We also have that for all $i \in [n-1] \backslash \{k\}$, $\sigma_2(i+1) = \sigma_1(i+1) = \pi_1(i) + 1$. Since $\sigma_2(i+1) = \pi_2(i) + 1$, we get that $\pi_1(i) + 1 = \pi_2(i) + 1$ or $\pi_1(i)= \pi_2(i)$ for all $i \in [n-1] \backslash \{k\}$. Thus, $\pi_1 = \pi_2$.
        
Hence, $\phi_{2}$ is injective into $C_n$. Since $\phi_2$ is injective, it is a bijection from $C_{n-1} \to R'$, where $R'$, as in the claim statement, is the subset of $C_n$ comprising $n$-cycles $\sigma$ such that $\sigma(1) = 2$. Thus, the matrix $M^{HC_n}_{R' \times T}$ must be the same as $M^{HC_{n-1}}$. It is easy to see that we can compute $\phi_2(\pi)$ in $O(n)$ time for any $\pi \in C_{n-1}$ if we treat $\pi$ as a list of length $n$. Thus, we can extend $M^{(n-1)}$ to a matrix with $n^2-n$ columns in $O(n^4)$ time by using $\phi_2$ to map each row of $M^{(n-1)}$ to a row of $M^{HC_n}$. Note that the rank of the extended matrix is at least $(n-2)(n-3)$. This is because under the action of $\phi_2$ all the columns $(i,j)$ and $(i,1)$ of $M^{(n-1)}$, where $i, j \in [n-1]$ with $j \neq 1$, map to $(i+1,j+1)$ and $(i+1,1)$ respectively in the extended matrix. Thus, the structure of $M^{(n-1)}$ is preserved in the extended matrix, hence the rank is also at least $(n-2)(n-3)$.
\end{proof}

\subsection{Proof of Corollary \ref{corollary-lie-soln-struct}} \label{proof-lie-soln-struct}
Let $T$ be the set 
\[ T = \{(i,j) | \ i,j \in [2,n], \ i \neq j, (i,j) \neq (2,3)\} \sqcup \{(n,1)\}.\] 
Let $T^c$ denote the complement of $T$. Note that $|T| = (n-1)(n-2)$. 

From the equations in \eqref{eqn-lin-form-liebasis} form the matrix $N \in \F^{(n-1)(n-2) \times (n^2-n)}$ such that $N\vecz = 0$, with each row expressing $z_{i,j}$, where $(i,j) \in T$, in terms of $z_{k, \ell}$, where $(k,\ell) \in T^{c}$. The rows of $N$ are indexed by $z_{i,j}$, while the columns are indexed by all $\vecz$ variables. Let the elements of $T$ and $T^c$ be ordered by lexicographic ordering. After reordering the rows of $N$ in lexicographic order, $N$ can be written as
\begin{equation*}
        \begin{blockarray}{cc}
            T & T^c \\
            \begin{block}{(c|c@{\hspace*{5pt}})}
                I_{(n-1)(n-2)} & M\\
            \end{block}
        \end{blockarray},      
\end{equation*}
where $I_{(n-1)(n-2)}$ is the $(n-1)(n-2) \times (n-1)(n-2)$ identity matrix and row $(i,j)$ of $M$ corresponds to the linear form for $z_{i,j}$ in \eqref{eqn-lin-form-liebasis}. Clearly, $N$ has full row rank, and thus the dimension of the null space of $N$ is $2n-2$. It is easily verifiable that the elements of $\mathcal{B}_n$ described in Proposition \ref{prop-lie-basis} satisfy the equations in \eqref{eqn-lin-form-liebasis}, proving the corollary.

\subsection{Further analysis of $M^{(n)}$} \label{proof-det1-abelgroup}
Lemma \ref{lemma-Mn-det1} shows that a subset $T$ of the columns of $M^{(n)}$ is such that $|T| = (n-1)(n-2)$ and $M^{(n)}_{\bullet \times T}$ has determinant $\pm 1$. Lemma \ref{lemma-abeliangroup-soln}, proved using Lemma \ref{lemma-Mn-det1}, is used to analyse the structure of the scaling symmetries of $HC_n$ over all fields. Lemma  \ref{lemma-abeliangroup-general}, proved using Lemma \ref{lemma-abeliangroup-soln}, shows that solving the system of linear equations represented by $M^{HC_n}$ over an Abelian group reduces to solving the system of linear equations represented by $M^{(n)}$. Lemma \ref{lemma-abeliangroup-general} is used to prove the correctness of Algorithm \ref{alg3} over all fields. 

Note that for Lemmas \ref{lemma-abeliangroup-soln} and \ref{lemma-abeliangroup-general} we work over an Abelian group $G$, where we treat $+$ as the group operation, $0$ as the group identity. For $g \in G$ and $k \in \Z^{+}$, $kg$ means $g$ added to itself $k$ times, $(-k)g$ means $-g$, the inverse of $g$, added to itself $k$ times, and $0g$ gives the group identity $0$.

\begin{lemma} \label{lemma-Mn-det1}
    Let $\F$ be any field, $M^{(n)}$ be the matrix as constructed in the proof of Proposition \ref{prop-lie-basis} (See Proposition \ref{prop-lie-dim-lb}).    Then, $det(M^{(n)}_{\bullet \times T}) = \pm 1$, where $T$ is as
    \[T = \{(k, \ell) \ | \ (k,\ell) \neq (1,j) , \ j \in [2,n] \text{ and } (k,\ell) \neq (i,1) , \ i \in [2,n-1] \text{ and } (k,\ell) \neq (2,3) \}.\] 
\end{lemma}
\begin{proof}
Consider the matrix $N$ as constructed in Appendix \ref{proof-lie-soln-struct} and the matrix $M^{(n)}$ as constructed in Proposition \ref{prop-lie-dim-lb}. Both $M^{(n)}$ and $N$ are $(n-1)(n-2) \times (n^2-n)$ matrices. From the argument in Appendix \ref{proof-lie-soln-struct}, it follows that both matrices also have the \emph{same} nullspace over \emph{any} $\F$. It then follows from a standard Linear Algebra result, see Section 2.5 in \cite{HK2006}, that their row spaces must be the same. Thus, there exists $B \in \GL_{(n-1)(n-2)}(\F)$ such that 
                \[B M^{(n)} = N\]
implying
    \[B M^{(n)}_{\bullet \times T} = N_{\bullet \times T} = I_{(n-1)(n-2)}\]
where the last equality follows from the construction of $N$. Thus,
                \[det(B) det(M^{(n)}_{\bullet \times T}) = 1\]
Now, $M$ being a $0/1$ matrix implies $det(M^{(n)}_{\bullet \times T})$ is an integer. Since the above matrix product above holds over \emph{any} $\F$, we get that $det(M^{(n)}_{\bullet \times T}) = \pm 1$ and $det(B) = \pm 1$. Further, if $M_1 = M^{(n)}_{\bullet \times T}$, then $B = M^{-1}_1$. Thus,
            \[B M^{(n)} = N\]
can be written as
\begin{equation*}
        M_1^{-1} \cdot \begin{blockarray}{cc}
            \begin{block}{(c|c@{\hspace*{5pt}})}
            M_1 & M_2 \\
            \end{block}  \end{blockarray}  =
         \begin{blockarray}{cc}
            \begin{block}{(c|c@{\hspace*{5pt}})}
            I_{(n-1)(n-2)} & M\\
            \end{block}
        \end{blockarray}  
\end{equation*}
In particular, we get that $-M_1^{-1}M_2$ produces the linear combination of $z_{i,j}$'s corresponding to the RHS of the equations in \eqref{eqn-lin-form-liebasis}.

\end{proof}
\begin{lemma} \label{lemma-abeliangroup-soln}
    Let $G$ be an Abelian group (equivalently a $\Z$-module) and $g_{i,j} \in G$, with $i,j \in [n]$ and $i \neq j$, such that $\sum_{i=1}^{n} g_{i,\sigma(i)} = 0$ holds for all $\sigma \in C_n$. Then, the $g_{i,j}$'s satisfy \eqref{eqn-lin-form-liebasis}.
\end{lemma}

\begin{proof}
The system of linear equations in the statement can be written as $M^{HC_n}.\vecg = \veczero$ because the rows of $M^{HC_n}$ correspond to such a system of linear equations. Now,

\[M^{HC_n}\vecg = \veczero \implies M^{(n)}\vecg = \veczero\]
because $M^{(n)}$, as constructed in Proposition \ref{prop-lie-dim-lb}, contains a subset of the rows of $M^{HC_n}$. Note, now the number of equations is $(n-1)(n-2)$. We can write the system $M^{(n)}\vecg = \veczero$ as
    \begin{equation*} 
        \begin{blockarray}{c|c}
            \begin{block}{(c|c@{\hspace*{5pt}})}
            M_1 & M_2 \\
            \end{block}
        \end{blockarray}  
        \cdot
        \begin{blockarray}{(c)}
            \vecg_1 \\
            \vecg_2
        \end{blockarray}  
        = \veczero,              
    \end{equation*}
where $M_1 = M^{(n)}_{\bullet \times T}$, where $T$ is as in Lemma \ref{lemma-Mn-det1}, $M_2 = M^{(n)}_{\bullet \times T^{c}}$, where  $T^c$ is the complement of $T$, $\vecg_1$ are the entries $g_{i,j}$ with $(i,j) \in T$, and $\vecg_2$ are the remaining entries. We then get that 
\[M_1\vecg_1 + M_2\vecg_2 = \veczero\]
Since $M_1$ is a $0/1$ matrix with $det(M_1) = \pm 1$, then $M_1^{-1}$ is also an integer matrix. Thus, multiplying on the left by $M_1^{-1}$ on both sides of the above equation is an invertible operation, which gives
\[\vecg_1 +  M_1^{-1}M_2\vecg_2 = \veczero \implies \vecg_1 = -M_1^{-1}M_2\vecg_2.\]
From the proof of Lemma \ref{lemma-Mn-det1}, it follows that the matrix product $-M_1^{-1}M_2$ produces a linear combination of $\vecg_2$ entries which is exactly the same as the linear forms in the RHS of the equations in \eqref{eqn-lin-form-liebasis}. Thus, we get that the $g_{i,j}$'s satisfy \eqref{eqn-lin-form-liebasis}. 
\end{proof}
\begin{lemma} \label{lemma-abeliangroup-general}
    Let $G$ be an Abelian group such that 
    \begin{equation} \label{eqn-entire-sys}
      \sum_{i=1}^{n} x_{i,\sigma(i)} = h_{\sigma} \ \sigma \in C_n,  
    \end{equation}
    where $h_{\sigma} \in G$, has a solution. Then, any solution to
    \begin{equation} \label{eqn-Mn-sys}
    \sum_{i=1}^{n} x_{i,\sigma(i)} = h_{\sigma} \ \sigma \in R,
    \end{equation}
    where $R$ is the set permutations corresponding to rows of $M^{(n)}$, is also a solution to \eqref{eqn-entire-sys}.    
\end{lemma}
\begin{proof}
    Let $\vecu$ be a solution to \eqref{eqn-entire-sys}, which exists by assumption. Then $\vecu$ is also a solution to \eqref{eqn-Mn-sys}, since these equations are a subset of those of \eqref{eqn-entire-sys}. Let $\vecv$ be a solution to \eqref{eqn-Mn-sys}. Then for $\sigma \in R$, we have that,
    \[\sum_{i=1}^{n} (u_{i,\sigma(i)} - v_{i,\sigma(i)}) = \sum_{i=1}^{n} u_{i,\sigma(i)} - \sum_{i=1}^{n} v_{i,\sigma(i)} = h_{\sigma} - h_{\sigma} = 0.\ \]
    where all equalities follow from the fact that $G$ is Abelian. Thus, $\vecu - \vecv$ is a solution to $M^{(n)}.\vecx = \veczero$ which implies $\vecu - \vecv$ is also a solution to $M^{HC_n}.\vecx = \veczero$ by the proof of Lemma \ref{lemma-abeliangroup-soln}. Now, note that
    \[M^{HC_n}.\vecv = M^{HC_n}.(\vecv - \vecu + \vecu) =\]
    \[M^{HC_n}.(\vecv - \vecu) + M^{HC_n}. \vecu = \vech. \]
    Here, $\vech$ is the vector with entries as $h_\sigma$. The equalities follow from $G$ being Abelian, the linearity of matrix operations, and the aforementioned observation about $\vecu - \vecv$. 
\end{proof}

\subsection{Proof of Lemma \ref{prop-dist-eigen}} \label{proof-dist-eigen}
Consider the matrix $M$ as
\[M = \sum_{k=2}^{n}w_k A^{(k)} + \sum_{\ell=2}^{n-1}y_{k}B^{(\ell)} + zC\]
where $w_k$'s, $y_{\ell}$'s and $z$ are formal variables. Assuming $i \neq j$, the entries of $M$ are as 
     \begin{equation*} 
         M_{i,j} = \begin{cases} 
          \sum_{k = 2}^{n} w_k - y_2 + z  & i = 1, j = 2 \\
          \sum_{k = 2}^{n} w_k - y_j  & i = 1, j \in [3,n-1] \\
          \sum_{k = 2}^{n} w_k & i = 1, j = n \\
 
          \sum_{\ell = 2}^{n-1} y_{\ell} - w_2 - z  & i = 2, j = 1 \\

          - w_2 - y_j - z  & i = 2, j \in [3,n-1] \\
          - w_2 - z  & i = 2, j = n \\

          \sum_{\ell = 2}^{n-1} y_{\ell} - w_i  & i \in [3,n], j = 1 \\
          - w_i - y_2 + z  & i \in [3,n], j = 2 \\
          - w_i - y_j  & i \in [3,n], j \in [3,n-1], i \neq j\\
          - w_i  & i \in [3,n-1], j = n \\
       \end{cases}
     \end{equation*}
It can be seen that all the entries of $M$ are distinct linear polynomials, hence the difference of any two entries of $M$ is a non-zero linear polynomial. Let $S \subseteq \F$, with $|S| > \binom{n^2-n}{2}$. Then, by the Polynomial Identity Lemma, for a random choice of the $w_k$'s, $y_{\ell}$'s and $z$ variables, two entries of $M$ are the same with probability at most $\frac{1}{|S|}$. Applying union bound on all $\binom{n^2-n}{2}$ pairs of entries, we get that there exists a pair of equal entries in $M$ with probability at most $\frac{\binom{n^2-n}{2}}{|S|} < 1$. Thus, with probability at least $1 - \frac{\binom{n^2-n}{2}}{|S|} > 0$, $M$ has distinct diagonal entries. Hence, there exists an $M \in \mathfrak{g}_{HC_n}$ with distinct eigenvalues.

\subsection{Proof of Proposition \ref{prop: symmetries-PS}} \label{proof-symmetries-PS}
    Let $A \in \mathcal{G}_{HC_n}$, thus $HC_n(A\vecx) = HC_n(\vecx)$ implying $\mathfrak{g}_{HC_n(A\vecx)} = \mathfrak{g}_{HC_n}$. By the conjugacy of Lie Algebras of equivalent polynomials, Lemma \ref{lemma-lieconjug}, we also get that for every $B \in \mathfrak{g}_{HC_n(\vecx)} $, there is a $C \in \mathfrak{g}_{HC_n}$ such that $B = A^{-1}CA$. Thus, $AB = CA$, with both $B$ and $C$ being diagonal matrices. Choose $B$ to be a matrix where all the entries are distinct, which exists by Lemma \ref{prop-dist-eigen}. Consider the $(i,j)$'th row of $A$. As $A$ is invertible, $A_{(i,j),(k,\ell)} \neq 0$ for some $k,\ell \in [n]$, $k \neq \ell$. Now, $AB = CA$ implies $(AB)_{((i,j),(k,\ell))} = (CA)_{((i,j),(k,\ell))}$. Thus,
    \[A_{(i,j),(k,\ell)}B_{(k,\ell),(k,\ell)} = A_{(i,j),(k,\ell)}C_{(i,j),(i,j)}.\] 
    As $A_{(i,j),(k,\ell)} \neq 0$ we get $B_{(k,\ell),(k,\ell)} = C_{(i,j),(i,j)}$. Suppose $A_{(i,j),(k_1,\ell_1)} \neq 0$ for some $(k_1, \ell_1) \neq (k,\ell)$. Then, we also get that 
    \[A_{(i,j),(k_1,\ell_1)}B_{(k_1,\ell_1),(k_1,\ell_1)} = A_{(i,j),(k_1,\ell_1)}C_{(i,j),(i,j)},\] 
    implying 
    \[B_{(k,\ell_1),(k,\ell_1)} = C_{(i,j),(i,j)} = B_{(k,\ell),(k,\ell)},\]
    a contradiction as all entries of $B$ are distinct. Thus, every row of $A$ has exactly one non-zero entry. Since $A$ is invertible, it must be that $A = PS$ for some permutation $P$ and scaling $S$.
\subsection{Proof of Proposition \ref{prop-scalesymcont}} \label{proof-scalesymcont}
Let $S$ be a scaling symmetry of $HC_n$. Then, we have that 
\[\prod_{i=1}^{n} S_{i,\sigma(i)} = 1 \quad \forall \sigma \in C_n\]
This is a system of linear equations over the Abelian Group $\F^{\times}$, where we treat group multiplication as addition and $1$ (multiplicative identity of $\F^{\times}$) as $0$. Thus, we get that the above system of equations corresponds to
\[M^{HC_n}\vecs = \veczero,\]
where the entries of $\vecs$ are the entries of $S$. By Lemma \ref{lemma-abeliangroup-soln}, we get that the entries of $S$ must satisfy the equations in \eqref{eqn-scaling-sym-lie}, which are the multiplicative version of \eqref{eqn-lin-form-liebasis}. Hence, $S$ is a continuous scaling symmetry.

\subsection{Proof of Proposition \ref{prop-psym}} \label{proof-psym}
Let $\pi \in C_n$. We can write $\pi$ as $(1 \ m_1 \ m_2 \dots \ m_n)$. The monomial corresponding to $\pi$ in $HC_n$ is $\prod_{i \in [n]} x_{i,\pi(i)} = x_{1,m_1} x_{m_1,m_2} \dots x_{m_{n-1},m_n} x_{m_n,1}$. Under the action of $P^{(\sigma)}$, we have that
\[\prod_{i \in [n]} x_{i,\pi(i)} \mapsto \prod_{i \in [n]} x_{\sigma(i),\sigma(\pi(i))}.\]
Note
\[\prod_{i \in [n]} x_{\sigma(i),\sigma(\pi(i))} = x_{\sigma(1),\sigma(m_1)} x_{\sigma(m_1),\sigma(m_2)} \dots x_{\sigma(m_{n-1}),\sigma(m_n)} x_{\sigma(m_n),\sigma(1)}.\]
As $\sigma \in S_n$, the cycle $\pi$ maps uniquely to $\tau = (\sigma(1) \ \sigma(m_1) \ \sigma(m_2) \dots \ \sigma(m_n))$. Hence $P^{(\sigma)}$ is a permutation symmetry of $HC_n$. 
Now we consider $P^{(T)}$. Under the action of $P^{(T)}$, we have
\[\prod_{i \in [n]} x_{i,\pi(i)} \mapsto \prod_{i \in [n]} x_{\pi(i),i}.\]
Thus, the cycle $\pi = (1 \ m_1 \ m_2 \dots \ m_n)$ maps, under $P^{(T)}$, to the $n$-cycle $ (1 \ m_{n} \ m_{n-1} \dots m_2 \ m_1)$, which is $\pi^{-1}$. Since the inverse of a permutation is unique, $P^{(T)}$ maps each monomial to a unique monomial and hence is also a permutation symmetry of $HC_n$. To see that $P^{(\sigma)}$ and $P^{(T)}$ commute, note
\[(P^{(\sigma)}P^{(T)})_{(i,j),(k,\ell)} = \sum_{(a,b), a \neq b} P^{(\sigma)}_{(i,j),(a,b)}P^{(T)}_{(a,b),(k,\ell)}.\]
From the definitions of $P^{(\sigma)}$ and $P^{(T)}$, we get that $P^{(\sigma)}_{(i,j),(a,b)}P^{(T)}_{(a,b),(k,\ell)} = 1$ if and only if $(a,b) = (\sigma(i),\sigma(j))$ and $(a,b) = (\ell,k)$. Equivalently, $P^{(\sigma)}_{(i,j),(a,b)}P^{(T)}_{(a,b),(k,\ell)} = 1$ if and only if $k = \sigma(j)$ and $ \ell =\sigma(i)$. Thus, $(P^{(\sigma)}P^{(T)})_{(i,j),(k,\ell)} = 1$ if $k = \sigma(j)$ and $\ell =\sigma(i)$, otherwise it is zero. Similarly, $(P^{(T)}P^{(\sigma)})_{(i,j),(k,\ell)} = 1$ if $k = \sigma(j)$ and $\ell = \sigma(i)$, otherwise it is zero. Hence, we get $P^{(\sigma)}P^{(T)} = P^{(T)}P^{(\sigma)}$.

\subsection{Proof of Proposition \ref{prop-Psym-gens}} \label{proof-Psym-gens}
    Let $P$ be a permutation symmetry, thus $HC_n(P\vecx) = HC_n(\vecx)$. Suppose $P(x_{1,2}) = x_{i,j}$. By Observation \ref{obs-HCn-pdzero}, $R_{i,j}$ is as
      \[R_{i,j} := Q_{i,j} \sqcup T_{i,j} \sqcup\{x_{j,i}\}.\]
From Observation \ref{obs-HCn-pdzero} and $P$ being a symmetry, we get that $P$ induces a bijection $\Phi$ from $R_{1,2}$ to $R_{i,j}$ such that,
\begin{equation} \label{Psym-case1}
  \Phi: Q_{1,2} \mapsto Q_{i,j} , \ T_{1,2} \mapsto T_{i,j}, \ \{x_{2,1}\} \mapsto \{x_{j,i}\}  
\end{equation}
or
\begin{equation} \label{Psym-case2}
\Phi: Q_{1,2} \mapsto T_{i,j} , \ T_{1,2} \mapsto Q_{i,j}, \ \{x_{2,1}\} \mapsto \{x_{j,i}\}
\end{equation}
We get that $P(x_{2,1}) = x_{j,i}$. Suppose \eqref{Psym-case1} holds. Since $T_{1,2} \mapsto T_{i,j}$ under $\Phi$, there exists a bijection $\pi: [3,n]
\rightarrow [n] \backslash \{i,j\}$, such that $\Phi(x_{1,t}) = x_{i,\pi(t)}$  with $t \in [3,n]$. Thus, $P(x_{1,t}) = x_{i,\pi(t)}$. Hence, $P$ also induces a bijection similar to $\Phi$ from $R_{1,t}$ to $R_{i,\pi(t)}$. Since $\{x_{t,1}\} \subset R_{1,t}$, we get that $P(x_{t,1}) = x_{\pi(t),i}$. Thus, $R_{t,1}$ maps bijectively to $R_{\pi(t),i}$ under $P$ as well. Extend $\pi$ to a permutation on $[n]$ by defining $\pi(1) = i$ and $\pi(2) = j$. Now, for any $x_{i_1,j_1}$, with $i_1,j_1 \in [2,n]$ and $i_1 \neq j_1$, it is easy to see that $x_{i_1,j_1} \in T_{i_1,1}$ and $x_{i_1,j_1} \in Q_{1,j_1}$. Thus, $P(x_{i_1,j_1}) \in T_{\pi(i_1),i}$ and $P(x_{i_1,j_1}) \in Q_{i,\pi(j_1)}$ which implies $P(x_{i_1,j_1}) = x_{\pi(i_1),\pi(j_1)}$. Hence, for all $i,j \in [n]$, $i \neq j$, we have that $P(x_{i,j}) = x_{\pi(i),\pi(j)} = P^{(\pi)}(x_{i,j})$. Thus, $P = P^{(\pi)}$. 

Now, suppose \eqref{Psym-case2} holds. Since $T_{1,2} \mapsto Q_{i,j}$ under $\Phi$, there is a bijection $\pi : [3,n] \rightarrow [n] \backslash \{i,j\}$, such that $\Phi(x_{1,t}) = x_{\pi(t),j}$. An argument similar as above shows that $P(x_{1,t}) = x_{\pi(t),j}$ for all $t \in [3,n]$ and further $P(x_{t,1}) = x_{j,\pi(t)}$. Extend $\pi$ to a permutation on $[n]$ by defining $\pi(1) = j$, $\pi(2) = i$. Like before, for any $x_{i_1,j_1}$, with $i_1,j_1 \in [2,n]$ and $i_1 \neq j_1$, $x_{i_1,j_1} \in T_{i_1,1}$ and $x_{i_1,j_1} \in Q_{1,j_1}$. Thus, $P(x_{i_1,j_1}) \in Q_{j,\pi(i_1)}$ and $P(x_{i_1,j_1}) \in T_{i,\pi(j_1)}$ which implies $P(x_{i_1,j_1}) = x_{\pi(j_1),\pi(i_1)}$. Hence, for all $i,j \in [n]$, $i \neq j$, we have that $P(x_{i,j}) = x_{\pi(j),\pi(i)} = P^{(T)}P^{(\pi)}(x_{i,j})$. Thus, $P = P^{(T)}P^{(\pi)}$. 

\subsection{Proof of Observation \ref{obs-HCn-pdzero}} \label{proof-HCn-pdzero}
The reverse implication easily follows from the fact that the monomials of $HC_n$ correspond to $n$-cycles. We prove the forward direction. So, suppose $x_{k, \ell} \in R_{i,j}$, that is, $\frac{\partial^2 HC_n}{\partial x_{i,j}\partial x_{k,\ell}} = 0$. If $i = k$ or $j = \ell$, then we are done by the reverse direction. Hence, assume $i \neq k$ and $j \neq \ell$. Note we also have that $i \neq j$ and $k \neq \ell$. Then, we have the following possibilities:
\begin{enumerate}
    \item $i \neq \ell$ and $k \neq j$: For $n = 3$, this case is not possible as $i,j,k,\ell$ are all distinct. For $n \geq 4$, it is easy to see that there exists $\pi \in C_n$, such that $\pi(i) = j$, $\pi(k) = \ell$. Hence, 
    \[\frac{\partial^2 HC_n}{\partial x_{i,j}\partial x_{k,\ell}} = \sum_{\substack{\pi \in C_n, \\  \pi(i) = j, \\ \pi(k) = \ell} } \prod_{i_1 \in [n] \backslash \{i,k\}}x_{i_1,\pi(i_1) } \neq 0,\]
    a contradiction.
    
    \item $i = \ell$ and $k \neq j$:  There exists $\pi \in C_n$, such that $\pi(i) = j$, $\pi(k) = i = \ell$. Hence, 
        \[\frac{\partial^2 HC_n}{\partial x_{i,j} \partial x_{k,\ell}} = \sum_{\substack{\pi \in C_n, \\  \pi(i) = j, \\ \pi(k) = i} } \prod_{i_1 \in [n] \backslash \{i,k\}}x_{i_1,\pi(i_1) } \neq 0,\]
    a contradiction. Note that the case $i \neq \ell$ and $k = j$ gives a contradiction in the same way.  

    \item $i = \ell$ and $k = j$: This is the assumption in the reverse direction.
\end{enumerate}

Hence, only the last case is possible, proving the forward direction. The partition of $R_{i,j}$ as described in the Observation statement follows easily by using the conditions proved for the vanishing of the second-order derivatives of $HC_n$.

\subsection{Proof of Proposition \ref{prop-nonchar}} \label{proof-nonchar}
    Let $\pi \in S_n \backslash C_n$ such that $\pi(i) \neq i$ for all $i \in [n]$. For $n \geq 5$, it can be seen that such $\pi$'s exist, a concrete example is $\pi = (1 \ 2)(3 \ 4 \ \dots \ n)$. Let $m_{\pi} := \prod_{i \in [n]} x_{i,\pi(i)}$. 
    Consider the polynomial 
    \[g(\vecx) = \sum_{P \in \mathcal{G}_{HC_n}}m_{\pi}(P\vecx),\]
    where $P$ goes over all permutation symmetries of $HC_n$. Clearly, then $g(P\vecx) = g(\vecx)$ for all permutation matrices $P \in \mathcal{G}_{HC_n}$. Now, let $S \in \mathcal{G}_{HC_n}$. By Proposition \ref{prop-scalesymcont}, we have that the entries of $S$ are as described in \eqref{eqn-scaling-sym-lie}. Any monomial in $g(\vecx)$ corresponds to some $\sigma \in S_n \backslash C_n$ such that $\sigma(i) \neq i$ for all $i \in [n]$. This is because the permutation symmetries of $HC_n$ preserve the cyclic decomposition of any permutation, which can be observed from the proof of Proposition \ref{prop-psym} (see Section \ref{proof-psym}). We now show for any $\sigma \in S_n \backslash C_n$, where $\sigma(i) \neq i$ for all $i \in [n]$, that $m_{\sigma}(S\vecx) = m_{\sigma}(\vecx)$ by showing that $\prod_{i=1}^{n} S_{i,\sigma(i)} = 1$ for such $\sigma$'s. Thus, any monomial of $g$ is preserved under $S$, implying $S \in \mathcal{G}_g$. Thus, $\mathcal{G}_{HC_n} \subseteq \mathcal{G}_g$ but $g \neq c \cdot HC_n$ for any $c \in \F^{\times}$.  
    
    If $\sigma(n) = 1$, then
    \begin{equation*}
        \begin{split}
      \prod_{i=1}^{n} S_{i,\sigma(i)} &= S_{1,\sigma(1)} S_{n,1} \prod_{i=2}^{n-1} S_{i,\sigma(i)} = S_{1,\sigma(1)} S_{n,1} \prod_{i=2}^{n-1} \left(S_{i,1}S_{1,\sigma(i)} \cdot \frac{S_{2,3}}{S_{1,3}S_{2,1}}\right)\\
      &=
      \left(\frac{S_{2,3}}{S_{1,3}S_{2,1}}\right)^{n-2} S_{1,\sigma(1)} S_{n,1} \prod_{i=2}^{n-1} S_{i,1}S_{1,\sigma(i)} \\
      &= \left(\frac{S_{2,3}}{S_{1,3}S_{2,1}}\right)^{n-2} S_{1,\sigma(1)}  \left(\frac{S_{1,3}S_{2,1}}{S_{2,3}}\right)^{n-2} \left(\prod_{i=2}^{n} S_{1,i} \prod_{i=2}^{n-1} S_{i,1}\right)^{-1} \prod_{i=2}^{n-1} S_{i,1}S_{1,\sigma(i)}\\
      &= \frac{\prod_{i=1}^{n-1}S_{1,\sigma(i)} \cdot \prod_{i=2}^{n-1}S_{i,1} }{\prod_{i=2}^{n}S_{1,i} \cdot \prod_{i=2}^{n-1}S_{i,1}  } = 1 \\
        \end{split}.
    \end{equation*}
The last equality holds because $\sigma(n) = 1$ implies $\sigma$ is a bijection from $[n-1]$ to $ [2,n]$. 

If $\sigma(n) = k$, where $k \in [2,n-1]$, then $\sigma(j) = 1$ for some $j \in [2,n-1]$. We then get that
\begin{equation*}
    \begin{split}
    \prod_{i=1}^{n} S_{i,\sigma(i)} &= S_{1,\sigma(1)} S_{n,k} \prod_{i=2}^{n-1} S_{i,\sigma(i)} = S_{1,\sigma(1)} S_{n,k} S_{j,1} \prod_{i \neq 1,j,n} \left(S_{i,1}S_{1,\sigma(i)} \frac{S_{2,3}}{S_{1,3}S_{2,1}}\right) \\ &=
 S_{1,\sigma(1)} S_{n,1} S_{1,k} \frac{S_{2,3}}{S_{1,3}S_{2,1}} S_{j,1} \prod_{i \neq 1,j,n} \left(S_{i,1}S_{1,\sigma(i)} \frac{S_{2,3}}{S_{1,3}S_{2,1}}\right)\\ 
 &= S_{1,\sigma(1)} \left(\frac{S_{1,3}S_{2,1}}{S_{2,3}}\right)^{n-2} \left(\prod_{i=2}^{n} S_{1,i} \prod_{i=2}^{n-1} S_{i,1}\right)^{-1}  \left(\frac{S_{2,3}}{S_{1,3}S_{2,1}}\right)^{n-2} S_{1,k} S_{j,1} \prod_{i \neq 1,j,n} S_{i,1}S_{1,\sigma(i)} \\ 
 &= \frac{(S_{1,\sigma(1)}S_{1,k}\prod_{i \neq 1,j,n}S_{1,\sigma(i)})\cdot (S_{j,1} \prod_{i \neq j} S_{i,1})}{\prod_{i=2}^{n}S_{1,i} \cdot \prod_{i = 2}^{n-1}S_{i,1}} = 1.\\
    \end{split}
\end{equation*}
The last equality holds because $\sigma(j) = 1$ implies $\sigma$ is a bijection from $[n] \backslash \{j\}$ to $[2,n]$. 

\subsection{Scaling symmetries of $Perm_n$ and $HC_n$} \label{sec-scalesymrelate-hcperm}
It is easy to observe that any scaling symmetry of $Perm_n$ also gives a scaling symmetry for $HC_n$ because the monomials of $HC_n$ are a subset of those of $Perm_n$. We record this as an observation.
\begin{observation}
If $S \in \mathcal{G}_{Perm_n}$ is a scaling symmetry of $\text{Perm}_n$, then a matrix $S' \in \text{GL}_{n^2-n}(\F)$ can be constructed from $S$ such that $S'$ is a scaling symmetry of $HC_n$.
\end{observation} 
We now show the converse, that is, any scaling symmetry $S \in \mathcal{G}_{HC_n}$ can be used to derive a scaling symmetry $S'$ of $Perm_n$, where $n \neq 4$. In Appendix \ref{subsec-charsym-HC4}, we show for $HC_4$ an $S$ over appropriate fields such that $S \in \mathcal{G}_{HC_4}$ but $S$ cannot be extended to a scaling symmetry of the Permanent.

\begin{proposition} \label{prop-HCntoPermn-scalsym}
    Let $S \in \mathcal{G}_{HC_n}$ be a scaling symmetry. Consider the diagonal matrix $S' \in \F^{n^2 \times n^2}$ defined as:
    \begin{equation} \label{eqn-HCNtoPermn-scalsym}
    \begin{split}
        S'_{1,j} &=  S_{1,j} \quad j \in [2,n], \quad S'_{i,1} =  S_{i,1} \quad i \in [2,n-1], \\ S'_{1,1} &= \frac{S_{1,3}S_{2,1}}{S_{2,3}}, \quad S'_{n,1} = \frac{S_{1,1}^{'n-2}}{\prod_{i=2}^{n}S'_{1,i} \cdot \prod_{i=2}^{n-1}S'_{i,1}} \\
        S'_{i,j} &= \frac{S'_{1,j}S'_{i,1}}{S'_{1,1}} \quad i,j \in [2,n]. 
    \end{split}
    \end{equation}
    Then $S' \in \mathcal{G}_{Perm_n}$.
\end{proposition}
\begin{proof}
    Let $\pi \in S_n$. Then there exists $j \in [n]$ such that $\pi(j) = 1$. Now,
    \[\prod_{i=1}^{n} S'_{i,\pi(i)} = \frac{\prod_{i=1}^{n}S'_{i,1}S'_{1,\pi(i)}}{S'^{n}_{1,1}} = \frac{\prod_{i=1}^{n}S'_{i,1}\prod_{i=1}^{n}S'_{1,\pi(i)}}{S'^{n}_{1,1}} \]
    \[ = \frac{S_{1,1}^{'2}S'_{n,1}\prod_{i=2}^{n-1}S'_{i,1}\prod_{i=1, i \neq j}^{n}S'_{1,\pi(i)}}{S^{'n}_{1,1}} = 1\]
    by using the definition of $S'_{n,1}$, proving the statement.
\end{proof}

\subsection{ Proof of Lemma \ref{lemma-HCn-monotone}} \label{proof-HCn-mono}
    Let $X_n$ denote the matrix of $n \times n$ many variables. We show first how any $\pi \in C_n$ can be extended to a $\sigma \in C_m$, where $m > n$. Using this extension, we define a $m \times m$ matrix $Y^{(m,n)}$ such that $HC_m(Y^{(m,n)}) = HC_n(\vecx)$. Let $j \in [n] \backslash \{1\}$ such that $\pi(j) = 1$. The mapping $\Phi: C_n \to C_m$ maps $\pi$ to $\sigma$ as 
    \[\sigma(i) = \pi(i) \quad i \in [n] \backslash \{j\}, \quad \sigma(j) = n+1, \] 
    \[\sigma(k) = k+1 \quad k \in [n+1,m-1], \text{  and } \sigma(m) = 1.\] 
    It is not hard to see that $\Phi$ is an injective map from $C_n$ to $C_m$. The entries of $Y^{(m,n)}$ are defined as follows
     \begin{equation*}
         Y^{(m,n)}_{i,j} = \begin{cases} 
        x_{i,j} & i \neq j, \ i \in [1,n], \ j \in [2,n] \\ 
        x_{i,1} & i \in [2,n], \ j = n+1 \\
        1 & i \in [n+1,m-1], j = i + 1 \text{ or } i = m, j = 1 \\ 
        0 & \text{ otherwise},\\
       \end{cases}
     \end{equation*}
    Then, $Y^{(m,n)}$ is as follows
        \begin{equation} \label{eqn-HCn-proj}
    \begin{blockarray}{cccccccccccccc}
        & 1 & 2 & 3 & \dots  & n-1 & n & n+1 & n+2 & n+3 & \dots & m-1 & m \\
        \begin{block}{c(cccccc|ccccccc@{\hspace*{2pt}})}
          1 & 0 & x_{1,2} & x_{1,3} & \dots  & x_{1,n-1} & x_{1,n} & 0 & 0 & 0 & \dots & 0 & 0  \\
          2 & 0 & 0 & x_{2,3} & \dots  & x_{2,n-1} & x_{2,n} & x_{2,1} &  0 & 0 &\dots & 0 & 0 \\
          3 & 0 & x_{3,2} & 0 & \dots  & x_{1,n-1} & x_{3,n} & x_{3,1} & 0 & 0 &\dots & 0 & 0 \\
           \vdots & \vdots & \vdots  & \vdots & \vdots & \vdots  & \vdots & \vdots & \vdots & \vdots & \vdots & \vdots & \vdots\\
          n-1 & 0 & x_{n-1,2} & x_{n-1,3} & \dots  & 0 & x_{n-1,n} & x_{n-1,1} & 0 & 0 & \dots & 0 & 0 \\
          n & 0 & x_{n,2} & x_{n,3} & \dots & x_{n,n-1} & 0 & x_{n,1} & 0 & 0 &\dots & 0 & 0  \\
          \cline{2-14}
          n+1 & 0 & 0 & 0 & \dots & 0  & 0 & 0 & 1 & 0 & \dots & 0 & 0  \\
          n+2 & 0 & 0 & 0 & \dots & 0 & 0 & 0 & 0 & 1 & \dots & 0 & 0 \\
           \vdots & \vdots & \vdots  & \vdots & \vdots & \vdots  & \vdots & \vdots & \vdots & \vdots & \vdots & \vdots & \vdots  \\
          m-1 & 0 & 0 & 0 & \dots & 0  & 0 & 0 & 0 & 0 & \dots & 0 & 1  \\
          m & 1 & 0 & 0 & \dots & 0  & 0 & 0 & 0 & 0 & \dots  & 0 & 0\\
        \end{block}
        \end{blockarray}  
    \end{equation}
    Now, let $\sigma \in C_n$. In $HC_m(Y^{(m,n)})$, the term corresponding to $\sigma$ is $\prod_{i=1}^{m}Y^{(m,n)}_{i,\sigma(i)}$. It follows from the definition of $Y^{(m,n)}$ that for the term to be non-zero, it must be that $\sigma(m) = 1$, $\sigma(i) = i+1$ for $i \in [n+1,m-1]$, and $\sigma(1) \neq n+1$. Thus, $\sigma$ is  $(1 \ k_1 \ k_2 \ \dots \ k_n \ n+1 \ n+2 \ \dots \ m-1 \ m)$ where $k_i \in [2,n]$ which implies,
    \[\prod_{i=1}^{m}Y^{(m,n)}_{i,\sigma(i)} = \prod_{i=1}^{n}Y^{(m,n)}_{i,\sigma(i)} \prod_{i=n+1}^{m}Y^{(m,n)}_{i,\sigma(i)} = \prod_{i=1}^{n}Y^{(m,n)}_{i,\sigma(i)}.\]
    For the remaining $i \in [n] \backslash \{k_n\}$, $\sigma(i) \in [2,n]$. This further gives
    \[\prod_{i=1}^{n}Y^{(m,n)}_{i,\sigma(i)} = Y^{(m,n)}_{k_n,n+1}\prod_{i=1, i \neq k_n}^{n}Y^{(m,n)}_{i,\sigma(i)} = x_{k_n,1}  \prod_{i=1, i \neq k_n}^{n}x_{i,\sigma(i)}.\]
    The $n$-cycle $\pi \in C_n$, defined as $\pi = (1 \ k_1 \ k_2 \ \dots \ k_n)$ clearly satisfies $\Phi(\pi) = \sigma$. Since $\phi$ is injective, $\pi$ is the unique pre-image of $\sigma$ under $\Phi$. Hence, we can write
    \[x_{k_n,1} \prod_{i=1, i \neq k_n}^{n}x_{i,\sigma(i)} = \prod_{i=1}^{n} x_{i,\pi(i)}.\]
    Hence, every non-zero term in $HC_m(Y^{(m,n)})$ is a monomial of $HC_n(\vecx)$. Conversely, it is easy to see from the map $\Phi$ that every $\pi \in C_n$ maps to a unique $\sigma \in C_n$ and that $\prod_{i=1}^{m}Y^{(m,n)}_{i,\sigma(i)} = \prod_{i=1}^{n} x_{i,\pi(i)}$. Thus, $HC_m(Y^{(m,n)}) = HC_n(\vecx)$. It is easy to see that $Y^{(m,n)}$ can be constructed in $O(m^2)$ time.

\subsection{Proof of Lemma \ref{lemma-HCn-downself}} \label{proof-HCn-downself}
    We first prove the forward direction. It suffices to show that for all $n \geq 3$ it holds that 
    \[HC_n(\vecx) = \sum_{i=2}^{n}x_{1,i}HC_{n-1}(X^{(i)}_{n-1}),\]
   where the matrix $X^{(i)}_{n-1}$ is obtained from $X_n$ by first swapping row $i$ with row $2$ and column $i$ with column $2$, and then removing row $1$ and column $2$ from the resulting matrix. By Lemma \ref{lemma-HCn-monotone}, for all $k < n$, we can express $HC_k(X_k)$ as $HC_n(Y^{(n,k)})$ and $HC_{k-1}(X^{(i)}_{k-1})$ as $HC_n(Y^{(n,k-1)}_i)$, completing the proof of the forward direction.\footnote{We assume the diagonal entries of $X_n$ are replaced by $0$'s} For $n = 3$ and $n = 4$, it can be easily verified that the above equality holds. Thus, we assume $n \geq 5$. We denote $X^{(i)}_{n-1}$ by $Z^{(i)}$ from now on for ease of exposition. Based on the above description, $Z^{(i)}$ looks as
    \begin{equation} \label{eqn-HCn-downself}
    \begin{blockarray}{cccccccccc}
         & 1 & 2 & 3 & \dots & i-2 & i-1 & i & \dots  & n -1   \\
        \begin{block}{c(ccccccccc@{\hspace*{2pt}})}
         1 & x_{i,1} & x_{i,3} & x_{i,4} & \dots & x_{i,i-1}  & x_{i,2} & x_{i,i+1} & \dots  & x_{i,n}\\
        2 & x_{3,1} & 0 & x_{3,4} & \dots & x_{3,i-1}  & x_{3,2} & x_{3,i+1} & \dots  & x_{3,n}\\
        3 & x_{4,1} & x_{4,3} & 0 & \dots & x_{4,i-1}  & x_{4,2} & x_{4,i+1} & \dots  & x_{4,n} \\
        \vdots & \vdots & \vdots  & \vdots & \vdots & \vdots  & \vdots & \vdots & \vdots  & \vdots \\
        i-2 & x_{i-1,1} & x_{i-1,3} & x_{i-1,4} & \dots & 0  & x_{i-1,2} & x_{i-1,i+1} & \dots  & x_{i-1,n}\\
        i-1 & x_{2,1} & x_{2,3} & x_{2,4} & \dots & x_{2,i-1}  & 0 & x_{2,i+1} & \dots  & x_{2,n}\\
        i & x_{i+1,1} & x_{i+1,3} & x_{i+1,4} & \dots & x_{i+1,i-1}  & x_{i+1,2} & 0 & \dots & x_{i+1,n}\\
        \vdots & \vdots & \vdots  & \vdots & \vdots & \vdots  & \vdots & \vdots & \vdots & \vdots \\
        n-2 & x_{n-1,1} & x_{n-1,3} & x_{n-1,4} & \dots & x_{n-1,i-1}  & x_{n-1,2} & x_{n-1,i+1} & \dots  & x_{n-1,n}\\
        n-1 & x_{n,1} & x_{n,3} & x_{n,4} & \dots & x_{n,i-1}  & x_{n,2} & x_{n,i+1} & \dots  & 0\\
        \end{block}
        \end{blockarray}  
    \end{equation}   
   Formally, the entries of $Z^{(i)}$ are defined as follows \begin{equation*}
         Z^{(i)}_{k,\ell} = \begin{cases} 
        x_{k+1,\ell+1} & k,\ell \in [n-1] \backslash \{1,i-1\}, \\ 
        x_{k+1,1} & k \in [n-1] \backslash \{1,i-1\}, \ell = 1 \\ 
        x_{k+1,2} & k \in [n-1] \backslash \{1,i-1\}, \ell = i-1, \\ 
        x_{i,\ell+1} & k = 1,\ell \in [n-1] \backslash \{1,i-1\}, \\ 
        x_{i,2} & k = 1, \ell = i-1, \\ 
        x_{2,\ell+1} & k = i - 1,\ell \in [n-1] \backslash \{1,i-1\}, \\ 
        x_{2,1} & k = i - 1,\ell = 1, \\ 
        x_{i,1} & k = 1, \ell = 1 \\
        0 & \text{ otherwise},\\
       \end{cases}
     \end{equation*}
    Note that we can replace $x_{i,1}$ by $0$ in $X^{(i)}_{n-1}$ as $HC_{n-1}$ does not depend on the diagonal entries. We show that
    \[HC_{n-1}(Z^{(i)}) = \sum_{\substack{\sigma \in C_n, \\ \sigma(1) = i}}\prod_{j=2}^{n}x_{j,\sigma(j)},\]
    Showing the above equality for a single $i$ suffices as the proof is similar for all $i$. Let $\pi \in C_{n-1}$. Then in $HC_{n-1}(Z^{(i)})$, the term corresponding to $\pi$ is
    \[\prod_{j=1}^{n-1}Z^{(i)}_{j, \pi(j)}.\]
     Based on the values of $\pi(1)$ and $\pi(i-1)$, we have the following cases:
    \begin{enumerate}
        \item $\pi(1) = i-1$. Then $\pi(i-1) = k$ and $\pi(\ell) = 1$, where $k, \ell \in [n-1] \backslash \{1,i-1\}$ and $k \neq \ell$ as $n \geq 5$. In this case
        \[\prod_{j=1}^{n-1}Z^{(i)}_{j, \pi(j)} = Z^{(i)}_{1,i-1} Z^{(i)}_{i-1,k} Z^{(i)}_{\ell,1} \prod_{j \neq 1,i-1,\ell} Z^{(i)}_{j, \pi(j)} = x_{i,2} x_{2,k+1} x_{\ell+1,1} \prod_{j \neq 1,i-1,\ell} x_{j+1, \pi(j)+1} \]
        \item $\pi(1) \neq i-1$ and $\pi(i-1) = 1$. Then $\pi(1) = k$ and $\pi(\ell) = i-1$, where $k, \ell \in [n-1] \backslash \{1,i-1\}$ and $k \neq \ell$ as $n \geq 5$. In this case
        \[\prod_{j=1}^{n-1}Z^{(i)}_{j, \pi(j)} = Z^{(i)}_{1,k} Z^{(i)}_{i-1,1} Z^{(i)}_{\ell,i-1} \prod_{j \neq 1,i-1,\ell} Z^{(i)}_{j, \pi(j)} = x_{i,k+1}x_{2,1}x_{\ell+1,2}\prod_{j \neq 1,i-1,\ell} x_{j+1, \pi(j)+1}\]
        \item $\pi(1) \neq i-1$ and $\pi(i-1) \neq 1$. Let $\pi(1) = j_1, \pi(i-1) = j_2$, $\pi(\ell) = i-1$ and $\pi(k) = 1$, with $j_1,j_2,\ell,k \in [n-1]$. In this case
        \[\prod_{j=1}^{n-1}Z^{(i)}_{j, \pi(j)} = Z^{(i)}_{1,j_1} Z^{(i)}_{i-1,j_2} Z^{(i)}_{\ell,i-1} Z^{(i)}_{k,1}\prod_{j \neq 1,i-1,\ell,k} Z^{(i)}_{j, \pi(j)}\]
        \[ = x_{i,j_1+1}x_{2,j_2 + 1}x_{\ell+1,2}x_{k+1,1}\prod_{j \neq 1,i-1,\ell,k} x_{j+1, \pi(j)+1}\]
    \end{enumerate}
    In all of the above cases, $\pi$ maps to a unique $\sigma \in C_n$ such that $\sigma(1) = i$ and $\prod_{j=1}^{n}x_{j,\sigma(j)} = x_{1,i} \prod_{j=1}^{n-1}Z^{(i)}_{j, \pi(j)}$. Thus, $x_{1,i}HC_{n-1}(X^{(i)}_{n-1})$ is the sum of all Hamiltonian cycles which visit the edge $(1,i)$. Hence,
    \[HC_n(\vecx) = \sum_{i=2}^{n}x_{1,i}\big(\sum_{\substack{\sigma \in C_n, \\ \sigma(1) = i}}\prod_{j=2}^{n}x_{j,\sigma(j)}\big) = \sum_{i=2}^{n}x_{1,i}HC_{n-1}(X^{(i)}_{n-1}), \]   
completing the proof of the forward direction. For the reverse direction, we use induction on $n$.
    
\emph{Base case $n = 2$.} We have by assumption that $f(Y^{(2,2)}) = f(X_2) =  x_{1,2}x_{2,1} = HC_2(X_2)$.
    
\emph{Inductive hypothesis.} Suppose the reverse direction holds for all $k \in [2,n-1]$. Let $f$ be such that it satisfies all $n-1$ identities, where $n > 2$. Then, we have that $f(Y^{(n,k)}) = HC_k(X_k)$ for $k \in [2,n-1]$. Thus, $HC_{n-1}(X^{(i)}_{n-1}) = HC_n(Y^{(n,n-1)}_i) = f(Y^{(n,n-1)}_i)$. Hence, 
    \[f(\vecx) = f(Y^{(n,n)}) = \sum_{i=2}^{n}x_{1,i}f(Y^{(n,n-1)}_i) = \sum_{i=2}^{n}x_{1,i}HC_{n-1}(X^{(i)}_{n-1}) = HC_n(\vecx).\] 
This proves the hypothesis for $n$.  

\subsection{Analysis for $HC_4$} \label{sec-HC4-results}
We first analyse $\mathfrak{g}_{HC_4}$ in Section \ref{subsec-lie-HC4} over all fields. We then determine the scaling symmetries of $HC_4$ over finite fields, $\Q$, $\R$ and $\C$ in Section \ref{subsec-Sym-HC4}. In Section \ref{subsec-charsym-HC4}, we show that $HC_4$ and $HC_3$ are characterised by its symmetries over appropriate fields. 

\subsubsection{Lie Algebra} \label{subsec-lie-HC4}
Proposition \ref{prop-lie-struct-HC4} shows that over characteristic $2$ fields, $\dim_{\F}(\mathfrak{g}_{HC_4})$ is $7$, while over other fields $\dim_{\F}(\mathfrak{g}_{HC_4})$ is $6$. 

\begin{proposition} \label{prop-lie-struct-HC4}
    Let $\F$ be any field. If $\text{char}(\F) \neq 2$, then $\dim_{\F}(\mathfrak{g}_{HC_4}) = 6$ and the set $\mathcal{B}_n$ as described in Proposition \ref{prop-lie-alg-struct} is a basis of $\mathfrak{g}_{HC_4}$ over $\F$. If $\text{char}(\F)$ is $2$, then $\dim_{\F}(\mathfrak{g}_{HC_4}) = 7$ and a basis of $\mathfrak{g}_{HC_4}$ over $\F$ is the set of vectors $\mathcal{B}_n \sqcup \{D\}$, represented as the following matrix
    \begin{equation} \label{eqn-basis-lie-HC4}
    \begin{blockarray}{cccccccc}
        & A^{(2)} & A^{(3)} & A^{(4)} & B^{(2)} & B^{(3)} & C  & D  \\
        \begin{block}{c(ccccccc@{\hspace*{5pt}})}
            (1,2) & 1 & 1 & 1 & 1 & 0 & 1 & 1 \\
            (1,3) & 1 & 1 & 1 & 0 & 1 & 0 & 0 \\
            (1,4) & 1 & 1 & 1 & 0 & 0 & 0 & 0  \\
            (2,1) & 1 & 0 & 0 & 1 & 1 & 1 & 0 \\
            (2,3) & 1 & 0 & 0 & 0 & 1 & 1 & 1 \\
            (2,4) & 1 & 0 & 0 & 0 & 0 & 1 & 0  \\
            (3,1) & 0 & 1 & 0 & 1 & 1 & 0 & 1\\
            (3,2) & 0 & 1 & 0 & 1 & 0 & 1 & 0 \\
            (3,4) & 0 & 1 & 0 & 0 & 0 & 0 & 0  \\
            (4,1) & 0 & 0 & 1 & 1 & 1 & 0 & 0 \\
            (4,2) & 0 & 0 & 1 & 1 & 0 & 1 & 0 \\
            (4,3) & 0 & 0 & 1 & 0 & 1 & 0 & 0  \\
        \end{block}
    \end{blockarray}.  
\end{equation}

\end{proposition}
\begin{proof}
Note that Proposition \ref{prop-lie-alg-struct} already describes the structure of any $M \in \mathfrak{g}_{HC_4}$ as a diagonal matrix, the entries of which are the nullspace of the matrix $M^{HC_4}$. We analyse the nullspace of $M^{HC_4}$ over all fields.

\begin{equation} \label{eqn-liemat-HC4}
M^{HC_4} = \begin{blockarray}{ccccccccccccc}
        \begin{block}{c(ccc|ccc|ccc|ccc@{\hspace*{5pt}})}
            (1,2,3,4) & 1 & 0 & 0 & 0 & 1 & 0 & 0 & 0 & 1 & 1 & 0 & 0  \\
            (1,2,4,3) & 1 & 0 & 0 & 0 & 0 & 1 & 1 & 0 & 0 & 0 & 0 & 1  \\
            (1,3,2,4) & 0 & 1 & 0 & 0 & 0 & 1 & 0 & 1 & 0 & 1 & 0 & 0  \\
            (1,3,4,2) & 0 & 1 & 0 & 1 & 0 & 0 & 0 & 0 & 1 & 0 & 1 & 0  \\
            (1,4,2,3) & 0 & 0 & 1 & 0 & 1 & 0 & 1 & 0 & 0 & 0 & 1 & 0  \\
            (1,4,3,2) & 0 & 0 & 1 & 1 & 0 & 0 & 0 & 1 & 0 & 0 & 0 & 1  \\
        \end{block}
    \end{blockarray}  
\end{equation}
\paragraph{Fields of characteristic other than $2$.} By applying Gaussian Elimination on $M^{HC_4}$, we get the following row reduced matrix. Note that each step of Gaussian Elimination uses only additions or subtractions of the rows.
\begin{equation*}                        
    \begin{blockarray}{ccccccccccccc}
        \begin{block}{c(ccc|ccc|ccc|ccc@{\hspace*{5pt}})}
            (1,2,3,4) & 1 & 0 & 0 & 0 & 1 & 0 & 0 & 0 & 1 & 1 & 0 & 0  \\
            (1,2,4,3) & 0 & 0 & 0 & 0 & -1 & 1 & 1 & 0 & -1 & -1 & 0 & 1  \\
            (1,3,2,4) & 0 & 1 & 0 & 0 & 0 & 1 & 0 & 1 & 0 & 1 & 0 & 0  \\
            (1,3,4,2) & 0 & 0 & 0 & 1 & 0 & -1 & 0 & -1 & 1 & -1 & 1 & 0  \\
            (1,4,2,3) & 0 & 0 & 1 & 0 & 1 & 0 & 1 & 0 & 0 & 0 & 1 & 0  \\
            (1,4,3,2) & 0 & 0 & 0 & 0 & 0 & 0 & -2 & 2 & 0 & 2 & -2 & 0  \\
        \end{block}
    \end{blockarray} 
\end{equation*}
As the characteristic is not equal to $2$, $M^{HC_4}$ has row rank $6$ and therefore the dimension of $\mathfrak{g}_{HC_4}$ is $6$ over such fields. The set $\mathcal{B}_n$ described in Proposition \ref{prop-lie-basis} is easily verified to be a basis for $\mathfrak{g}_{HC_4}$. The linear independence of $\mathcal{B}_n$ over any field has already been argued in Section \ref{proof-lie-basis}.

\paragraph{Characteristic $2$ fields.} As before, after Gaussian Elimination the row reduced form of $M^{HC_4}$ is as follows
\begin{equation*}                        
    \begin{blockarray}{ccccccccccccc}
        \begin{block}{c(ccc|ccc|ccc|ccc@{\hspace*{5pt}})}
            (1,2,3,4) & 1 & 0 & 0 & 0 & 1 & 0 & 0 & 0 & 1 & 1 & 0 & 0  \\
            (1,2,4,3) & 0 & 0 & 0 & 0 & 1 & 1 & 1 & 0 & 1 & 1 & 0 & 1  \\
            (1,3,2,4) & 0 & 1 & 0 & 0 & 0 & 1 & 0 & 1 & 0 & 1 & 0 & 0  \\
            (1,3,4,2) & 0 & 0 & 0 & 1 & 0 & 1 & 0 & 1 & 1 & 1 & 1 & 0  \\
            (1,4,2,3) & 0 & 0 & 1 & 0 & 1 & 0 & 1 & 0 & 0 & 0 & 1 & 0  \\
            (1,4,3,2) & 0 & 0 & 0 & 0 & 0 & 0 & 0 & 0 & 0 & 0 & 0 & 0  \\
        \end{block}
    \end{blockarray} 
\end{equation*}
Thus, the first five rows are linearly independent while the last row is a linear combination of the first five rows over characteristic $2$ fields. Hence, the dimension of $\mathfrak{g}_{HC_4}$ is $7$ over such fields. We now show that the columns of the matrix $M$ in \eqref{eqn-basis-lie-HC4} form a basis for $\mathfrak{g}_{HC_4}$. 

The first six columns of $M$ are precisely the basis elements in the non-characteristic $2$ case. It can be easily checked that the last column lies in $\mathfrak{g}_{HC_4}$. The submatrix indexed by the first $7$ rows of $M$ is 
\begin{equation*}                
    \begin{blockarray}{cccccccc}
        & A^{(2)} & A^{(3)} & A^{(4)} & B^{(2)} & B^{(3)} & C  & D  \\
        \begin{block}{c(ccccccc@{\hspace*{5pt}})}
            (1,2) & 1 & 1 & 1 & 1 & 0 & 1 & 1 \\
            (1,3) & 1 & 1 & 1 & 0 & 1 & 0 & 0 \\
            (1,4) & 1 & 1 & 1 & 0 & 0 & 0 & 0  \\
            (2,1) & 1 & 0 & 0 & 1 & 1 & 1 & 0 \\
            (2,3) & 1 & 0 & 0 & 0 & 1 & 1 & 1 \\
            (2,4) & 1 & 0 & 0 & 0 & 0 & 1 & 0  \\
            (3,1) & 0 & 1 & 0 & 1 & 1 & 0 & 1\\
        \end{block}
    \end{blockarray}.  
\end{equation*}
By Gaussian Elimination, it can be verified that the determinant of this submatrix is $1$. Hence, all the columns are linearly independent and thus form a basis for $\mathfrak{g}_{HC_4}$ over characteristic $2$ fields.  \end{proof}

Proposition \ref{prop-lie-descript-char2} describes how the entries of any $\vecz \in \mathfrak{g}_{HC_4}$ are related and can be proved similarly to Corollary \ref{corollary-lie-soln-struct}. Lemma \ref{lemma-ring-soln} is used in analysing the scaling symmetries of $HC_4$ and can be proved similarly as Lemma \ref{lemma-abeliangroup-soln}.

\begin{proposition} \label{prop-lie-descript-char2}
Over characteristic $2$ fields, $\vecz \in \mathfrak{g}_{HC_4}$ if and only if the entries of $\vecz$ satisfy\begin{equation} \label{eqn-lin-form-liebasis-char2}
    \begin{split}
    z_{3,2} &= z_{3,1} + z_{1,2} + z_{2,4} + z_{1,4} + z_{2,1}, \\
    z_{3,4} &= z_{3,1} + z_{1,4} + z_{2,3} + z_{1,3} + z_{2,1}, \\
    z_{4,1} &= z_{1,2} + z_{1,3} + z_{1,4} + z_{2,1} + z_{3,1}, \\
    z_{4,2} &= z_{1,4} + z_{2,3} + z_{3,1} \\
            &= z_{4,1} + z_{1,2} + z_{2,3} + z_{1,3} + z_{2,1}, \\
    z_{4,3} &= z_{1,2} + z_{2,4} + z_{3,1} \\
            &= z_{4,1} + z_{1,3} + z_{2,4} + z_{1,4} + z_{2,1}.
    \end{split}
\end{equation}

\end{proposition}

\begin{lemma} \label{lemma-ring-soln}
    Let $R$ be a commutative ring with identity such that $2$ is invertible in $R$. Consider the system of linear equations given by $M^{(HC_4)}\vecx = 0$ over $R$. Then, the entries of any solution $\vecu$ to this system of equations satisfy \eqref{eqn-lin-form-liebasis}.
\end{lemma}

\begin{proof}
Let $\vecu$ be a solution. We can then write
    \begin{equation*} 
        \begin{blockarray}{c|c}
            \begin{block}{(c|c@{\hspace*{5pt}})}
            M_1 & M_2 \\
            \end{block}
        \end{blockarray}  
        \cdot
        \begin{blockarray}{(c)}
            \vecu_1 \\
            \vecu_2
        \end{blockarray}
        = \veczero,              
    \end{equation*}
where $M_1$ is the $6 \times 6$ matrix indexed by the set of columns $\{(2,4),(3,2),(3,4),(4,1),(4,2),(4,3)\}$, $M_2$ is the matrix indexed by the remaining columns, $\vecu_1$ and $\vecu_2$ are entries of $\vecu$ corresponding to the set of columns of $M_1$ and $M_2$ respectively. Thus, $M_1$ and $M_2$ are as follows
\begin{equation*}
  M_1 =  \begin{blockarray}{cccccc}
        \begin{block}{(cccccc)@{\hspace*{5pt}})}
               0 & 0 & 1 & 1 & 0 & 0  \\
               1 & 0 & 0 & 0 & 0 & 1  \\
               1 & 1 & 0 & 1 & 0 & 0  \\
               0 & 0 & 1 & 0 & 1 & 0  \\
               0 & 0 & 0 & 0 & 1 & 0  \\
               0 & 1 & 0 & 0 & 0 & 1  \\
        \end{block}
    \end{blockarray} 
    \quad
M_2 =   \begin{blockarray}{cccccc}
        \begin{block}{(cccccc)@{\hspace*{5pt}})}
            1 & 0 & 0 & 0 & 1 & 0 \\
            1 & 0 & 0 & 0 & 0 & 1 \\
            0 & 1 & 0 & 0 & 0 & 0 \\
            0 & 1 & 0 & 1 & 0 & 0 \\
            0 & 0 & 1 & 0 & 1 & 1 \\
            0 & 0 & 1 & 1 & 0 & 0 \\
        \end{block}
    \end{blockarray} \\  
\end{equation*}
It can be verified by Gaussian Elimination that $det(M_1)$ is $-2$. Since $2$ is invertible over $R$, then so is $-2$, implying that the adjoint of $M_1$ is well defined over $R$. Multiplying on the left by the adjoint of $M_1$ gives
\begin{equation*}
     \begin{blockarray}{cccccccccccc}
        \begin{block}{(cccccc|cccccc)@{\hspace*{5pt}})}
               -2 & 0 & 0 & 0 & 0 & 0 & 0 & -2 & 2 & 0 & 2 & 0  \\
               0 & -2 & 0 & 0 & 0 & 0 & 2 & -2 & 0 & -2 & 2 & 2 \\
               0 & 0 & -2 & 0 & 0 & 0 & 0 & -2 & 2 & -2 & 2 & 2 \\
               0 & 0 & 0 & -2 & 0 & 0 & -2 & 2 & -2 & 2 & -4 & -2 \\
               0 & 0 & 0 & 0 & -2 & 0 & 0 & 0 & -2 & 0 & -2 & -2 \\
               0 & 0 & 0 & 0 & 0 & -2 & -2 & 2 & -2 & 0 & -2 & -2\\
        \end{block}
    \end{blockarray} 
    \cdot
        \begin{blockarray}{(c)}
            \vecu_1 \\
            \vecu_2
        \end{blockarray} \\  
        = \veczero.  
\end{equation*}
This then implies 
\begin{equation*}
          -2 \cdot \begin{blockarray}{(c)}
            u_{2,4} + u_{1,3} - u_{1,4} - u_{2,3} \\
            u_{3,2} - u_{1,2} + u_{1,3} + u_{2,1} - u_{2,3} - u_{3,1} \\
            u_{3,4} + u_{1,3} - u_{1,4} + u_{2,1} - u_{2,3} - u_{3,1} \\
            u_{4,1} + u_{1,2} - u_{1,3} + u_{1,4} - u_{2,1} + 2u_{2,3} + u_{3,1} \\
            u_{4,2} + u_{1,4} + u_{2,3} + u_{3,1} \\
            u_{4,3} + u_{1,2} - u_{1,3} + u_{1,4} + u_{2,3} + u_{3,1}
        \end{blockarray} \\  
 = \veczero
\end{equation*}
Since $2$ is invertible in the ring, we can cancel out $-2$ which gives
\begin{equation*}
          \cdot \begin{blockarray}{(c)}
            u_{2,4} + u_{1,3} - u_{1,4} - u_{2,3} \\
            u_{3,2} - u_{1,2} + u_{1,3} + u_{2,1} - u_{2,3} - u_{3,1} \\
            u_{3,4} + u_{1,3} - u_{1,4} + u_{2,1} - u_{2,3} - u_{3,1} \\
            u_{4,1} + u_{1,2} - u_{1,3} + u_{1,4} - u_{2,1} + 2u_{2,3} + u_{3,1} \\
            u_{4,2} + u_{1,4} + u_{2,3} + u_{3,1} \\
            u_{4,3} + u_{1,2} - u_{1,3} + u_{1,4} + u_{2,3} + u_{3,1}
        \end{blockarray} \\  
 = \veczero
\end{equation*}
which essentially means that the entries of $\vecu$ satisfy the linear equations in \eqref{eqn-lin-form-liebasis}, proving the lemma statement.
\end{proof}
\subsubsection{Symmetries} \label{subsec-Sym-HC4}
Lemma \ref{prop-dist-eigen} and Proposition \ref{prop: symmetries-PS} can be proved for $HC_4$ similarly to the general case. Thus, over fields $\F$ such that $|\F| >\binom{12}{2}$, the symmetries of $HC_4$ are generated by permutation and scaling matrices. 

The permutation symmetries of $HC_4$ are the same as those described in Proposition \ref{prop-Psym-gens}. Observation \ref{obs-HC4-sym-charnot2} shows that $HC_4$, over fields of characteristic other than $2$ has a scaling symmetry $D$ which is discrete, that is, $D$ does not satisfy \eqref{eqn-scaling-sym-lie}. The observation can be verified easily using \eqref{eqn-scaling-sym-lie}. In particular, over $\Q$, $\R$, $\C$ and over $\F_q$, such that $2 \mid q - 1$, $HC_4$ has a discrete symmetry.\footnote{$2 \mid q-1$ implies $q$ is the power of an odd prime. Hence, $\text{char}(\F_q) \neq 2$ in this case.} Over characteristic $2$ fields, the symmetry $D$ becomes the identity matrix. Lemma \ref{lemma-scalsym-HC4} describes the scaling symmetries of $HC_4$ over $\Q$, $\R$ and finite fields $\F_q$, where $2 \nmid q - 1$, that is, $q = 2^k$ for some $k > 1$. Lemma \ref{lemma-HC4-discretesym-fq} describes the symmetries over finite fields $\F_q$ where $2 \mid q-1$, and its proof can be adapted for $\C$ and general fields $\F$. 

\begin{observation} \label{obs-HC4-sym-charnot2}
    Let $\F$ be such that $\text{char}(\F) \neq 2$. Then,
    \begin{enumerate}
        \item \label{item-dsym-HC4-gen} The diagonal matrix $D$ with $D_{1,2} = D_{2,3} = D_{3,1} = -1$ and the rest of the entries as $1$ is a discrete scaling symmetry of $HC_4$.
        \item \label{item-csym-HC4-gen} The diagonal matrix $S$, with entries satisfying
        \begin{equation} \label{eqn-contsym-HC4-gen}
    \begin{split}
            S_{1,2} &= \frac{S_{4,2}}{S_{2,4}S_{3,2}S_{4,1}S_{4,3}}, \quad S_{1,3} = \frac{1}{S_{2,4}S_{3,2}S_{4,1}}, \quad S_{1,4} = \frac{S_{3,4}S_{4,2}}{S_{2,4}S^2_{3,2}S_{4,1}S_{4,3}}, \\ 
            S_{2,1} &= \frac{S_{2,4}S_{3,2}S_{4,1}}{S_{3,4}S_{4,2}}, \quad S_{2,3} = \frac{S_{2,4}S_{3,2}S_{4,3}}{S_{3,4}S_{4,2}}, \quad S_{3,1} = \frac{S_{3,2}S_{4,1}}{S_{4,2}}, \\ 
    \end{split}
\end{equation}
        is a continuous scaling symmetry of $HC_4$.
    \end{enumerate}
    When $\text{char}(\F) = 2$, then $D$ is the identity matrix.
\end{observation}

\begin{lemma} \label{lemma-scalsym-HC4}
Over $\F_q$ such that $2 \nmid q-1$,  every scaling symmetry $S$ of $HC_4$ is continuous, that is, the $S_{i,j}$'s satisfy the equations in \eqref{eqn-scaling-sym-lie}. Over $\R$ and $\Q$, the entries $S_{i,j}$ are of the form $(-1)^{z_{i,j}} \cdot r_{i,j}$ with $z_{i,j} \in \F_2$ and $r_{i,j} > 0$ such that the $r_{i,j}$'s satisfy the equations in \eqref{eqn-scaling-sym-lie} and $z_{i,j}$'s satisfy \eqref{eqn-lin-form-liebasis-char2}. The symmetries over $\Q$ and $\R$ given by $S'_{i,j} = (-1)^{z_{i,j}}$ are discrete symmetries if and only if $z_{2,4} \neq z_{1,4}+z_{2,3}+z_{1,3}$.
\end{lemma}
\begin{proof}
Let $S$ be a scaling symmetry of $HC_4$. Then we have
    \[ \prod_{i=1}^{4} S_{i,\sigma(i)} = 1  \quad \sigma \in C_4.\]
    \emph{Over $\F_q$.} Note that $2 \nmid q-1$ implies $q = 2^k$. Let $\gamma$ be a generator of $\F_q^{\times}$. Then, we can write $S_{i,j}$ as $\gamma^{z_{i,j}}$. Thus, we can rewrite the above equations as
    \[ \sum_{i=1}^{4} z_{i,\sigma(i)} = 0  \ \ \text{ over } \Z_{q-1},\]
    Since $2 \nmid q-1$, therefore $2$ is invertible in $\Z_{q-1}$. Thus, by Lemma \ref{lemma-ring-soln}, we get that the $z_{i,j}$'s are as described in \eqref{eqn-lin-form-liebasis}. Consequently, the entries $S_{i,j}$ satisfy \eqref{eqn-scaling-sym-lie}.  

    \noindent\emph{Over $\R$.} 
    We can write $S_{i,j}$ as $(-1)^{z_{i,j}} \cdot 2^{w_{i,j}}$, where $2^{w_{i,j}} = r_{i,j}$. Thus, we can rewrite the above equations as
    \[ \sum_{i=1}^{4} z_{i,\sigma(i)} = 0  \ \ \text{ over } \F_2, \text{ and } \sum_{i=1}^{4} w_{i,\sigma(i)} = 0  \ \ \text{ over } \R.\]
    Note that over $\F_2$ and $\R$, we get, from the analysis in Section \ref{subsec-lie-HC4}, that the solutions to both these systems of equations must be described by the $\mathfrak{g}_{HC_4}$ basis elements over $\F_2$ and $\R$, respectively. Hence, the $z_{i,j}$'s satisfy \eqref{eqn-lin-form-liebasis-char2} while the $r_{i,j}$'s satisfy \eqref{eqn-scaling-sym-lie}. Since $\Q$ is a subfield of $\R$, every scaling symmetry over $\Q$ must satisfy the same conditions. The claim regarding $S'_{i,j}$ follows by noting that $z_{2,4} \neq z_{1,4} + z_{2,3} + z_{1,3}$ implies $z_{2,4} = 1 + z_{1,4} + z_{2,3} + z_{1,3}$ and then showing $S'_{i,j}$ does \emph{not} satisfy \eqref{eqn-scaling-sym-lie}.  If $z_{2,4} = z_{1,4} + z_{2,3} + z_{1,3}$, then it is easily verified that $S'_{i,j}$ satisfy \eqref{eqn-scaling-sym-lie}.
\end{proof}
\begin{lemma} \label{lemma-HC4-discretesym-fq}
    Let $\F_q$ be such that $2 \mid q-1$. If $S' \in \mathcal{G}_{HC_4}$ is a scaling symmetry over such fields $\F$, then $S'$ is $D$, $S$, or $DS$, where $D$ and $S$ are as in items \ref{item-dsym-HC4-gen} and \ref{item-csym-HC4-gen} of Observation \ref{obs-HC4-sym-charnot2}. 
\end{lemma}
\begin{proof}
Suppose $S'$ is a scaling symmetry of $HC_4$. Then we get
    \[ \prod_{i=1}^{4} S'_{i,\sigma(i)} = 1  \quad \sigma \in C_4.\]
Let $\gamma$ be a generator of $F_q^{\times}$ and $S'_{i,j} = \gamma^{z_{i,j}}$. Then we can rewrite the above system as
    \[ \sum_{i=1}^{4} z_{i,\sigma(i)} = 0  \ \ \text{ over } \Z_{q-1},\]
We can express this system as   
\begin{equation} \label{eqn-sym-HC4-fq0}
\begin{blockarray}{cccccccccccc}
        \begin{block}{(ccc|ccc|ccc|ccc@{\hspace*{5pt}})}
             1 & 0 & 0 & 0 & 1 & 0 & 0 & 0 & 1 & 1 & 0 & 0  \\
             1 & 0 & 0 & 0 & 0 & 1 & 1 & 0 & 0 & 0 & 0 & 1  \\
              0 & 1 & 0 & 0 & 0 & 1 & 0 & 1 & 0 & 1 & 0 & 0  \\
              0 & 1 & 0 & 1 & 0 & 0 & 0 & 0 & 1 & 0 & 1 & 0  \\
              0 & 0 & 1 & 0 & 1 & 0 & 1 & 0 & 0 & 0 & 1 & 0  \\
             0 & 0 & 1 & 1 & 0 & 0 & 0 & 1 & 0 & 0 & 0 & 1  \\
        \end{block}
    \end{blockarray}  
    \cdot \vecz
     = \veczero.
\end{equation}
Let $R_i$ denote the $i$'th row of the above matrix. Applying the elementary row operations, $R_2 \to R_2 - R_1$, $R_4 \to R_4 - R_3$, $R_6 \to R_6 - R_5 - R_4 - R_2$  gives us \eqref{eqn-sym-HC4-fq1}. Note that all these operations are invertible, hence the solution set is the same.
\begin{equation} \label{eqn-sym-HC4-fq1}
    \begin{blockarray}{ccccccccccccc}
        \begin{block}{c(ccc|ccc|ccc|ccc@{\hspace*{5pt}})}
            & 1 & 0 & 0 & 0 & 1 & 0 & 0 & 0 & 1 & 1 & 0 & 0  \\
             & 0 & 0 & 0 & 0 & -1 & 1 & 1 & 0 & -1 & -1 & 0 & 1  \\
            & 0 & 1 & 0 & 0 & 0 & 1 & 0 & 1 & 0 & 1 & 0 & 0  \\
             & 0 & 0 & 0 & 1 & 0 & -1 & 0 & -1 & 1 & -1 & 1 & 0  \\
             & 0 & 0 & 1 & 0 & 1 & 0 & 1 & 0 & 0 & 0 & 1 & 0  \\
             & 0 & 0 & 0 & 0 & 0 & 0 & -2 & 2 & 0 & 2 & -2 & 0  \\
        \end{block}
    \end{blockarray} 
      \cdot  \vecz = \veczero .
\end{equation}
It is not hard to see that $\vecz$ is a solution to \eqref{eqn-sym-HC4-fq1} if and only if $\vecz$ is a solution to \eqref{eqn-sym-HC4-fq2} or is in the nullspace (over $\Z_{q-1})$ of the matrix in \eqref{eqn-sym-HC4-fq2}. This follows because $\frac{q-1}{2}$ is the only element of order $2$ under addition in $\Z_{q-1}$.

\begin{equation}  \label{eqn-sym-HC4-fq2}                      
    \begin{blockarray}{cccccccccccc}
        \begin{block}{(ccc|ccc|ccc|ccc@{\hspace*{5pt}})}
             1 & 0 & 0 & 0 & 1 & 0 & 0 & 0 & 1 & 1 & 0 & 0  \\
              0 & 0 & 0 & 0 & -1 & 1 & 1 & 0 & -1 & -1 & 0 & 1  \\
             0 & 1 & 0 & 0 & 0 & 1 & 0 & 1 & 0 & 1 & 0 & 0  \\
              0 & 0 & 0 & 1 & 0 & -1 & 0 & -1 & 1 & -1 & 1 & 0  \\
              0 & 0 & 1 & 0 & 1 & 0 & 1 & 0 & 0 & 0 & 1 & 0  \\
              0 & 0 & 0 & 0 & 0 & 0 & -1 & 1 & 0 & 1 & -1 & 0  \\
        \end{block}
    \end{blockarray} 
      \cdot  \vecz = 
      \begin{blockarray}{c}
          \begin{block}{(c)}
              0 \\
              0 \\
              0 \\
              0 \\
              0 \\
              \frac{q-1}{2} \\
          \end{block}
      \end{blockarray}
\end{equation}
Applying the elementary row operations, $R_5 \to R_5 + R_2$, $R_1 \to R_1 + R_2$, $R_5 \to R_5 + 2R_6$, $R_2 \to R_2 + R_6$ and $R_1 \to R_1 + R_6$ gives us \eqref{eqn-sym-HC4-fq3}. These operations are invertible, hence the solution set remains unchanged.
\begin{equation}  \label{eqn-sym-HC4-fq3}                      
    \begin{blockarray}{cccccccccccc}
        \begin{block}{(ccc|ccc|ccc|ccc@{\hspace*{5pt}})}
             1 & 0 & 0 & 0 & 0 & 1 & 0 & 1 & 0 & 1 & -1 & 1  \\
              0 & 0 & 0 & 0 & -1 & 1 & 0 & 1 & -1 & 0 & -1 & 1  \\
             0 & 1 & 0 & 0 & 0 & 1 & 0 & 1 & 0 & 1 & 0 & 0  \\
              0 & 0 & 0 & 1 & 0 & -1 & 0 & -1 & 1 & -1 & 1 & 0  \\
              0 & 0 & 1 & 0 & 0 & 1 & 0 & 2 & -1 & 1 & -1 & 1  \\
              0 & 0 & 0 & 0 & 0 & 0 & -1 & 1 & 0 & 1 & -1 & 0  \\
        \end{block}
    \end{blockarray} 
      \cdot  \vecz = 
      \begin{blockarray}{c}
          \begin{block}{(c)}
              \frac{q-1}{2} \\
              \frac{q-1}{2} \\
              0 \\
              0 \\
              0 \\
              \frac{q-1}{2} \\
          \end{block}
      \end{blockarray}
\end{equation}
Thus, from \eqref{eqn-sym-HC4-fq3}, we get that any solution $\vecz$ to the system of equations in \eqref{eqn-sym-HC4-fq1} is of the form
\begin{equation} \label{eqn-sol1-sym-HC4}
    \begin{split}
    z_{1,2} &= \frac{q-1}{2} - z_{2,4} - z_{3,2} - z_{4,1} + z_{4,2} -z_{4,3}\\
    z_{1,3} &= -z_{2,4}-z_{3,2}-z_{4,1}\\
    z_{1,4} &= -z_{2,4}-2z_{3,2}+z_{3,4}-z_{4,1} + z_{4,2} - z_{4,3} \\
    z_{2,1} &=  z_{2,4} + z_{3,2} - z_{3,4} + z_{4,1}-z_{4,2} \\
    z_{2,3} &= \frac{q-1}{2} + z_{2,4} + z_{3,2} - z_{3,4} - z_{4,2} + z_{4,3}, \\
    z_{3,1} &= \frac{q-1}{2} + z_{3,2} + z_{4,1} - z_{4,2}.
   \end{split}
\end{equation}
or is of form
\begin{equation} \label{eqn-sol2-sym-HC4}
\begin{split}
    z_{1,2} &= -z_{2,4} - z_{3,2} - z_{4,1} + z_{4,2} -z_{4,3}\\
    z_{1,3} &= -z_{2,4}-z_{3,2}-z_{4,1}\\
    z_{1,4} &= -z_{2,4}-2z_{3,2}+z_{3,4}-z_{4,1} + z_{4,2} - z_{4,3} \\
    z_{2,1} &=  z_{2,4} + z_{3,2} - z_{3,4} + z_{4,1}-z_{4,2} \\
    z_{2,3} &=  z_{2,4} + z_{3,2} - z_{3,4} - z_{4,2} + z_{4,3}, \\
    z_{3,1} &=  z_{3,2} + z_{4,1} - z_{4,2}.
\end{split}
\end{equation}
Thus, we get that $S'_{i,j} = \gamma^{z_{i,j}}$ are either as
\begin{equation} \label{eqn-sol-sym-HC4-finalform1}
\begin{split}
 S'_{1,2} &= -\frac{S'_{4,2}}{S'_{2,4}S'_{3,2}S'_{4,1}S'_{4,3}}, \quad S'_{1,3} = \frac{1}{S'_{2,4}S'_{3,2}S'_{4,1}}, \quad S'_{1,4} = \frac{S'_{3,4}S'_{4,2}}{S'_{2,4}S'^2_{3,2}S'_{4,1}S'_{4,3}}, \\ 
            S'_{2,1} &= \frac{S'_{2,4}S'_{3,2}S'_{4,1}}{S'_{3,4}S'_{4,2}}, \quad S'_{2,3} = -\frac{S'_{2,4}S'_{3,2}S'_{4,3}}{S'_{3,4}S'_{4,2}}, \quad S'_{3,1} = -\frac{S'_{3,2}S'_{4,1}}{S'_{4,2}}, \\     
\end{split}
\end{equation}
or as
\begin{equation} \label{eqn-sol-sym-HC4-finalform2}
\begin{split}
 S'_{1,2} &= -\frac{S'_{4,2}}{S'_{2,4}S'_{3,2}S'_{4,1}S'_{4,3}}, \quad S'_{1,3} = \frac{1}{S'_{2,4}S'_{3,2}S'_{4,1}}, \quad S'_{1,4} = \frac{S'_{3,4}S'_{4,2}}{S'_{2,4}S'^2_{3,2}S'_{4,1}S'_{4,3}}, \\ 
            S'_{2,1} &= \frac{S'_{2,4}S'_{3,2}S'_{4,1}}{S'_{3,4}S'_{4,2}}, \quad S'_{2,3} = -\frac{S'_{2,4}S'_{3,2}S'_{4,3}}{S'_{3,4}S'_{4,2}}, \quad S'_{3,1} = -\frac{S'_{3,2}S'_{4,1}}{S'_{4,2}}. \\     
\end{split}
\end{equation}
Hence, any scaling symmetry $S'$ is of the form $D$, $S$, or $DS$, where $D$ and $S$ are as per items \ref{item-dsym-HC4-gen} and \ref{item-csym-HC4-gen} of Observation \ref{obs-HC4-sym-charnot2} respectively. 
\end{proof}
\paragraph{Over $\C$ and general fields $\F$.} \label{para-HC4sym-F} The proof of Lemma \ref{lemma-HC4-discretesym-fq} can be adapted to work over $\C$, and general fields $\F$ where $\text{char}(\F) \neq 2$, by treating the system of equations given by $\prod_{i=1}^{4} S'_{i,\sigma(i)} = 1$ over $\C$ (resp. $\F$) as a system of linear equations over $\C^{\times}$ (resp. $\F^{\times}$) where $+$ corresponds to multiplication and $0$ is the multiplicative identity $1$. Thus, we can represent this system just like \eqref{eqn-sym-HC4-fq0}, where the entries of $\vecz$ are the $S'_{i,j}$'s and $\veczero$ is replaced by the all ones vector $\vecone$. Applying Gaussian Elimination gives \eqref{eqn-sym-HC4-fq1}. Then, $S'$ is a solution to \eqref{eqn-sym-HC4-fq1} if and only if $S'$ is a solution to \eqref{eqn-sym-HC4-fq2}, with $(q-1)/2$ replaced by $-1$ since $x^2 = 1$ has $1$ and $-1$ as the only solutions over $\F$, or $S'$ is in the nullspace (over $\F^{\times}$) of the matrix in \eqref{eqn-sym-HC4-fq2}. Then, we further perform elementary row operations to obtain \eqref{eqn-sym-HC4-fq3}, which then shows that $S'_{i,j}$'s satisfy \eqref{eqn-sol-sym-HC4-finalform1} or \eqref{eqn-sol-sym-HC4-finalform2}. Thus, $S'$ must be of form $D$,$S$, or $DS$.

Over fields $\F$ where $\text{char}(\F) = 2$, the matrix $D$ in the proof becomes the identity matrix. Hence, we get that all the scaling symmetries of $HC_4$ are continuous over such $\F$.

\subsubsection{Characterisation by symmetries} \label{subsec-charsym-HC4}
The existence of the symmetry $D$ as described in item \ref{item-dsym-HC4-gen} of Observation \ref{obs-HC4-sym-charnot2} helps show that $HC_4$ is characterised by its symmetries over large enough fields $\F$ with $\text{char}(\F) \neq 2$ as we show in Proposition \ref{prop-symchar-HC4}. We use the approach in the proof of Proposition 3.4.5 in \cite{grochowPhD}, where it is shown that the Permanent is characterised by its symmetries. When $\text{char}(\F) = 2$, then Proposition \ref{prop-nonsymchar-HC4-F2} shows that $HC_4$ is \emph{not} characterised by its symmetries. Proposition \ref{prop-symchar-HC3} shows $HC_3$ is characterised by its symmetries when $|\F| > \binom{6}{2}$.

\begin{proposition} \label{prop-symchar-HC4}
    Over any field $\F$ such that $|\F| > \binom{12}{2}$ and $\text{char}(\F) \neq 2$, $HC_4$ is characterised by its symmetries.
\end{proposition}
\begin{proof}
    The assumption on the field size ensures that $\mathcal{G}_{HC_4}$ is generated by permutation and scaling matrices. Suppose $f(\vecx)$ is a $12$-variate homogeneous\footnote{A polynomial $f$ is homogeneous if all the monomials in $f$ have the same degree.} degree-$4$ polynomial such that $\mathcal{G}_{HC_4} \subseteq \mathcal{G}_f$. Let $m$ be a monomial in $f(\vecx)$ and $x_{i,j}$ be a variable in $m$, with $i \neq j$. Let $c \in \F^{\times}$, $k \in [4]$ and $k \neq i$. Apply the mapping $x_{k,\ell} \mapsto c\cdot x_{k,\ell}$ for all $\ell \in [4] \backslash \{k\}$, and $x_{\ell,j} \mapsto c^{-1}\cdot x_{\ell,j}$ for all $\ell \neq j$. It is easily verified that for all $k \neq i$, the aforementioned mapping is a scaling symmetry of $HC_4$ satisfying item \ref{item-csym-HC4-gen} of Observation \ref{obs-HC4-sym-charnot2}, and therefore is also a scaling symmetry of $f$. This implies $m$ must contain for all $k \neq i$ some variable $x_{k,\ell}$ because $x_{i,j}$ is scaled down by $c$ and if all variables in $m$ are of form $x_{i, \ell}$ then the above mapping is not a scaling symmetry of $f$, a contradiction. Thus, $m$ contains, for all $k \in [4]$, some $x_{k, \ell}$. Since $f$ is homogeneous, $m$ must contain exactly one $x_{k, \ell}$ for all $k \in [4]$. 
    
    The above argument can be adapted to further show that $m$ must correspond to a permutation $\pi \in S_4$ with $\pi(i) \neq i$ for all $i \in [4]$, since $\vecx$ does \emph{not} contain $x_{i,i}$ variables. Further, consider $f(D\vecx)$, where $D$ is as described in item \ref{item-dsym-HC4-gen} of Observation \ref{obs-HC4-sym-charnot2}. If $m$ corresponded to one of the non-cyclic permutations $(1 \ 2)(3 \ 4)$, $(1 \ 3)(2 \ 4)$, or $(1 \ 4)(2 \ 3)$, then it is not hard to see that under $D$, $m$ gets mapped to $-m$, a contradiction. Thus, $m$ corresponds to a cyclic permutation in $C_4$. Finally, applying all the permutation symmetries of $HC_4$ on $f$ shows that for every $\sigma \in C_4$, there must be a monomial corresponding to it present in $f$. Hence $f = \alpha. HC_4$ for some $\alpha \in \F^{\times}$.
\end{proof}

\begin{proposition} \label{prop-nonsymchar-HC4-F2}
    Over any field $\F$ such that $|\F| > \binom{12}{2}$ and $\text{char}(\F) = 2$, $HC_4$ is \emph{not} characterised by its symmetries.
\end{proposition}
\begin{proof}
The condition $|\F| > \binom{12}{2}$ ensures $\mathcal{G}_{HC_4}$ is generated by permutation and scaling symmetries. Since $\text{char}(\F)=2$, all scaling symmetries of $HC_4$ are continuous by the modifications to the proof of Lemma \ref{lemma-HC4-discretesym-fq} as suggested in the paragraph \nameref{para-HC4sym-F} in Appendix \ref{subsec-Sym-HC4}. Over such fields, $HC_4$ is \emph{not} characterised by its symmetries because for the polynomial 
\[f(\vecx) = x_{1,2}x_{2,1}x_{3,4}x_{4,3}+x_{1,3}x_{3,1}x_{2,4}x_{4,2}+x_{1,4}x_{4,1}x_{3,2}x_{2,3},\] it can be easily verified that $\mathcal{G}_{HC_4} \subseteq \mathcal{G}_{f}$ but $f \neq c HC_4$ for any $c \in \F^{\times}$.   
\end{proof}

\begin{proposition} \label{prop-symchar-HC3}
    Over any field $\F$ such that $|\F| > \binom{6}{2}$, $HC_3$ is characterised by its symmetries.
\end{proposition}
\begin{proof}
    The assumption on the field size ensures that $\mathcal{G}_{HC_3}$ is generated by permutation and scaling matrices. Suppose $f(\vecx)$ is a $6$-variate homogeneous degree-$3$ polynomial such that $\mathcal{G}_{HC_3} \subseteq \mathcal{G}_f$. Applying the same scaling map as in the proof of Proposition \ref{prop-symchar-HC4}, it can be shown that if $m$ is a monomial of $f$, then $m$ corresponds to a permutation $\sigma \in S_3 \backslash C_3$, such that $\sigma(i) \neq i$ for all $i \in [3]$. This immediately implies that $\sigma \in C_3$. Thus, $m$ corresponds to a cyclic permutation. Finally, applying all the permutation symmetries of $HC_3$ on $m$ shows that for every $\sigma \in C_3$, there must be a monomial corresponding to it present in $f$. Hence $f = \alpha. HC_3$ for some $\alpha \in \F^{\times}$. 
\end{proof}

\section{Missing Proofs from Section \ref{sec-ETalgo}} \label{sec-ETalgo-proofs}
\subsection{Proof of Proposition \ref{prop:PS-equiv-reduction}} \label{proof-PSequiv}
    If $f = HC_n(A\vecx)$, then $\mathfrak{g}_f$ is $2n-2$ dimensional by Lemma \ref{lemma-lieconjug} and Proposition \ref{prop-lie-basis}. As $C = \sum_{i=1}^{2n-2} a_i B_i$, then $B' = ACA^{-1} = \sum_{i=1}^{2n-2} a_i AB_iA^{-1}$. By Lemma \ref{lemma-lieconjug}, $B'$ lies in $\mathfrak{g}_{HC_n}$ and is a diagonal matrix with its entries being an $\F$-linear combination of $a_i$'s. The matrices $A B_i A^{-1}$ are a basis for $\mathfrak{g}_{HC_n}$. Since $|\F| > \binom{n^2-n}{2}$, by Lemma \ref{prop-dist-eigen}, there is a matrix in $\mathfrak{g}_{HC_n}$ with distinct eigenvalues. Thus, there exists a choice of $a_i$'s such that $B'$ has distinct eigenvalues. By the Polynomial Identity Lemma, $B'$, over a random choice of $a_i$'s from $U$, has distinct eigenvalues with probability $\geq 1 -\frac{\binom{n^2-n}{2}}{|U|} > 1 - \frac{1}{3n}$. Thus, $C$ has distinct eigenvalues with high probability because $C$ and $B'$ are similar matrices. Hence, they also have the same characteristic polynomial which factorises over $\F$.
    
    Therefore, there exists $D \in \GL_{n^2-n}(\F)$ such that $C' =D^{-1} C D$ is diagonal. We can compute $D$ by computing the characteristic polynomial of $C$, factorising the polynomial to get the eigenvalues $\lambda_1, \dots, \lambda_{n^2-n}$, which are all distinct, and then finding the basis vector of the nullspace of $C - \lambda_iI_{(n^2-n) \times (n^2-n)}$ for all $i \in [n^2-n]$. The matrix $D$ is formed from all these basis vectors. 
    
    Note that $D$ also diagonalises the $B_i$'s. Hence, the polynomial $g = f(D\vecx) = HC_n(AD\vecx)$ is such that $\mathfrak{g}_g$ is diagonal. By Lemma \ref{lemma-lieconjug}, we get that
    \[C' = D^{-1}CD = (AD)^{-1}\cdot B' \cdot AD,\]
    with $C'$ and $B'$ being diagonal matrices with distinct entries. Then, following the same argument as in the proof of Proposition \ref{prop: symmetries-PS} (see Appendix \ref{proof-symmetries-PS}), we get that $AD = PS$ for some permutation matrix $P$ and scaling matrix $S$.

\paragraph{Complexity analysis.} Step 1 can be performed via Lemma \ref{lemma-liebasis-comp} in $n^{O(1)}$ time. For Step 2, we need access to an efficient univariate polynomial factorisation algorithm to compute $D$. The running time is randomised $n^{O(1)}$.
    
\subsection{Proof of Claim \ref{claim-Ptest-assum}} \label{proof-Ptest-assum}
    Let $P(x_{i,j}) = x_{1,2}$. Let $\sigma \in S_n$ such that $\sigma(1) = i$ and $\sigma(2)=j$. Then $P^{(\sigma)}$, defined as per Proposition \ref{prop-psym}, is a permutation symmetry of $HC_n$. Thus, $HC_n(P^{(\sigma)}\vecx) = HC_n(\vecx)$ implies $HC_n(P^{(\sigma)}P\vecx) = HC_n(P\vecx)$. It then follows that $P^{(\sigma)}P(x_{1,2}) = x_{1,2}$ because
    \[(P^{(\sigma)}P)_{(1,2),(1,2)} = \sum_{a,b \in [n], a \neq b}P^{(\sigma)}_{(1,2),(a,b)}P_{(a,b),(1,2)}\]
    \[ = P^{(\sigma)}_{(1,2),(i,j)}P_{(i,j),(1,2)} = P^{(\sigma)}_{(1,2),(\sigma(1),\sigma(2))}P_{(i,j),(1,2)} = 1.\]
    The third last equality holds because $P_{(i,j)(1,2)} = 1$ by assumption, and the definition of $\sigma$ ensures that the second last equality holds.

\subsection{Proof of Proposition \ref{prop-ptest-correct}} \label{proof-ptest-correct}
Step 1 of Algorithm \ref{alg2} errs with probability at most $\frac{n^5}{|U|}$, where $U \subset \F$, by the Polynomial Identity Lemma and union bound. For $|U| \geq 3n^5$, the error is upper bounded by $\frac{1}{3}$. Repeating the algorithm will boost the success probability. 

The correctness of Steps 2 and 3 follows from Claim \ref{claim-Ptest-assum} and Observation \ref{obs-HCn-pdzero}. By Observation \ref{obs-HCn-pdzero} we have that $P'(x_{2,1})$ must be $x_{i,j}$, proving the correctness of Step 4. Now, the set $T'_{1,2}$ is either the image of $T_{1,2}$ or that of $Q_{1,2}$ under $P$. 

In the first case, the elements of $T'_{1,2}$ are of form $P(x_{1,t})$, $t \in [3,n]$. By mapping  $x_{1,t}$'s to some permutation $\pi$ of the elements of $T'_{1,2}$, we get that $P'P(x_{1,t}) = x_{1,\pi(t)}$. Step 5 then sets $P'(x_{t,1})$ such that $P'P(x_{t,1}) = x_{\pi(t),1}$. Extend $\pi$ to a permutation on $[n]$ by defining $\pi(1) = 1, \pi(2) = 2$. For the remaining $x_{a,b}$ with $a,b \in [2,n]$, we set $P'(x_{a,b})$ consistently in Step 6 so that we get $P'P(x_{a,b}) = x_{\pi(a),\pi(b)}$. Thus, we get that $P'P = P^{(\pi)}$. 

In the second case, the elements of $T'_{1,2}$ are of form $P(x_{t,2})$, $t \in [3,n]$. By mapping  $x_{1,t}$'s to some permutation $\pi$ of the elements of $T'_{1,2}$, we get that $P'P(x_{1,t}) = x_{\pi(t),2}$. Step 5 then sets $P'(x_{t,1})$ such that $P'P(x_{t,1}) = x_{2,\pi(t)}$. Extend $\pi$ to a permutation on $[n]$ by defining $\pi(1) = 1, \pi(2) = 2$. For the remaining $x_{a,b}$ with $a,b \in [2,n]$, we set $P'(x_{a,b})$ consistently in Step 6 so that again we get $P'P(x_{a,b}) = x_{\pi(b),\pi(a)}$. Thus, we get that $P'P = P^{(\pi)}P^{(T)}$.

\paragraph{Complexity analysis.} Step 1 of Algorithm \ref{alg2} uses Lemma \ref{lemma-pd-compute} to compute black-box access to second-order partial derivatives and uses the Polynomial Identity Lemma to check if the derivatives are identically zero, and can be performed in randomised $n^{O(1)}$ time. The remaining steps can be performed in $O(n^4)$ time. Hence, Algorithm \ref{alg2} has running time randomised $n^{O(1)}$.

\subsection{Proof of Claim \ref{claim-Ssym-assum}} \label{proof-Ssym-assum}
As $f(\vecx) = HC_n(S\vecx)$, we define another scaling matrix $D$ by setting $D_{1,j} :=  S_{1,j}^{-1}$ for $j \in [2,n]$, $D_{2,3} = S^{-1}_{2,3}$ and $D_{i,1} = S^{-1}_{i,1}$, where $i \in [2,n-1]$. Set the rest of the $D_{i,j}$'s as described in \eqref{eqn-scaling-sym-lie}. Then, $D \in \mathcal{G}_{HC_n}$ by construction, hence $HC_n(D\vecx) = HC_n(\vecx)$ implies $f = HC_n(S\vecx) = HC_n(DS\vecx)$. Set $S' := DS$ and note that $S'$ satisfies the observation statement.
\subsection{Proof of Proposition \ref{prop-stest-correct}}
\label{proof-stest-correct}
Suppose $\gamma$ is a generator of $\F^{\times}_{q}$. Then we can write $S'_{i,j}$ as $\gamma^{y_{i,j}}$ and $c_{k}$ as $\gamma^{e_{k}}$. This gives the following system of equations over $\Z_{q-1}$
\[\sum_{i=1}^{n} y_{i,\sigma_k(i)} = e_{k} \text{ mod } q-1.\]
In terms of the matrix $M^{(n)}$, the above system can be written as 
\[M^{(n)}. \vecy = \vece\]
where the entries of $\vece$ are the exponents $e_k$. As $\Z_{q-1}$ is a commutative ring, by Lemma \ref{lemma-abeliangroup-general} any solution to the above system of linear equations will give a scaling matrix. Note that $y_{i,j} = 0$ for all $(i,j) \notin T$ by assumption, with $T$ as described in Step 3 of Algorithm \ref{alg3}. Thus, we get a system of $(n-1)(n-2)$ many linear equations in $(n-1)(n-2)$ many $\vecy$ variables over $\Z_{q-1}$ which can be written as
\[M.\vecy = \vece\]
 with $M = M^{(n)}_{\bullet \times T}$ and $det(M) = \beta = \pm 1$ by Lemma \ref{lemma-Mn-det1}. From Cramer's rule, we get that
\[y_{i,j} = \beta \cdot \sum_{k = 1}^{(n-1)(n-2)}e_{k}\alpha_{k,(i,j)} \text{ mod } q-1.\]
Then,
\[S'_{i,j} = \gamma^{y_{i,j}} = \gamma^{\beta \cdot \sum_{k = 1}^{(n-1)(n-2)}e_{k}\alpha_{k,(i,j)} \text{ mod } q-1}.\]
As $c_{k} = \gamma^{e_{k}}$, therefore
\[S'_{i,j} = \prod_{k = 1}^{(n-1)(n-2)}c^{\beta \cdot \alpha_{k,(i,j)} \text{ mod } q-1}_{k}. \]
Clearly, $HC_n(S'\vecx) = HC_n(S\vecx)$ which implies $S'S^{-1}$ is a scaling symmetry of $HC_n$. Note that Algorithm \ref{alg3} does \emph{not} need to compute the discrete logarithm of any element. The algorithm computes $\beta$ and $\alpha_{k,(i,j)}$'s over $\Z_{q-1}$ and raises the coefficients $c_k$ to powers as described above to compute $S'_{i,j}$.

\paragraph{Complexity analysis.} Step 2 can be executed in $n^{O(1)}$ time as shown in the proof of Proposition \ref{prop-lie-dim-lb}. Step 3 and 4 can be performed in $(n \log q)^{O(1)}$ time as $M$ is a 0/1 matrix of dimension $(n-1)(n-2) \times (n-1)(n-2)$, implying the magnitude $|\alpha_{k,(i,j)}|$ of any cofactor is at most $((n-1)(n-2))!$ and has bit complexity $O(n^2 \log n \log q)$. Since $\beta = \pm 1$ by Lemma \ref{lemma-Mn-det1}. Thus, we can compute the cofactors in $(n \log q)^{O(1)}$ time over $\Z_{q-1}$, and the powers of the coefficients $c_k$ by repeated squaring. Thus, the running time is deterministic $(n\log q)^{O(1)}$.

\subsubsection{Correctness and Complexity over $\Q$ and other fields} \label{proof-Sequiv-allfields}
To make Algorithm \ref{alg3} work over other fields, the main changes are in Steps 3 and 4, where all computations involving the matrix $M$ are over $\Z$ now. The argument given above for finite fields can be adapted to argue the correctness over all fields $\F$ by noting that the system we want to solve is 
\[\prod_{i=1}^{n} S'_{i,\sigma_k(i)} = c_{k}\]
which is a system of linear equations over the group $\F^{\times}$, where we interpret multiplication as addition and $1$ as $0$. This system of linear equations can be written as 
\[M^{(n)}.\vecs = \vecc,\] 
where the entries of $\vecs$ are $S'_{i,j}$ and $\vecc$ are the coefficients $c_k$. Lemma \ref{lemma-abeliangroup-general} shows that solving the above system of equations will produce an $S'$ such that $HC_n(S'\vecx) = f(\vecx)$. By Claim \ref{claim-Ssym-assum}, we reduce to solving
\[M.\vecs_1 = \vecc,\] 
where $M = M^{(n)}_{.\times T}$ and the entries of $\vecs_1$ are $S'_{i,j}$ with $(i,j) \in T$. Since $M$ is a 0/1 matrix, and by Lemma \ref{lemma-Mn-det1} $det(M)$ is $\pm 1$, we get that $M^{-1}$ is an integer matrix. We can thus multiply by $M^{-1}$ on the left to get
\[\vecs_1 = M^{-1}\vecc.\]
Note, since we treat $+$ as the multiplication operation, $M^{-1}\vecc$ is computed by raising $c_k$'s or their inverses to integral powers as determined by the rows of $M^{-1}$. Thus, each $S'_{i,j}$ is obtained by multiplying integral powers of the $c_k$'s.  
\paragraph{Complexity analysis.} For the complexity analysis, we can perform Steps 3 and 4 in $n^{O(1)}$ time using repeated squaring required for each $S_{i,j}$. Over $\Q$, the running time is $(nc)^{O(1)}$, where $c$ is the bit complexity of $f$. Over $\R$, $\C$, and general fields $\F$, assuming that we can perform exact arithmetic, we get $n^{O(1)}$ running time.  

\subsection{$S$-equivalence test for $HC_4$} \label{subsec-Sequiv-HC4}
In this section, we show how to perform $S$-equivalence test for $HC_4$ over $\Q$, $\R$, and finite fields $\F_q$.
\paragraph{Over $\F_q$ when $2 \nmid q - 1$.} If $2 \nmid q - 1$, then Algorithm \ref{alg3} continues to work for $HC_4$, because $2$ is invertible in the ring $\Z_{q-1}$ and Lemma \ref{lemma-ring-soln} continues to hold. 

\paragraph{Over $\Q$, $\R$, and $\F_q$ where $2 \mid q - 1$.} We will present the test over $\Q$ and the same will also hold over $\R$ and $\F_q$. Over $\Q$, we assume that rational numbers are given in the form $\frac{a}{b}$, where $a,b \in \Z$, $b \neq 0$. Over $\R$, we assume that the model of computation is able to compute positive square roots \cite{Brent76}. Over $\F_q$, where $\text{char}(\F_q) > 2$ because $2 \mid q - 1$, square root computation can be done in randomised $\log^{O(1)}(q)$ time \cite{Tonelli1891,Shanks73} and in deterministic $\log^{O(1)}(q)$ time assuming the Generalised Reimann Hypothesis \cite{Bach90}.

The main idea is to solve the system of linear equations arising from the coefficients of each monomial. Suppose $f = HC_4(S\vecx)$. Treat the permutations $\sigma \in C_4$ as 4-tuples of the form $(1,i_2,i_3,i_4)$ and order them lexicographically. Then, we have that  
    \[ \prod_{i=1}^{4} S_{i,\sigma_j(i)} = c_{j}  \quad j \in [6].\]
with $c_j \in \Q^{\times}$ as the coefficient corresponding to $\sigma_j$, which can be obtained from $f(\vecx)$ by querying it at the respective monomial. Each $c_{j}$ is of form $c_j = (-1)^{e_j} r_j$, where $r_j > 0$. We create a new diagonal matrix $S'$, with 
$S'_{i,k} = (-1)^{z_{i,k}}{w_{i,k}}$, with $z_{i,k}$ and $w_{i,k}$ as variables over $\F_2$ and $\Q^{>0}$ (the multiplicative group of positive rationals) respectively. We thus have to solve the following systems
\[ \sum_{i=1}^{4} z_{i,\sigma_j(i)} = e_j  \ \ \text{ over } \F_2, \text{ and } \prod_{i=1}^{4} w_{i,\sigma_j(i)} = r_j  \ \ \text{ over } \Q^{>0}.\]
Like in the case for general $n$, we can assume that $z_{1,2},z_{1,3},z_{1,4},z_{2,1},z_{2,3},z_{2,4}$ and $z_{3,1}$ are $0$ by Proposition \ref{prop-lie-descript-char2} and we can solve the system over $\F_2$ separately by adapting Algorithm \ref{alg3}. We now solve the system in $w_{i,k}$ variables, which we can express as follows
\begin{equation} 
\begin{blockarray}{cccccccccccc}
        \begin{block}{(ccc|ccc|ccc|ccc@{\hspace*{5pt}})}
             1 & 0 & 0 & 0 & 1 & 0 & 0 & 0 & 1 & 1 & 0 & 0  \\
             1 & 0 & 0 & 0 & 0 & 1 & 1 & 0 & 0 & 0 & 0 & 1  \\
              0 & 1 & 0 & 0 & 0 & 1 & 0 & 1 & 0 & 1 & 0 & 0  \\
              0 & 1 & 0 & 1 & 0 & 0 & 0 & 0 & 1 & 0 & 1 & 0  \\
              0 & 0 & 1 & 0 & 1 & 0 & 1 & 0 & 0 & 0 & 1 & 0  \\
             0 & 0 & 1 & 1 & 0 & 0 & 0 & 1 & 0 & 0 & 0 & 1  \\
        \end{block}
    \end{blockarray}  
    \cdot \vecw
     =  \begin{blockarray}{c}
          \begin{block}{(c)}
              r_1 \\
              r_2\\
              r_3\\
              r_4\\
              r_5\\
              r_6\\
          \end{block}
      \end{blockarray}
\end{equation}
where $\vecw$ is a vector of variables over $(\Q^{>})^{6}$, addition and subtraction correspond to multiplication and division in $\Q^{>}$, scaling by positive integer to raising to integral powers, scaling by $-1$ to taking the inverse of an element in $\Q^{>}$ and scaling by $0$ means replacing the element by $1$. Then, as in the proof of Lemma \ref{lemma-HC4-discretesym-fq}, Gaussian elimination gives
\begin{equation} \label{eqn-HC4-scalerecov-1}
    \begin{blockarray}{ccccccccccccc}
        \begin{block}{c(ccc|ccc|ccc|ccc@{\hspace*{5pt}})}
            & 1 & 0 & 0 & 0 & 1 & 0 & 0 & 0 & 1 & 1 & 0 & 0  \\
             & 0 & 0 & 0 & 0 & -1 & 1 & 1 & 0 & -1 & -1 & 0 & 1  \\
            & 0 & 1 & 0 & 0 & 0 & 1 & 0 & 1 & 0 & 1 & 0 & 0  \\
             & 0 & 0 & 0 & 1 & 0 & -1 & 0 & -1 & 1 & -1 & 1 & 0  \\
             & 0 & 0 & 1 & 0 & 1 & 0 & 1 & 0 & 0 & 0 & 1 & 0  \\
             & 0 & 0 & 0 & 0 & 0 & 0 & -2 & 2 & 0 & 2 & -2 & 0  \\
        \end{block}
    \end{blockarray} 
    \cdot \vecw
     =  \begin{blockarray}{c}
          \begin{block}{(c)}
              r_1 \\
              \frac{r_2}{r_1}\\
              r_3\\
              \frac{r_4}{r_3}\\
              r_5\\
              \frac{r_6r_3r_1}{r_5r_4r_2}\\
          \end{block}
      \end{blockarray}
\end{equation}
It is verifiable that $\frac{r_6r_3r_1}{r_5r_4r_2}$ is a square using $r_j = \prod_{i=1}^{4} S_{i,\sigma_j(i)}$. 

A solution to \eqref{eqn-HC4-scalerecov-2} is also a solution to \eqref{eqn-HC4-scalerecov-1}. As $r_i$'s are assumed to be given in the form $\frac{a_i}{b_i}$, we can compute the positive square root of $\frac{r_6r_3r_1}{r_5r_4r_2}$ by computing its numerator and denominator, and then using binary search to compute the positive square root of the numerator and the denominator.

\begin{equation} \label{eqn-HC4-scalerecov-2}
    \begin{blockarray}{ccccccccccccc}
        \begin{block}{c(ccc|ccc|ccc|ccc@{\hspace*{5pt}})}
            & 1 & 0 & 0 & 0 & 1 & 0 & 0 & 0 & 1 & 1 & 0 & 0  \\
             & 0 & 0 & 0 & 0 & -1 & 1 & 1 & 0 & -1 & -1 & 0 & 1  \\
            & 0 & 1 & 0 & 0 & 0 & 1 & 0 & 1 & 0 & 1 & 0 & 0  \\
             & 0 & 0 & 0 & 1 & 0 & -1 & 0 & -1 & 1 & -1 & 1 & 0  \\
             & 0 & 0 & 1 & 0 & 1 & 0 & 1 & 0 & 0 & 0 & 1 & 0  \\
             & 0 & 0 & 0 & 0 & 0 & 0 & -1 & 1 & 0 & 1 & -1 & 0  \\
        \end{block}
    \end{blockarray} 
    \cdot \vecw
     =  \begin{blockarray}{c}
          \begin{block}{(c)}
              r_1 \\
              \frac{r_2}{r_1}\\
              r_3\\
              \frac{r_4}{r_3}\\
              r_5\\
              \sqrt{\frac{r_6r_3r_1}{r_5r_4r_2}}\\
          \end{block}
      \end{blockarray}
\end{equation}
We can proceed further with Gaussian elimination to get

\begin{equation}  \label{eqn-HC4-scalerecov-3}                  
    \begin{blockarray}{cccccccccccc}
        \begin{block}{(ccc|ccc|ccc|ccc@{\hspace*{5pt}})}
             1 & 0 & 0 & 0 & 0 & 1 & 0 & 1 & 0 & 1 & -1 & 1  \\
              0 & 0 & 0 & 0 & -1 & 1 & 0 & 1 & -1 & 0 & -1 & 1  \\
             0 & 1 & 0 & 0 & 0 & 1 & 0 & 1 & 0 & 1 & 0 & 0  \\
              0 & 0 & 0 & 1 & 0 & -1 & 0 & -1 & 1 & -1 & 1 & 0  \\
              0 & 0 & 1 & 0 & 0 & 1 & 0 & 2 & -1 & 1 & -1 & 1  \\
              0 & 0 & 0 & 0 & 0 & 0 & -1 & 1 & 0 & 1 & -1 & 0  \\
        \end{block}
    \end{blockarray} 
        \cdot \vecw =  \begin{blockarray}{c}
          \begin{block}{(c)}
            r'_1 \\
            r'_2  \\
            r_3\\
            r'_4\\
            r'_5 \\
            r'_6\\
          \end{block}
      \end{blockarray}
\end{equation}
where $r'_1 = \sqrt{\frac{r_6r_3r_1r_2}{r_5r_4}}$, $r'_2 = \sqrt{\frac{r_6r_3r_2}{r_5r_4r_1}}$, $r'_4 = \frac{r_4}{r_3}$, $r'_5 = \frac{r_6r_3}{r_4}$, and $r'_6 = \sqrt{\frac{r_6r_3r_1}{r_5r_4r_2}}$. It can again be verified that $\frac{r_6r_3r_1r_2}{r_5r_4}$ and $\frac{r_6r_3r_2}{r_5r_4r_1}$ are both squares, and hence we can compute their positive square roots. We can then set $w_{i,j}$'s as
\begin{equation}
\begin{split}
  w_{1,2} &= \sqrt{\frac{r_6r_3r_1r_2}{r_5r_4}} \quad w_{2,1} = \frac{r_4}{r_3} \\
  w_{1,3} &= r_3 \quad w_{2,3} = \sqrt{\frac{r_5r_4r_1}{r_6r_3r_2}}\\
  w_{1,4} &= r_6\frac{r_3}{r_4} \quad w_{3,1} = \sqrt{\frac{r_5r_4r_2}{r_6r_3r_1}}\\
\end{split}    
\end{equation}
and the remaining $w_{i,j}$'s to $1$. Thus, we get $S'_{i,j} = (-1)^{z_{i,j}}w_{i,j} \in \Q^{\times}$. It is then easily verifiable that $HC_4(S'\vecx) = f(\vecx)$. Over $\R$, we can set $S'_{i,j}$'s in the same way, assuming the model of computation allows us to compute square roots exactly. Over $\F_q$, we do not need to consider signs separately and can represent the equations $\prod_{i=1}^{4} S_{i,\sigma_j(i)} = c_j$ as the system of equations in \eqref{eqn-HC4-scalerecov-1} and solve this system of equations as before using square root computations wherever necessary.

\end{document}